\documentclass[12pt]{article}

\usepackage{mathpazo} 

\usepackage[utf8]{inputenc}

\usepackage{amssymb}
\usepackage{amsmath}
\usepackage{amsthm}
\usepackage{tikz}
\usepackage{booktabs}
\usepackage{cite} 
\usepackage{multirow}
\usepackage{eso-pic}
\usepackage[round]{natbib}   
\usepackage{tikz}
\usetikzlibrary{decorations, shapes, arrows.meta, positioning,calc,arrows}
\usepackage{bigintcalc}
\usepackage{enumitem}
\usepackage{comment}
\usepackage{bbm}
\usepackage{hyperref,float}
\usepackage{cleveref}
\usepackage{caption}
\usepackage{adjustbox}
\usepackage{environ,etoolbox}
\usepackage{booktabs}
\usepackage{chngcntr}
\usepackage{apptools}
\AtAppendix{\counterwithin{lem}{section}}

\newtheorem{lem}{Lemma}
\newtheorem{cor}{Corollary}
\newtheorem{prop}{Proposition}
\newtheorem*{prop*}{Proposition}

\newtheorem*{thm*}{Theorem}

\newtheorem{nota}{Notation}

\definecolor{nberblue}{RGB}{0, 90, 155}
\definecolor{resgreen}{RGB}{34, 84, 67}
\definecolor{pennblue}{RGB}{0, 44, 119}
\definecolor{pennred}{RGB}{152, 30, 50}
\definecolor{aerred}{RGB}{134, 27, 57}
\definecolor{esblue}{RGB}{59, 99, 143}
\definecolor{uncblue}{RGB}{75, 156, 211}
\definecolor{zjublue}{RGB}{0,63,136}
\definecolor{dblue}{RGB}{3,32,96}
\definecolor{dred}{RGB}{192,0,0}
\colorlet{shadecolor}{gray!40}

\hypersetup{
	colorlinks=true,
	linkcolor=nberblue,
	urlcolor=nberblue,
	citecolor=nberblue 
}

\begin{document}

\title{Data Trade and Consumer Privacy}
\author{Jiadong Gu\thanks{School of Economics, Zhejiang University. Email: \href{mailto:jgu123@zju.edu.cn}{jgu123@zju.edu.cn}. 
		I am deeply indebted to my advisors Fei Li, Peter Norman, Gary Biglaiser, and Andrew Yates for their invaluable guidance and continued support. I especially thank Alessandro Pavan, my mentor of the ES Virtual Mentorship Initiative, for many extremely insightful comments and helpful suggestions and the ES mentoring program 2025. 
		I am grateful to the participants of the UNC Micro Theory and Industrial Organization workshops, Triangle Microeconomics Conference 2022, ACDE 2022, 
		the 7th CCER Summer Institute, APIOC 2023, ESWC 2025, ES Virtual Mentorship Meeting 2025, HKUST (Guangzhou), BNU, CUFE, Fudan University, Jinan University, Shandong University, UIBE, Zhejiang University seminars for their helpful comments. All errors are my own. 
} }

\date{\today}

\maketitle

\begin{abstract}
	This paper studies optimal mechanisms for collecting and trading data. Consumers benefit from revealing information about their tastes to a service provider because this improves the service. However, the information is also valuable to a third party as it may extract more revenue from the consumer in another market called the product market. The paper characterizes the constrained optimal mechanism for the service provider subject to incentive feasibility. It is shown that although the service provider sometimes sells no information or only partial information in order to preserve profits in the service market, selling full information is optimal when the data-sourcing market is highly differentiated.
    Moreover, a ban on data trade may reduce social welfare because it makes it harder to price discriminate in the product market. Instead, reducing the intermediary's bargaining power can protect privacy without hurting social welfare, which suggests that the regulation of market power is more efficient than the regulation of data sharing. 
	
	\vskip .5in
	
	\vskip .1in
	
	
	\thispagestyle{empty}
\end{abstract}


\newpage \baselineskip 18pt

\pagenumbering{arabic}

\section{Introduction}

The digital economy is deeply rooted in the sharing and use of personal data. Meanwhile, the transmission of personal data {\it across markets} brings concerns about privacy,\footnote{For instance, personal data of Facebook users was collected by British consulting firm Cambridge Analytica in the 2010s. The privacy issues attract attention from both the public and regulatory authorities. The General Data Protection Regulation (GDPR) was put into effect in Europe in 2018, the California Consumer Privacy Act (CCPA) went into effect in 2020 and China’s Personal Information Protection Law was in force in late 2021.} which can bite: The privacy concerns induced by data transmission, such as data trade, can adversely affect the data generating activities (e.g., data sharing, service usage) in related markets. The special feature of data endows it with a {\it cross-market} role. 

Data is usually collected in one market and is also valuable in other markets. For example, Amazon Music, a music service provider, may collect taste data when users are enjoying the service of music recommendations. Amazon Music can trade the user data to a headphone seller (firm, producer, or data buyer), and the latter can use the data to price against users' interests. As consumers have become increasingly {\it aware} of how their data can be collected, traded, and used, their willingness to use Amazon Music service and to share their data will be affected. With the cross-market role of data trade, the intermediary (service provider or platform) would design mechanisms for both collecting and trading data. What are the optimal mechanisms? How do the data trade and its cross-market effect shape the equilibrium privacy, surplus allocation, and social welfare? Does the regulation of data sharing work better than the regulation of market power?

To address these questions, this paper develops an integrated model including three sorts of interactions: (1) the intermediary’s data collection from consumers by providing service in the service market; (2) data trade between the intermediary and a firm in the data market; (3) pricing game between the firm and consumers in the product market. \Cref{fig_overview} illustrates an overview of these interactions.
We use the framework to study the intermediary optimal mechanisms of data collection and trade, and data market regulations. 

\begin{figure}[h]
	\centering
	\begin{tikzpicture}
		[ node distance=4ex and 0em,
		block/.style={circle, draw, fill=nberblue!20, text width=5.5em, text centered, rounded corners, minimum height=5em},
		line/.style={draw,-latex} ]
		
		\node [block,line width=1pt] (0) {intermediary};
		\node [block,line width=1pt, below left = 5em of 0] (1) {consumers};
		\node [block,line width=1pt, below right = 5em of 0] (2) {firm};
		
		\path[-stealth,line width=2pt,transparent!18]	
		([yshift=10pt]0.east) edge ([xshift=10pt]2.north) node [above = 2.5em of 2,black] {data trade}
		([xshift=-25pt]2.north) edge ([yshift=-25pt]0.east) node [below right= 1em of 0,black] {payment}
		
		([xshift=-10pt]1.north) edge ([yshift=10pt]0.west) node [above = 2.5em of 1,black] {data collection}
		([xshift=-10pt]1.north) edge ([yshift=10pt]0.west) node [above = 1.25em of 1,black] {payment}
		([yshift=-25pt]0.west) edge ([xshift=25pt]1.north) node [below left= 1em of 0,black] {service}
		
		([yshift=0pt]1.east) edge ([yshift=0pt]2.west) node [below = 3.5em of 0,black] {payment}
		([yshift=-25pt]2.west) edge ([yshift=-25pt]1.east) node [below = 5.5em of 0,black] {product};
		
	\end{tikzpicture}
	\caption{The overview of the modeling ingredients and interactions among players.}
	\label{fig_overview}
\end{figure}
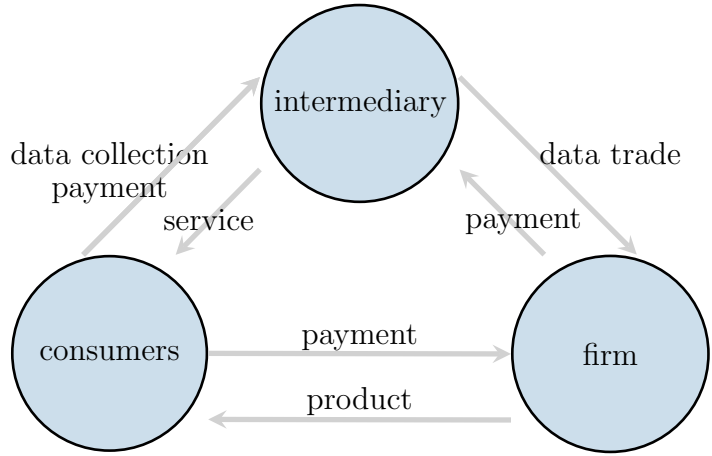

Our framework enables formal investigations on data sourcing and trade in a comprehensive way, and hence the equilibrium interactions between second degree price discrimination in the data-sourcing market and third degree discrimination endogenous to data trade. Many works have separately investigated some specific elements, like information collection by screening \citep{mussa1978monopoly}, information sales by certification \citep{lizzeri1999information}, or data brokers' information selling mechanism \citep{yang2022selling}, and less has offered a comprehensive investigation on data collection, data trade, and data usage.
The holistic study of the intermediary's collection, trade, and usage of data and their effects and regulatory implications cannot be easily discussed in divided models. 
This framework contributes to the advocacy of integrated models to comprehensively understand the information markets \citep{bergemann2019markets} with the seminal envision back to \citet{marschak1968economics}. 

The framework also offers to study the intermediary's bargaining power. The distribution of bargaining power is important when it comes to the transfer of data across markets via intermediation, because transferring the data from consumers to the intermediary does not necessarily transfer the bargaining power to the intermediary. Although a large literature has looked at the comparative statics across the information available to the market participants,\footnote{E.g., the information available to the seller \citep{bergemann2015limits} or to the buyer \citep{roesler2017buyer} or to both \citep{kartik2023lemonade}.} the information intermediary's bargaining power and its regulation implications are under-explored.\footnote{See discussions about the bargaining power of a durable good reseller in \citet{calzolari2006monopoly}. We discuss the literature more fully below.
}

In the framework of~\Cref{fig_overview}, consumers who have private information about horizontal characteristics (tastes) benefit from sharing information with the intermediary by receiving a better-matched service but at the risk of the shared information being sold to a third party, the producer who can price against consumers using the bought data.
In detail, consumers' tastes are $H$ or $L$ (e.g., classical or rock music fans). Each consumer’s utility is from consuming the service and product. The service provided to consumers may generate a mismatch as the intermediary does not know consumers’ true tastes. Each consumer also has unit demand for the product. The producer, a monopolist in the product market, does not know consumers’ willingness to pay for the product, which is related to consumers’ tastes. Hence, the taste data collected by the intermediary is valuable for the producer to set product prices. The intermediary can sell consumer data to the producer and it is the unique data source. The intermediary specifies (and commits to) {\it a contract} that consists of service provision, service fee, privacy policy (data trade), and data fee to consumers and the producer {\it simultaneously}. Consumers and the producer decide whether to accept. 
Understanding that the data trade may discourage the service usage and data revelation of consumers, the intermediary’s contract would trade off between benefits of trading data and costs of hindering consumers’ data sharing and service usage.\footnote{The interactions in~\Cref{fig_overview} fit into many market scenarios such as Borrower-Fintech-Lender, in which consumers reveal their personal data while using the financial service provided by banks. Banks may transfer data via Application Programming Interface, i.e., open banking, to third-party fintech firms, which may offer personalized interest rates to consumers as borrowers in the lending market.
	In the example of Amazon Music, although how much a consumer listens to music could also be indicative of willingness to pay for headphones, the music service provider usually can verify how much a consumer listens to from the history, which is more in the spirit of the behavior-based price discrimination (see, e.g., \citet{fudenberg2006behavior}).
}

We characterize the intermediary constrained optimal mechanisms. We first consider, given an intermediary's contract, the decisions of the producer and consumers which impose participation and incentive constraints on the intermediary's design problem. First, the producer buys data for a binary pricing decision, the binary-signal data works as product price recommendations subject to the {\it obedience constraints} \citep{bergemann2016bayes}. That is, the producer follows price recommendations of the traded data. Thus, data trade affects the product price and hence consumers surplus. The intermediary’s data offer takes a leave-it-or-take-it form. The payment of data trade makes the producer willing to accept the intermediary’s data offer. 
Second, we turn to consumers' decisions. As the privacy policy affects consumers surplus in the product market, consumers’ incentive to reveal their types in the service market is then affected by the privacy policy. The choices of privacy policy and service make consumers willing to accept the offer and also to reveal their types.

One key feature of the intermediary optimal mechanisms is that the optimal informativeness of data trade depends on comparisons of the {\it horizontal} differentiation level in the data-sourcing market (captured by the variation in service tastes) and {\it vertical} differentiation level in the data-using market (captured by the variation in willingness to pay for the product). The optimal data trade sells {\it no} or {\it partial} information when the horizontal differentiation level is sufficiently low relative to the vertical one.\footnote{Although the intermediary protects privacy in \citet{calzolari2006optimality}, the setups and driving forces are different, as discussed in the literature below.}
In this case, as the service mismatch is small, consumers have less incentive to share their data to get the right service. The intermediary commits to selling less data to maintain consumers' incentives to reveal information. Otherwise, suppose that the intermediary sells full information to the producer. Although the intermediary obtains most consumers surplus in the product market, it has to incentivize consumers' information revelation by reducing service fees and hence loses the service profit. When the horizontal differentiation level is relatively low, the service profit losses outweigh the extra surplus gain in the product market. Hence, the intermediary would prefer to sell less data to preserve the service profits. In contrast, when the horizontal differentiation level is high enough, the intermediary can sell {\bf full} information. The service profit losses due to incentive distortions are small even if the intermediary sells full information. 


The second remarkable result is that introducing the intermediary's partial bargaining power in the data market can decrease the amount of information disclosed by the optimal mechanism, but does not affect the surpluses of the intermediary, consumers, and producer. First, decreasing informativeness has two effects: the service profit gain and the data profit loss. With the intermediary's full bargaining power, these two effects cancel each other out, leading to multiple solutions of disclosure policy; with partial bargaining power, the service profit gain outweighs the data profit loss, which drives the multiple solutions of data trade to the unique one, i.e., the minimum information among them.
Second, the partial bargaining power over data profit does not affect the intermediary's payoff because the equilibrium data profit is zero, while the consumer and producer get their outside option.

Using the framework, we study regulations on the data market. In particular, we focus on regulations of the intermediary's data sharing and market power. First, our results show that the data trade ban does {\it not improve} social welfare. This happens because without any data it is difficult for the producer to price discriminate among consumers. 
Under the optimal data trade, the producer can extract by price discrimination the surpluses of all the consumers, some of whom are otherwise excluded from the market.
Second, regulating the intermediary's power (how much the intermediary can extract the data revenue) will protect consumer privacy without hurting social welfare, because reducing the intermediary's bargaining power doesn't decrease the surpluses nor total welfare. 
The findings suggest that the regulation of the intermediary's market power would be more efficient than the regulation of the intermediary's data sharing.

\paragraph{Related Literature.} The paper relates to several strands of literature: information market, the economics of privacy, price discrimination, information design, and mechanism design.

First, our paper adds to the literature on information markets (see recent surveys in \citet{bergemann2019markets} and \citet{bergemann2021information}). \citet{admati1990direct} focus on a monopolist that sells information about an asset, and \citet{lizzeri1999information} focuses on an informed seller who pays for the intermediary's certification to signal quality. \citet{bergemann2015selling} study a data broker who sells cookies for targeting markets; \citet{babaioff2012optimal} and \citet{bergemann2018design} study optimal mechanisms for selling information in different contracting environments. \citet{ichihashi2021competing} and \citet{bonatti2024selling} respectively investigate the competition among data intermediaries and data buyers. \cite{park2025selling} study the intermediary's incentives of information design to influence the consumer's willingness to pay for information and the seller's pricing decision, with a focus on the consumer's purchase of information. The majority of literature treats data sourcing as given; few of them investigate the intermediary's collection and sale of information.\footnote{ Consumer behaviors, especially as main data sources, can be distorted by the data trade. \citet{argenziano2023datam} analyze strategic behaviors of privacy-conscious consumers. Several works show that the consumers’ manipulation incentives reduce the amount of information transmission in equilibrium, e.g., \citet{frankel2022improving}, \citet{bonatti2020consumer}.} 

One notable exception is \citet{calzolari2006optimality}, that study disclosure in a sequential contracting setting and find that the upstream principal offers full privacy for the agent. Our simultaneous contracting setting admits selling information to be optimal. Their full privacy result requires the intermediary's preference to have the same sign of the single-crossing condition in terms of trade or not with the third party, but the preference structure has different signs in our paper. The insight in the simultaneous setting is driven by the interaction between third degree price discrimination in the data-using market and second degree price discrimination in the data-sourcing market. Moreover, \citet{calzolari2006optimality} do not consider the bargaining power of the intermediary, and \citet{calzolari2006monopoly} discuss the bargaining power but don't study its regulatory implications.
Our paper highlights both the data sourcing and the intermediary's bargaining power.

Second, the paper is related to the economics of privacy with seminal works of \citet{taylor2004consumer} and \citet{calzolari2006optimality}; see \citet{acquisti2016economics} and \citet{goldfarb2023economics} for recent surveys. The privacy concerns arise from sellers’ learning about consumers’ type \citep{villas2004price,conitzer2012hide} and buyers’ information revelation to sellers \citep{hidir2021privacy,ichihashi2020online}. \citet{argenziano2023datal} study the data linkages instead of the intermediary’s data selling. \citet{ali2023voluntary} combine the consumer information disclosure and personalized pricing. In dynamic environments, the privacy concerns on discrimination induce a ratchet effect \citep{laffont1988dynamics}. The ratchet effect can bite all the information value as the data collecting intermediary gets better off committing to offering full privacy \citep{calzolari2006optimality}. \citet{doval2025purchase} find in a limited commitment environment that the price discrimination induces a product line gap due to this effect. \citet{rhodes2021personalized} study consumer privacy choice in the context of personalized pricing. Our paper extends the literature by endogenizing the privacy policy as a part of the intermediary's design of information collection and sale.

Several other works share a similar feature of tradeoffs between data exploitation and user privacy. \citet{fainmesser2023digital} study the tradeoff between data collection and data protection but without a role of intermediary's data trade, and hence the privacy is not endogenous to the data trade.\footnote{They assume that the personal information (consumer activity such as clicks) is verifiable. The privacy level is the protection investment. \citet{argenziano2023datal} micro founds the privacy preferences. In our paper consumer privacy is based on the equilibrium data sale.} \citet{jullien2020privacy} consider the tradeoff between data exploitation and user activity. \citet{dosis2023ownership} focus on a tradeoff between data monetization and processing. In these works, consumer information is verifiable. Our paper adopts the mechanism design approach because consumers control the information and information is {\it unverifiable}, which is more about consumers' private characteristics than their verifiable activities, like browsing data. 
Related to the cross-market role of data, \citet{doval2025purchase} study the {\it inter-temporal} role of data. They focus on the tradeoff between product personalization and price discrimination in a dynamic model. 
Rather, our paper investigates in a static environment the intermediary's balance between the data sourcing and data sales and how the intermediary's bargaining power affects the balance.

Third, this paper also adds to the literature on price discrimination and market segmentation by emphasizing the endogenous interactions between the second degree and third degree price discrimination. 
\citet{bergemann2015limits} study how market segmentation can support an arbitrary surplus distribution between consumers and a monopolist.\footnote{\citet{haghpanah2022limits} identify markets for which the entire surplus triangle of \citet{bergemann2015limits} is achievable and \citet{haghpanah2023pareto} study Pareto improving market segmentation with multiple products.} The endogenous market segment also arises from consumers’ equilibrium disclosure about their characteristics in \citet{ichihashi2020online} and \citet{ali2023voluntary} and data brokers' creation and sale of market segments for profits in \citet{yang2022selling}.
So our framework contributes to discussing the relation between the third degree price discrimination by market segmentation and the second degree discrimination induced by the monopoly screening.

Lastly, it is related to the broad literature of information design \citep{kamenica2011bayesian,bergemann2019information} and mechanism design.\footnote{See, e.g., \citet{mussa1978monopoly}, \citet{myerson1981optimal} for classic works and \citet{bergemann2006information} for a survey on information in mechanism design.} Several works study the comparative statics of information available to the seller \citep{bergemann2015limits} or to the buyer \citep{roesler2017buyer} or to both \citep{kartik2023lemonade}.
\citet{bergemann2018design} study the mechanism of selling information. \citet{yang2022selling} studies a setting where the data buyer has private information about costs. \citet{bonatti2024selling} and \citet{doval2025purchase} investigate mechanism design problems with multiple buyers and multiple periods, respectively. They don't consider the source of information.
In our paper, consumers, as the data sources, have private unverifiable information. Hence, we rely on the mechanism design approach to incentivize consumers' information revelation, and the privacy policy or data trade is one of the intermediary's design elements.
The study of the interaction between monopoly screening and endogenous market segmentation complements the standard monopoly screening \citep{mussa1978monopoly,maskin1984monopoly}, which does not consider the information disclosure as part of allocations.
\citet{calzolari2006optimality} study screening plus persuasion and \citet{dworczak2020mechanism} confines to cut-off mechanism in sequential contracting. This paper differs in the simultaneous setup and complements the literature by looking at the critical role of the intermediary's bargaining power.




\paragraph{Outline.} The rest of this paper is organized as follows. \Cref{model} sets up the model. \Cref{preliminaries} takes some preliminary analyses of individual decisions. \Cref{optimal_design} investigates the intermediary optimal design: \Cref{binary_service} focuses on the role of data trade, \Cref{general_service} completely characterizes the intermediary optimal mechanism and discusses the regulation of data sharing. \Cref{bargaining_power} studies the intermediary's bargaining power in the data market and the regulation of the intermediary's power. \Cref{discussion} provides some discussions and \Cref{conclusion} concludes. Omitted proofs and additional results are relegated to Appendices.

\section{The Model }\label{model}
There are three types of players: one intermediary (it), one producer (he), and one unit mass of consumers (she). Consumers, the intermediary, and the producer interact in three markets: a service market, a product market, and a data market. In the service market, the intermediary collects data about consumers' type by providing consumers service. The data can be traded in the data market from the intermediary to the producer, who sells products to consumers in the product market.

Each consumer has a type $\theta \in \left\{L,H\right\}$ with $L, H\in\mathbb R$. The type is a horizontal taste parameter.\footnote{In this sense notations $H$ and $L$ have no essential meanings. One can relabel them as ``A"``B", e.t.c.} Consumers’ type is realized from a prior distribution $\mu_0=\text{Pr}\left(\theta = H\right)\in [0,1]$, which is common knowledge. Each consumer knows her type, but neither the producer nor the intermediary observes consumers' true types.

Consumers’ utility relies on the service and the product. A consumer of type $\theta$ who consumes both the service and product will get
\begin{equation*}
	V - \underbrace{ \left( x - \theta \right)^2 }_{\text{ mismatch in service} } - f
	+ \underbrace{ \max \left\{ v_\theta  - p,0 \right\} }_{\text{product market} } 
\end{equation*}
which consists of two parts. First, the recommended service by the intermediary generates a baseline utility level $V\in \mathbb{R}$, but there can be recommendation bias or mismatch between a recommendation $x$ and consumer true taste $\theta$, which is measured by the squared item.\footnote{The squared form specifies only the mismatch formula in the equilibrium.} Consumers need to pay the service fee $f$ to get service.\footnote{The fee $f$ can be negative in the setup, which means that the intermediary pays the consumers. As we will see later the optimal $f$ will be always positive if $V$ is sufficiently large.} 

Second, each consumer has unit demand for the product and decides whether to purchase the product or not. Consumers of $\theta$ purchase the product if and only if $v_\theta - p \geq 0$ where $p$ denotes the product price and $v_\theta\in\mathbb{R}$ is the willingness to pay for the product of type $\theta$ consumers. We assume $v_L\neq v_H$ so that the horizontal type is informative about the willingness to pay.\footnote{In contrast, if $v_L=v_H$ then data on consumers' type is not valuable to the producer.} It is without loss to assume $v_L<v_H$.

The producer is a monopolist in the product market. As the product demand relies on the producer's information about consumers' willingness to pay, the product revenue is
\begin{equation*}
	p\times \text{Pr}\left(v\geq p \vert \mathcal I \right)
\end{equation*} 
where $\text{Pr}\left(v\geq p \vert \mathcal I \right)$ is the total mass of consumers who purchase the product at a price $p$ conditional on the producer's information about consumer willingness to pay, $\mathcal I$. 

The intermediary plays in both service and data markets by offering a contract of data collection and data trade. In the service market the intermediary is a monopolist who provides consumers with services to collect their taste data. The intermediary sells data via a data offer and the producer buys data for product pricing decisions. The intermediary is the unique data source. Its profit comes from service provision and data sale.

By the revelation principle \citep{myerson1979incentive}, the intermediary designs a direct mechanism which specifies: (1) a service menu for consumers including the service provision $x:\{L,H\} \to \mathcal X$ where $\mathcal X$ is the set of service; and the service fee $f:\{L,H\} \to \mathbb{R}$; and (2) a data offer for the producer consisting of the disclosure policy captured by a pair of distributions over disclosed signals about consumers types $\pi:\left\{L,H\right\}\to\Delta \mathcal S$ where $\mathcal S$ is the set of possible signals; and the lump-sum data fee $T\in \mathbb{R}$. Consumers {\it observe all} the specifications. If a consumer reports $\theta'$, then a signal will realize according to $\pi\left(s\vert \theta'\right)$ in the product market. In sum, the intermediary offers a contract of data collection and data trade of which $\left(\mathcal X, x, f\right)$ and $\left(\left\{\mathcal S,\pi\right\},T\right)$ specify the part of service market and data market respectively.

\paragraph{Timing.} Consumers privately observe type $\theta$. The intermediary posts the entire contract $\{(\mathcal X, x, f), (\{\mathcal S,\pi\},T ) \}$ including the service menu in the service market and data offer in the data market. Then consumers and the producer make acceptance decisions. If consumers accept, they report type $\theta'$ and then get service $x_{\theta'}$ and pay $f_{\theta'}$. If consumers reject, they get the outside option in the service market. If the producer accepts, after data trade, in the product market a signal $s$ about consumer's type is realized according to $\pi$. The producer observes $s$ and then sets the product price. Consumers decide whether to purchase the product or not. If the producer rejects the data offer, then the product pricing is made under the prior information $\mu_0$. The timing of events is summarized in \Cref{timing}.

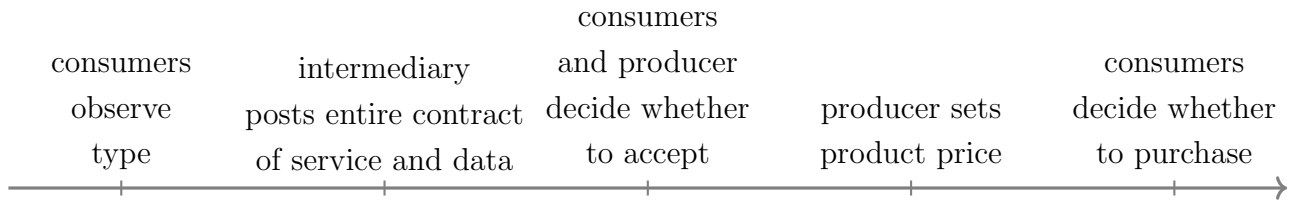
\begin{figure}[h]
	\centering
	\begin{tikzpicture}[xscale=1]
		\draw [very thick,gray] [->] (-1.5,0) -- (15.5,0);
		\draw [thick,gray] (0,-.1) -- (0,.1);
		\draw [thick,gray] (3.5,-.1) -- (3.5, .1);
		\draw [thick,gray] (7,-.1) -- (7, .1);
		\draw [thick,gray] (10.5,-.1) -- (10.5, .1);
		\draw [thick,gray] (14,-.1) -- (14, .1);
		
		\node[align=center, above]   at (0,0.1) {consumers \\ observe \\ type};
		\node[align=center, above]   at (3.5,0.1) {intermediary \\ posts entire contract \\ of service and data};
		\node[align=center, above]   at (7,0.1) {consumers \\ and producer \\ decide whether \\ to accept};
		\node[align=center, above] at (10.5,0.1) {producer sets \\ product price};
		\node[align=center, above]  at (14,0.1) {consumers \\ decide whether \\ to purchase};

	\end{tikzpicture}
	\caption{Timing of the events.}
	\label{timing}
\end{figure}

\paragraph{Discussion about Modeling Assumption.} A key modeling assumption is the perfect correlation between consumers’ taste and their willingness to pay. That is to say, $v_L\neq v_H$ in the setup. Otherwise, the data about consumers' tastes is not valuable to the producer. The relation between consumers’ taste parameter and their willingness to pay can be represented by $\rho_\theta = \text{Pr}\left(v_H\vert \theta\right)=1-\text{Pr}\left(v_L\vert \theta\right)$. The assumption reads $(\rho_H=1$, $\rho_L=0)$ or $(\rho_H=0$, $\rho_L=1)$. If we relax the assumption by $\rho_H$, $\rho_L\in\left(0,1\right)$ with $\rho_H\neq \rho_L$, the privacy results don't change.\footnote{We describe in the supplementary online appendix how our results changed/remained unchanged as we move toward the imperfect correlation to consider heterogeneity in valuations conditional on type.} Discussions on other assumptions are relegated to~\Cref{discussion}.

\section{Preliminaries} \label{preliminaries}

The intermediary's total payoff consists of profit in the service market $\mu_0 f_H + \left(1-\mu_0\right)f_L$ and profit in the data market $T$. That is,
\begin{equation}
	\mu_0 f_H + \left(1-\mu_0\right)f_L + T. \label{payoff_intermediary}
\end{equation}

Before digging into the intermediary optimal mechanism, we first consider the decisions of consumers and the producer {\it after} observing the intermediary's contract for both service and data markets and {\it before} observing a signal realization.

\paragraph{Interim Payoff of Producer.} In the product market, a consumer of $\theta$ will purchase the product if and only if $v_\theta\geq p$ at the product price $p$. After observing a signal $s\in\mathcal S$ from the traded data $\pi$, the producer's posterior belief is $\mu_s\equiv\text{Pr}(\theta =H \vert s)$ or $\text{Pr}(v =v_H \vert s)$
\begin{equation}
	\mu_s = \frac{\mu_0 \pi\left(s\vert H\right) }{\mu_0 \pi\left(s\vert H\right) +\left(1-\mu_0\right)\pi\left(s\vert L \right)} \label{posteriors}
\end{equation}
Hence, for each $s\in \mathcal S$ the producer's demand function reads
\begin{equation*}
	D\left(p;s \right) \equiv \text{Pr}\left(v\geq p\vert s \right)
	=\left\{
	\begin{array}{cc}
		0      &\text{if } p > v_H \\
		\mu_s  &\text{if } v_L< p \leq v_H \\
		1      &\text{if } p \leq v_L 
	\end{array} 
	\right. 
\end{equation*}
So the producer's product pricing decision after observing $s$ is
\begin{equation}
	p_s
	\in \left\{
	\begin{array}{cc}
		\{v_H\}      &\text{if } \mu_s > \frac{v_L}{v_H} \\
		\{v_L, v_H \}  &\text{if }  \mu_s = \frac{v_L}{v_H} \label{prod_price}\\
		\{ v_L\}      &\text{if } \mu_s < \frac{v_L}{v_H}
	\end{array} 
	\right. 
\end{equation}
and the corresponding product revenue is $\max\{ v_H\mu_s, v_L \}$.

At the time of data trade, the producer's expected product revenue from accepting the data trade offer $(\{\mathcal S,\pi \},T )$ is
\begin{equation}
	R( \pi ) = \mathbb{E}_{\pi}\left[ \max \left\{ v_H\mu_s, v_L \right\} \right] \label{prod_reve}
\end{equation}
The producer’s payoff is $R(\pi)-T$ if accepting the data offer.

\paragraph{Interim Payoff of Consumers.} Given the intermediary's contract, consumers of type $\theta$'s total interim payoff from accepting the offer and reporting $\theta'$ is
\begin{eqnarray}
	u ( \theta,\theta';\pi ) = V - \left( x_{\theta'} - \theta \right)^2 - f_{\theta'}
	+ \mathbb{E}_{\pi\left( s \vert \theta'\right) } \left[ \max \left\{ v_\theta  - p_s,0 \right\} \right] \label{payoff_consumer}
\end{eqnarray}
where if a consumer of $\theta$ reports $\theta'$, the recommendation bias is measured by $\left( x_{\theta'} - \theta \right)^2$, and the continuation value of accepting is $\mathbb{E}_{ \pi\left( s \vert \theta'\right) } \left[ \max \left\{ v_\theta - p_s,0 \right\} \right]$.

The following investigates, given an intermediary's contract, the optimal decisions of the producer and consumers, which then impose constraints on the intermediary's design problem.

\subsection{Producer Optimality} 
The producer buys data to make the product pricing decision. As there are two product prices $\left\{v_L, v_H \right\}$, it is without loss of generality that the data consists of two signals $\mathcal S = \left\{l,h\right\}$. The data trade takes the form of price recommendation for the producer \citep{bergemann2016bayes}. Let the producer charge $v_L$ after observing signal $l$ and $v_H$ after observing signal $h$. That is,
\begin{equation*}
	p_l = v_L\text{,  } p_h = v_H 
\end{equation*}	
Given a data trade $\pi$, the producer charges a consumer of type $\theta$ a product price $v_L$ with probability $\pi\left( l \vert \theta\right)$ and $v_H$ with probability $\pi\left(h\vert \theta\right)$. We have the following result without loss of generality by \citet{bergemann2016bayes}.
\begin{lem}[Direct Data Trade]\label{lem_binary}
	The trade data consists of two signals $\mathcal S =\left\{l,h\right\}$ with
	\begin{eqnarray}
		\mu_h &\geq& \frac{v_L}{v_H} \label{obh_0} \\
		\mu_l &\leq& \frac{v_L}{v_H} \label{obl_0}.
	\end{eqnarray}
\end{lem}
We call inequalities~\eqref{obh_0} and~\eqref{obl_0} the producer's \textit{obedience} constraints. Under the constraints, the product pricing is the producer's best response to the direct data trade in which each signal is a price recommendation to the producer: pricing $v_L$ after signal $l$ and $v_H$ after signal $h$.

Plugging~\eqref{posteriors} into~\eqref{obh_0}~\eqref{obl_0}, the obedience constraints for the producer to follow pricing recommendations are summarized into
\begin{equation}
	 \frac{v_L}{v_H}\left( 1- \mu_0 \right) \pi\left(l\vert L\right) - \left( 1- \frac{v_L}{v_H} \right)\mu_0 \pi\left(l\vert H\right)  
\geq \max \left\{ \frac{v_L}{v_H} - \mu_0, 0\right\} \label{obs}
\end{equation}

Now we turn to the producer's acceptance decision over the intermediary's data trade offer $\left(\left\{\mathcal S, \pi\right\},T\right)$. When the producer does not accept the intermediary's data offer, he has no more information than the prior $\mu_0$ and then the revenue is $R_0 = \max \{ v_H\mu_0,v_L \}$. The producer accepts the data offer if and only if the product revenue from accepting is weakly larger than the one from rejecting
\begin{equation}
	R\left( \pi \right) - T \geq R_0 \label{irp}
\end{equation}
where $R(\pi)$ is from~\eqref{prod_reve}. 

In the following, we discuss the value of data to the producer under the direct data trade. 
Given bought data $\pi$, the producer charges $v_L$ (and all the consumers purchase the product)  with a probability $\mu_0 \pi(l\vert H) + ( 1-\mu_0 ) \pi(l\vert L)$, charges $v_H$ (and only the type-H consumers will purchase the product) with a probability $\mu_0 \pi(h\vert H) + ( 1-\mu_0 ) \pi(h\vert L)$. Hence, we can rewrite the revenue~\eqref{prod_reve} as
\begin{eqnarray*}
	R\left(\pi\right) 
	&=& \left[ \mu_0 \pi\left(l\vert H\right) + \left( 1-\mu_0 \right) \pi\left(l\vert L\right) \right] v_L + \left[ \mu_0 \pi\left(h\vert H\right) + \left( 1-\mu_0 \right) \pi\left(h\vert L\right) \right] \mu_h v_H  \\
	&=& \left[ \mu_0 \pi\left(l\vert H\right) + \left( 1-\mu_0 \right) \pi\left(l\vert L\right) \right] v_L +  \mu_0 \pi\left(h\vert H\right) v_H 
\end{eqnarray*}
where the second equality follows~\eqref{posteriors}.


The net gain from the data trade to the producer is	then	
\begin{eqnarray}
	   R\left(\pi\right) - R_0 
	&=& \left[ \mu_0 \pi\left(l\vert H\right) + \left( 1-\mu_0 \right) \pi\left(l\vert L\right) \right] v_L +  \mu_0 \left(1-\pi\left(l\vert H\right) \right) v_H  - \max\{ v_H\mu_0, v_L \} \notag \\
	&=&-\mu_0 \left(v_H - v_L \right) \underbrace{ \left[ \pi\left(l\vert H\right) - \mathbbm{1}_{\mu_0<\frac{v_L}{v_H}} \right] }_{\text{ surplus extraction } }
	+\left(1-\mu_0\right) v_L \underbrace{\left[ \pi\left(l\vert L\right) - \mathbbm{1}_{\mu_0<\frac{v_L}{v_H}} \right]}_{\text{market inclusion} } \label{netgain}
\end{eqnarray}
where $\mathbbm{1}$ denotes the indicator function. 
The data trade has two effects on the producer's revenue. The first is the change in the probability of serving the type-L consumers. With data trade the producer sells to the type-L consumers (of mass $(1-\mu_0)$) with probability $\pi(l\vert L)$. Instead, without data trade the type-L consumers are served with probability $\mathbbm{1}_{\mu_0<\frac{v_L}{v_H}} $. This is the extensive market \textit{inclusion} effect of data trade. 

The second effect is the change in the probability of yielding surplus to the type-H consumers. With data trade the producer charges the type-H consumers (of mass $\mu_0$) a low price $v_L$ with probability $\pi(l\vert H)$. Without data trade the type-H consumers are charged the low price $v_L$ with probability $\mathbbm{1}_{\mu_0<\frac{v_L}{v_H}} $. This is the intensive surplus \textit{extraction} effect of data trade.\footnote{It is direct that when $\mu_0<\frac{v_L}{v_H}$, the market inclusion effect is negative: The market serving the type-L consumers shrinks. The surplus extraction effect is positive: The producer extracts the type-H consumer surplus with a larger probability. Instead, when $\mu_0>\frac{v_L}{v_H}$, the market inclusion effect is positive: The market serving the type-L consumers expands. the surplus extraction effect is negative: the producer extracts the type-H consumers surplus with a smaller probability. }

The triangle gray set in~\Cref{fig_tradegain} illustrates the feasible data trade gain, which depends on the prior. Take $\mu_0>v_L/v_H$ as an example. Without data trade, the revenue is $\mu_0v_H$ which is at point $A$. It can go up to point $B$ under a data trade which induces posteriors $\mu_l$ and $\mu_h$. If the traded data is of full information, the intermediary extracts the maximum data trade surplus given $\mu_0$, which is $\mu_0 v_H + \left( 1- \mu_0 \right) v_L$.

\begin{figure}[h]
	\begin{center}
		\begin{tikzpicture}
			\draw [very thick,gray,<->] (0,5.5) node [above,black] {$R(\mu)$}--(0,5) node [left,black] {$v_H$} --(0,2) node [left,black] {$v_L$}--(0,0) node [below left,black] {0}--(1,0) node [below,black] {$\mu_l$}--(2,0) node [below,black] {$\frac{v_L}{v_H}$}--(3.5,0) node [below,black] {$\mu_0$}--(4.5,0) node [below,black] {$\mu_h$}--(5,0) node [below,black] {1}--(5.5,0) node [below right,black] {$\mu$};
			
			\foreach \x in {1,2,3.5,4.5,5}
			\draw (\x cm,4pt) -- (\x cm,-0pt);		
			\draw [very thick] (0,2)--(2,2)--(5,5);
			\draw [fill=gray!20!white] (0,2) -- (2,2) -- (5,5);	
			
			\draw [dashed] (1,0)--(1,2) (4.5,0)--(4.5,4.5);
			\draw [dashed] (3.5,0)--(3.5,2+2.5*2.5/3.5);
			\draw [red, fill=red] (3.5,2+2.5*2.5/3.5) circle (1.5pt);
			\draw [black, fill=black] (3.5,3.5) circle (1.5pt);
			
			\draw[->,dotted] (3,4.05) node [above] {\color{red}$B$} -- (3.45,3.8);
			\draw[->,dotted] (4,3.2) node [below] {$A$} -- (3.55,3.45);
			
			\draw [->, thick, dotted, blue] (3.5,0) to [out=150,in=30] (1,0);
			\draw [->, thick, dotted, blue] (3.5,0) to [out=60,in=120] (4.5,0);	
			
			\draw [red,thick] (1,2)--(4.5,4.5);
		\end{tikzpicture}	
	\caption{The triangle gray part is the feasible data trade gain from a data trade. Given $\mu_0 > v_L/v_H$: the product revenue is $\mu_0v_H$ at point $A$  without data trade. It goes up to point $B$ under a partial information data which induces posteriors $\mu_l$ and $\mu_h$. If the data is of full information, the product revenue is $\mu_0 v_H + \left( 1- \mu_0 \right) v_L$.}
	\label{fig_tradegain}	
	\end{center}
\end{figure}
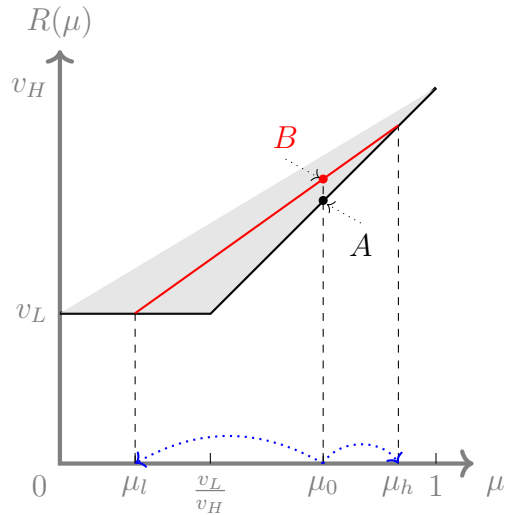

\subsection{Consumers Optimality} 

\paragraph{Consumers of Type $H$.} If the type-H consumers accept the service offer $\left(\mathcal X,x,f\right)$ and report $\theta'$ to the intermediary, the payoff in~\eqref{payoff_consumer} is
\begin{equation*}
	V -  \left( x_{\theta'} - H \right)^2 - f_{\theta'}
	+ \left( v_H - v_L \right) \pi\left( l \vert \theta'\right)
\end{equation*}
where $\left( v_H - v_L \right) \pi\left( l \vert \theta'\right)$ is the expected surplus for the type-H consumers in the product market. By reporting $\theta'$, a type-H consumer will be charged $v_L$ and obtain surplus $v_H-v_L$ with probability $\pi\left(l\vert \theta'\right)$. The data trade affects the probability of getting the surplus. The type-H consumers’ reporting affects both the service and its fee in the service market and the potential surplus in the product market. Hence, the type-$H$ consumers report the truth if and only if
\begin{equation}
	V - \left( x_H - H \right)^2 - f_H +  \left( v_H - v_L \right) \pi\left( l \vert H \right) 
\geq
	V - \left( x_L - H \right)^2 - f_L +  \left( v_H - v_L \right) \pi\left( l \vert L \right) \label{ich} 
\end{equation}

If the type-H consumers reject the service offer, they will get nothing in the service market. In this case, the producer will treat consumers as coming from the population and make product pricing decisions based on the prior $\mu_0$.\footnote{The equilibrium selection criterion restricts off-path beliefs to the prior, which is in the spirit of ``no signaling what you don't know" \citep{fudenberg1991perfect}: the intermediary cannot sell data to the producer if the offer is rejected by consumers, and the producer doesn't update beliefs about consumers from a non-participation decision. See more discussions in \Cref{discussion}.} That is, the product price $p_0$ without extra data is determined by~\eqref{prod_price} with the belief $\mu_0$. So the type-H consumers’ continuation value of rejecting the intermediary's contract is
\begin{equation*}
	\max \left\{ v_H - p_0, 0 \right\}
	\in \left\{
	\begin{array}{cc}
		\left\{ 0 \right\}  &\text{if } \mu_0 > \frac{v_L}{v_H} \\
		\{ v_H - v_L,0 \}  &\text{if }  \mu_0 = \frac{v_L}{v_H} \\
		\left\{  v_H - v_L \right\}  &\text{if } \mu_0 < \frac{v_L}{v_H}
	\end{array} 
	\right. 
\end{equation*}
Hence, the type-H consumers accept the contract if and only if	
\begin{equation}
	V - \left( x_H - H \right)^2  - f_H
	  +  \left( v_H - v_L \right)  \pi\left( l \vert H \right)
	\geq \max \left\{ v_H - p_0,0 \right\} \label{irh} 
\end{equation}

\paragraph{Consumers of Type $L$.} The type-$L$ consumers get nothing in the service market if rejecting the intermediary's service offer $\left(\mathcal X,x,f\right)$. In the product market the producer will charge $p_0$ and the type-L consumers' payoff is
\begin{equation*}
	\max \left\{ v_L - p_0, 0 \right\} = 0.
\end{equation*}
If the type-L consumers accept the offer and reporting $\theta'$, the type-L consumers receive $ V - ( x_{\theta'} - L )^2 - f_{\theta'}$ in the service market and they always get zero in the product market as $\mathbb{E}_{\pi\left(s \vert\theta' \right)} [ \max \{ v_L - p_s,0 \} ] = 0$. Thus, the type-L consumers' payoff from accepting the offer and reporting $\theta'$ in~\eqref{payoff_consumer} is 
\begin{equation*}
 V - \left( x_{\theta'} - L \right)^2 - f_{\theta'}
\end{equation*}
Hence, the type-L consumers report the truth if and only if
\begin{equation}\label{icl}
	V - \left( x_L - L \right)^2  - f_{L} 
	\geq V - \left( x_H - L \right)^2  - f_{H} 
\end{equation}
and accept the intermediary's contract if and only if
\begin{equation}\label{irl}
	   V - \left( x_L - L \right)^2 - f_{L} \geq 0 
\end{equation}

\subsection{Intermediary Optimality} \label{intermediary_optimality}
The intermediary maximizes the total profit~\eqref{payoff_intermediary} from both the service and data markets by specifying a contract subject to the producer's obedience constraints~\eqref{obs} 
and participation constraint~\eqref{irp}, the consumers' incentive constraints~\eqref{ich}~\eqref{icl} and participation constraints~\eqref{irh}~\eqref{irl}, and the following two feasibility conditions of data trade
\begin{equation}
	\pi\left(l\vert L\right)\text{,}\text{ }\pi\left(l\vert H\right)\in\left[0,1\right]  \label{fesi} 
\end{equation}

As the intermediary is a monopolist in the data market, it extracts all the trade surplus by setting $T=R\left(\pi\right) - R_0$. The design setting is non-standard because of the presence of data trade, which endogenizes the consumer surplus. 



\section{Optimal Design}\label{optimal_design}

It is ready to characterize the intermediary optimal mechanism. \Cref{binary_service} assumes the binary service to focus on the role of data trade; \Cref{general_service} returns to the general service to complete the characterization of intermediary optimal design of data collection and trade. 

\subsection{Special Case of Binary Service}\label{binary_service}

To focus on the role of data trade, this subsection assumes a binary service $\mathcal X = \left\{L,H\right\}$,\footnote{One can treat the assumption as a delegated provision of service by the intermediary to its department of marketing, while the data trade is directly controlled by the intermediary. The marketing department’s choice of service level does not coordinate with the data trade policy. This is one of the practices: the intermediary is usually not the producer of service, but recommends service.} under which there is no distortion in the service provided to consumers in the intermediary optimal mechanism. 

\begin{lem}\label{lem_2service}
	If $\mathcal X = \left\{L,H\right\}$, then $x_H = H$ and $x_L = L$.
\end{lem}
We make the following short notations for exposures of results.
\begin{nota}
	$\delta \equiv H- L$, $\Delta v \equiv v_H - v_L$, and $C\equiv (1-\mu_0) v_L/\left(\mu_0 \Delta v \right)$. \label{notation}
\end{nota}
The parameter $\delta^2$ captures the variation in consumers' tastes, reflecting the {\it horizontal} differentiation level in the service market, while $\Delta v$ captures the variation in consumers' willingness to pay, measuring the {\it vertical} differentiation level in the product market.
When the service is binary, the service distortion if any is fixed at $\delta^2$. Distorting a service provision requires compensation to consumers. Instead, without any distortion of service provisions, the intermediary can extract the most consumer surplus by charging a higher service fee. Hence, the intermediary will not distort the service provided to consumers.

The intermediary's optimal mechanism for data collecting and selling is summarized in the following proposition.
\begin{prop}[Optimal data\textemdash binary service]\label{prop_binary}
	Given a binary service space. The optimal data trade depends on the \textit{cross-market} comparison between the horizontal differentiation level in the service $\delta^2$ and the vertical differentiation in the product $\Delta v$:
	
	\begin{enumerate}[label=\arabic*.]
		
		\item Suppose $\mu_0>v_L/v_H$. The intermediary's data trade is $\pi\left(l\vert L\right)=\min\{ \frac{\delta^2}{\Delta v}, 1\}$, $\pi\left(l\vert H\right)$ can be anywhere between $0$ and $C\min\{ \frac{\delta^2}{\Delta v}, 1\}$.
		
		       	
		
		\item Suppose $\mu_0<v_L/v_H$. The intermediary's optimal data trade is 
			$\pi\left(l\vert L\right) = 1$, $\pi\left(l\vert H\right)$ can be anywhere between $\max\left\{1-\frac{\delta^2}{\Delta v}, 0\right\}$ and $1$. 
		
	\end{enumerate}
	
\end{prop}

The obedience constraints of the producer in~\eqref{obs}, the producer's revenue without any data in~\eqref{irp}, and the type-H consumers' participation constraint in~\eqref{irh} depend on the market prior $\mu_0$. The following illustrates the results, taking $\mu_0 > v_L/v_H$ as an example.\footnote{The detailed results of $\mu_0 < v_L/v_H$ are relegated to the appendix. The special parameter of $\mu_0=v_L/v_H$ is captured by either case, but it also raises more potential solutions. To see this, note that at $\mu_0=v_L/v_H$, a) the constraint~\eqref{obs} falls into the case of $\mu_0 > v_L/v_H$; b) $p_0$ equals to $v_L$ or $v_H$ and $R_0=v_L$ or $v_H\mu_0$, but~\eqref{irp} doesn't affect analysis because $T=R(\pi)-R_0$ and $R_0$ enters the objective as a constant; c)~\eqref{irh} becomes either~\eqref{irh_largemu} of the case of $\mu_0 > v_L/v_H$, or~\eqref{irh_smallmu} of the case of $\mu_0 < v_L/v_H$. Thus, the special case of $\mu_0=v_L/v_H$ could raise more solutions as there is an extra combination of constraints including~\eqref{obl} and~\eqref{irh_smallmu} together. This combination of constraints affects the data fee $f_H$, but does not affect the results about data trade.} 
When $\mu_0> v_L/v_H$, the obedience constraint for the producer to charge a high price~\eqref{obh_0} is redundant and the obedience constraint for the producer to charge a low price reads
\begin{equation}
	 \pi\left(l\vert H\right) \leq C \pi\left(l\vert L\right) \label{obl}
\end{equation} 
where $C \in \left(0,1\right)$ is in~\Cref{notation}.
The type-H consumers will be charged a product price $p_0=v_H$ when they reject the service offer, for $\mu_0> v_L/v_H$. Hence, the type-H consumers’ outside option $\max\left\{v_H-p_0,0\right\}$ is zero. The participation constraint of the type-H consumers~\eqref{irh} is
\begin{equation}
	V - \left( x_H - H \right)^2 - f_H
	+ \left( v_H - v_L \right) \pi\left( l \vert H \right)
	\geq 0 \label{irh_largemu}
\end{equation}

The producer charges $v_H$ without buying any data. The type-H consumers purchase the product and the type-L consumers are excluded from the market. The producer’s revenue is $R_0 = \mu_0 v_H$. The producer's net gain from data trade in~\eqref{netgain} becomes
\begin{equation*}
	R\left(\pi\right) - R_0 
	=\underbrace{ - \mu_0\left(v_H - v_L \right) \pi\left(l\vert H\right) }_{\text{surplus extraction } }
	+\underbrace{\left(1-\mu_0\right) v_L \pi\left(l\vert L\right) }_{\text{market inclusion} }.
\end{equation*}
The market inclusion effect of data trade is positive $(1-\mu_0)v_L>0$: The data trade extends the product market by allowing the producer to serve the type-L consumers (of mass $1-\mu_0$) with a higher probability $\pi (l\vert L )>0$. The surplus extraction effect of data trade is negative $- \mu_0 \Delta v <0$. The producer extracts less type-H consumer surplus with data trade. With data trade, the producer charges the type-H consumers a low price $v_L$ with probability $\pi\left(l\vert H\right)$. Instead, without any data trade, all the type-H consumers are charged the high price $v_H$ and the intermediary does not yield any surplus to consumers.

\paragraph{Irrelevance of $\pi\left(l\vert H\right)$ in Intermediary's Objective.} The intermediary extracts all the surplus from both types of consumers in the service and product markets, the participation constraints of both types of consumers are binding.\footnote{This is different from the standard mechanism design: 
all consumer surpluses are endogenous and extracted in the optimal mechanism with data trade.} That is,
\begin{equation*}
	f_H = V + \pi\left(l\vert H\right) \Delta v,\text{ }
	f_L = V .
\end{equation*}
In the optimal mechanism, each type-H consumer has surplus $V$ from the service and gets surplus $\Delta v$ with probability $\pi\left(l\vert H\right)$ from the product. Each type-L consumer has a surplus $V$ from the service and nothing from the product. 

However, in the data trade market, the intermediary needs to shift surplus $\Delta v$ with probability $\pi\left(l\vert H\right)$ to each type-H consumer. The intermediary in the data market works as the producer in the product market because it extracts all the data trade surplus. Due to the surplus-shifting, the probability of charging type-H consumers a low price $\pi\left(l\vert H\right)$ will not affect the intermediary's payoff. To see this, plugging the binding $IR_L$~\eqref{irl} and $IR_H$~\eqref{irh_largemu}: $f_L =V$, $f_H = V + \pi\left(l\vert H\right) \Delta v$ into the intermediary's objective function
\begin{equation*}
	V + \mu_0 \underbrace{  \pi\left(l\vert H\right) \Delta v }_{\text{surplus-extracting}} 
  - \mu_0 \underbrace{ \pi\left(l\vert H\right)\Delta v}_{\text{surplus-shifting }} + \left(1-\mu_0\right)\pi\left(l\vert L\right) v_L 
\end{equation*}
where the profit from the service market is $V+\mu_0\pi\left(l\vert H\right)\Delta v$ and from the data market is $-\mu_0\pi\left(l\vert H\right)\Delta v + \left(1-\mu_0\right)\pi\left(l\vert L\right) v_L$. The surplus-shifting and the surplus-extracting offset each other.\footnote{It is because the intermediary has full bargaining power over the data trade surplus in the data market. \Cref{bargaining_power} instead looks at the partial bargaining power of the intermediary.} 

\paragraph{Intermediary's Tradeoff.} With binding $IR_L$~\eqref{irl} and $IR_H$~\eqref{irh_largemu}, then $IC_L$~\eqref{icl} will be redundant, the intermediary only needs to consider $OB_L$~\eqref{obl} and $IC_H$~\eqref{ich}. The intermediary's reduced problem is to solve 
\begin{eqnarray*}
	&\max_{\pi}&V + \left(1-\mu_0\right)\pi\left(l\vert L\right) v_L \\
	&\text{subject to} &\pi\left(l\vert L\right) \Delta v  - \delta^2 \leq 0 \\
	&&\pi\left(l\vert H\right) \leq C \pi\left(l\vert L\right)
\end{eqnarray*}
In the first constraint of $\pi\left(l\vert L\right) \Delta v - \delta^2 \leq 0$, if a type-H consumer reports the truth, she will get nothing. Instead, if she misreports, she will get the surplus $\Delta v$ with probability $\pi\left(l\vert L\right)$ but incur mismatch loss in the service $\delta^2$. To incentivize the type-Hs to tell the truth, the intermediary’s choice of $\pi\left(l\vert L\right) $ has to satisfy $\pi\left(l\vert L\right) \Delta v \leq \delta^2$.

The above reduced problem clarifies the tradeoff faced by the intermediary. The type-H consumers’ expected surplus in the product market by misreporting, $\pi\left(l\vert L\right) \Delta v$, is bounded by the mismatch in the service $\delta^2$. However, the intermediary prefers $\pi\left(l\vert L\right)$ as high as possible because it can extract the type-L consumers’ surplus with a higher probability in the product market.

Intuitively, there are two instruments for the intermediary to incentivize the type-H consumers to tell the truth: lowering the probability of charging $v_L$ for those who report $L$, and the mismatch (wrong service) itself as a penalty. The mismatch penalty is preferred by the intermediary because it does not require to lower $\pi\left(l\vert L\right)$. 

When the horizontal differentiation level is weakly above the vertical one, $\delta^2\geq \Delta v$, the mismatch cost in service is high enough as a penalty. In this case, the type-H consumers will not misreport and the intermediary would like to set $\pi\left(l\vert L\right) $ as high as possible, that is $\pi\left(l\vert L\right) =1$. The type-L consumers are all served in the product market. 

When the horizontal differentiation level is below the vertical one, $\delta^2< \Delta v$, the mismatch cost in service is low as a penalty for the type-H consumers' misreport. The mismatch itself cannot provide sufficient incentive for the type-H consumers to reveal their information. Instead, the intermediary would like to lower $\pi\left(l\vert L\right)$ to incentivize truth-telling to collect information. In this case, some type-L consumers are excluded from the product market. 


\subsection{Intermediary Optimal Mechanisms} \label{general_service}

We return to the complete characterization of optimal data collection and data trade under the service space $\mathcal X = \mathbb R$. 
The analysis shows that the fundamental tradeoff of data trade works under the general service, while admitting flexible service distortions and their effects on consumers' incentives of sharing data.

When $\mu_0<v_L/v_H$, even allowing general service provisions, there is {\it no distortion} of service in the intermediary optimal mechanism. The results on service fee, data trade, and data fee are the same as those in \Cref{prop_binary}. 
We state them as follows.

\begin{prop}[Optimal design\textemdash small $\mu_0$]\label{prop_general_service_smallmu0}
	Suppose $\mu_0<v_L/v_H$. In the optimal mechanism, there is no distortion in service provisions $x_L = L$, $x_H = H$, the service fees are $f_L=V$, $f_H = V + [ \pi( l \vert H)-1 ] \Delta v $, and the intermediary's data trade is
	\begin{equation*}
		\pi\left(l\vert L\right)=1,\text{ } \pi\left(l\vert H\right)\in\Big[\max \Big\{ 0,1-\frac{\delta^2}{\Delta v} \Big\},1\Big],
	\end{equation*}
	and the data trade fee is $T = \mu_0 [ 1-\pi\left(l\vert H\right) ]\Delta v$.
	In this case, social welfare is $V + \left(1-\mu_0\right)v_L + \mu_0 v_H.$
	
\end{prop}

Now we consider the case of $\mu_0>v_L/v_H$. Although there is no distortion in service provisions (and the service fees, data trade are as in part 1 of~\Cref{prop_binary}) for $\delta^2\geq \Delta v$, the service provided to the type-L is downward distorted for $\delta^2 < \Delta v$. We summarize the results for $\mu_0>v_L/v_H$ in the following proposition.

\begin{prop}[Optimal design\textemdash large $\mu_0$]\label{prop_general_service_largemu0}
	Suppose $\mu_0>v_L/v_H$, in the optimal mechanisms: 
	
	\begin{enumerate}[label=\arabic*.]
		
		\item For $\frac{\delta^2}{\Delta v} \geq 1$, there are no distortion in service provisions $x_L = L$, $x_H = H$, the service fees are $f_L=V$, $f_H = V + \pi( l \vert H)\Delta v $, and the intermediary's data trade is
		\begin{equation*}
			\pi\left(l\vert L\right)=1,\text{ } \pi\left(l\vert H\right)\in\left[0,C\right]
		\end{equation*}
		and the data trade fee is $	T = -\mu_0\pi\left(l\vert H\right)\Delta v + \left(1-\mu_0\right)v_L$,
		and the intermediary's total revenue is $V+(1-\mu_0)v_L$.
		In this case, social welfare is $V + \left(1-\mu_0\right)v_L + \mu_0v_H$.
		
		\item 
		For $\delta^2<\Delta v$, there is a downward distortion in service provided to type-$L$ consumers, but there is no service distortion for type-$H$ consumers $x_H=H$. The following holds:
		\begin{enumerate}
			
			\item If $ \frac{\Delta v}{v_L+v_H}<\frac{\delta^2}{\Delta v}<1$, then the service for type-$L$ is $x_L = L - \frac{\Delta v -\delta^2}{2\delta}$, the service fees are $f_L=V - (\frac{\Delta v -\delta^2}{2\delta})^2$, $f_H = V + \pi( l \vert H)\Delta v $; and the data trade is 
			\begin{equation*}
				\pi(l\vert L) = 1,\text{ } \pi(l\vert H) \in [ 0,C]
			\end{equation*}
			the data trade fee is $T = - \mu_0 \pi(l\vert H) \Delta v  + (1-\mu_0) v_L $
			and the intermediary's total revenue is $V - (1-\mu_0) \left(\frac{\Delta v -\delta^2}{2\delta}\right)^2
			+ (1-\mu_0)v_L$.
			In this case, social welfare is $V - (1-\mu_0)\left( \frac{\Delta v - \delta^2}{2\delta}\right)^2 + (1-\mu_0)v_L + \mu_0 v_H$.
			
			\item If $\frac{\delta^2}{\Delta v} \leq \frac{\Delta v}{v_L+v_H}$, then the service for type-$L$ is $x_L= L - \frac{\delta}{\Delta v} v_L$, the service fees are $f_L = V - ( \frac{\delta}{\Delta v} v_L )^2$, $f_H = V + \pi( l \vert H)\Delta v $; and the data trade is
			\begin{equation*}
				\pi(l\vert L) = \frac{\delta^2}{\Delta v} \frac{v_H+v_L}{\Delta v}, \text{ }  \pi(l\vert H) \in \left[ 0, C \frac{\delta^2}{\Delta v} \frac{v_H+v_L}{\Delta v} \right]
			\end{equation*}
			the data trade fee is $T = -\mu_0  \pi(l\vert H) \Delta v + (1-\mu_0)v_L \frac{\delta^2}{(\Delta v)^2}(v_L+v_H)$,
			and the intermediary's total revenue is $V + (1-\mu_0) \left(\frac{\delta}{\Delta v}\right)^2 v_L v_H$.
			In this case, social welfare is $V + (1-\mu_0) \left(\frac{\delta}{\Delta v}\right)^2 v_L v_H + \mu_0 v_H$.
			
		\end{enumerate}
		
	\end{enumerate}
	
\end{prop}
When $\mu_0>v_L/v_H$ and the horizontal differentiation level is high relative to the vertical one (i.e., $\delta^2 \geq \Delta v$), there is no distortion of service provisions to consumers. This is because the mismatch cost in service is so high that consumers are willing to report the private information without any further distortion in service level. The intermediary sets the highest $\pi\left(l\vert L\right)=1$ and does not need to distort any service level to maintain the type-H's incentive. 

However, when $\mu_0>v_L/v_H$ and the horizontal differentiation level is below the vertical one, $\delta^2<\Delta v$, the intermediary needs to downward distort the type-L's service. In this case, the penalty of the wrong service cannot provide enough incentive. The intermediary will manipulate the type-H consumers’ incentive via the other two instruments: the probability of surplus in the product and service provision. Although distorting the service will reduce the surplus extracted from the type-Ls in the service market, the intermediary can extract more surplus from the type-H consumers in the product market. When the benefit of more surplus from type-H consumers outweighs the cost of reduced surplus from type-L consumers, the intermediary would prefer to distort the service to the type-L consumers. This happens when the population of type-H consumers is large enough $\mu_0>v_L/v_H$. 

\begin{table}[h]
	\centering
	\caption{Main Results of Optimal Mechanism}
	\label{tab_results_general}
	
	\begin{adjustbox}{width=1\textwidth}
		\begin{tabular}{ccccc}
			Prior & Condition & Service Provision & Data Trade  \\	
			\cmidrule{1-4}   
			\multirow{3}[1]{*}{$\mu_0>\frac{v_L}{v_H}$} 
			&\textcolor{black}{$\frac{\delta^2}{\Delta v} \geq 1$} & no distortion &\textcolor{black}{$\pi(l|L)=1,\pi(l|H)\in [0,C]$}  \\
			&\textcolor{black}{$\frac{\Delta v}{v_H+v_L}\leq \frac{\delta^2}{\Delta v}<1$} & downward distortion at bottom &\textcolor{black}{$\pi(l|L)=1,\pi(l|H)\in[0,C]$} \\
			&\textcolor{black}{$\frac{\delta^2}{\Delta v} < \frac{\Delta v}{v_H+v_L}$} & downward distortion at bottom &\textcolor{black}{$\pi(l|L)=\frac{\delta^2}{\Delta v}\frac{v_H+v_L}{\Delta v},\pi(l|H)\in[0,C\frac{\delta^2}{\Delta v}\frac{v_H+v_L}{\Delta v}]$} \\
			
			\cmidrule{1-4}
			\multirow{2}[1]{*}{$\mu_0<\frac{v_L}{v_H}$}
			&\textcolor{black}{$\frac{\delta^2}{\Delta v}\geq 1$} &\textcolor{black}{no distortion} 
			& \textcolor{black}{$\pi(l|L)=1,\pi(l|H)\in [0,1]$} \\
			&\textcolor{black}{$\frac{\delta^2}{\Delta v}<1$} &\textcolor{black}{no distortion} 
			& \textcolor{black}{$\pi(l|L)=1,\pi(l|H)\in [1-\frac{\delta^2}{\Delta v},1]$} \\
			
			\cmidrule{1-4} 
		\end{tabular}
	\end{adjustbox}
	
\end{table}

Following~\Cref{prop_general_service_smallmu0} and~\ref{prop_general_service_largemu0} (summarized in~\Cref{tab_results_general}) directly, we can discuss some special data trade policies, such as full privacy and full disclosure.
\begin{cor}[Full privacy]\label{cor_L} If $\mu_0>v_L/v_H$, the null information data trade cannot be optimal. If $\mu_0<v_L/v_H$, the null information data trade is always optimal.
\end{cor} 
It says that full privacy cannot be optimal for $\mu_0>v_L/v_H$.\footnote{\cite{calzolari2006optimality} offer conditions for full privacy to be optimal. Their condition about the same sign of single-crossing conditions is violated in our setup.} To see it by contradiction. Suppose in the optimal mechanism, $\pi\left(l\vert H\right) = \pi\left(l\vert L\right)$. Then the obedience constraint~\eqref{obl} requires $(1-\frac{v_L}{v_H})\mu_0 \leq \frac{v_L}{v_H}\left(1-\mu_0\right)$ which contradicts with $\mu_0>\frac{v_L}{v_H}$. Intuitively, when $\mu_0>\frac{v_L}{v_H}$ any extra information will make it easier for the producer to price discriminate the consumers. This is because in the case of $\mu_0>\frac{v_L}{v_H}$ the market inclusion effect is positive: without any data trade all the type-L consumers are excluded in the product market. 

But the optimal data trade can feature full privacy for $\mu_0 < v_L/v_H$. As the product market has included the type-L consumers when $\mu_0 < v_L/v_H$, there is no space for the producer to extract type-L consumer surplus from the extra information of data trade. 

		

\begin{cor}[Full disclosure] \label{cor_notfull} 
	For $\mu_0>v_L/v_H$, selling full information is optimal to the intermediary if and only if $\frac{\delta^2}{\Delta v} \geq \frac{\Delta v}{v_L+v_H}$. For $\mu_0<v_L/v_H$, selling full information is optimal to the intermediary if and only if $\frac{\delta^2}{\Delta v} \geq 1$.
\end{cor}

For instance, in the case of $\mu_0>v_L/v_H$, the intermediary will not sell full information in the optimal mechanism when $\delta^2<(\Delta v)^2/(v_L+v_H)$.
To see this by contradiction, in the optimal mechanism, the service profit is $V+\mu_0\pi(l|H)\Delta v-(1-\mu_0)(\frac{\delta}{\Delta v}v_L)^2$ and the data profit is $-\mu_0 \pi\left(l\vert H\right) \Delta v + (1-\mu_0) v_L \frac{\delta^2}{\Delta v}\frac{v_L+v_H}{\Delta v}$. 
Suppose the data policy sells full information, the service profit decreases to $V-(1-\mu_0)(\frac{\Delta v - \delta^2}{2\delta})^2$, while the data profit increases to $(1-\mu_0) v_L$.
However, the amount of loss in service profit is larger than that of increment in the data profit when $\delta^2<(\Delta v)^2/(v_L+v_H)$.\footnote{It is because given $\delta^2<(\Delta v)^2/(v_L+v_H)$, {\footnotesize\begin{eqnarray*}
		&&V+\mu_0\pi(l|H)\Delta v-(1-\mu_0)(\frac{\delta}{\Delta v}v_L)^2
		-  \left( V+(1-\mu_0)(\frac{\Delta v - \delta^2}{2\delta})^2\right) \\
		&&\hspace{1in} 
		> (1-\mu_0) v_L
		-\left( -\mu_0 \pi\left(l\vert H\right) \Delta v + (1-\mu_0) v_L \frac{\delta^2}{\Delta v}\frac{v_L+v_H}{\Delta v} \right).
\end{eqnarray*}}} 
To preserve the service profit, the intermediary would prefer to sell partial information.
Instead, when $\delta^2\geq (\Delta v)^2/(v_L+v_H)$, selling full information can be optimal for the intermediary, even considering the cross-market effect of data trade. 

Another interesting observation about the service fees by contrasting the results in~\Cref{prop_general_service_smallmu0} with~\Cref{prop_general_service_largemu0}  is that the high type consumers end up paying more service fee in the data-sourcing market for a larger $\mu_0$; but the low type consumers pay less (equal) service fee if $\delta^2<\Delta v$ ($\delta^2\geq\Delta v$ ) for a larger $\mu_0$.

\subsection{Regulation of Data Sharing} \label{policy}

The above framework with a holistic viewpoint, considering data sourcing, selling, and using, provides an opportunity to investigate the data market regulations in an integrated way. The regulatory policies have been established across different countries to address concerns about consumer privacy, led by data transfers across markets. However, a comprehensive understanding of the potential effects of regulations on the data market is still missing. The subsection studies the privacy regulation (or the regulation of data sharing).\footnote{\Cref{bargaining_power} discusses the intermediary's partial bargaining power and its regulation.}  

\paragraph{Efficiency.} In the socially optimal allocation, each consumer is provided a correct service and offered a product, and hence the socially optimal welfare is $V+\mu_0v_H+(1-\mu_0)v_L$. Following the propositions~\ref{prop_general_service_smallmu0} and~\ref{prop_general_service_largemu0} (summarized in~\Cref{tab_results_general}), it is trivial that the intermediary optimal allocation is efficient for $\mu_0<v_L/v_H$. However, for $\mu_0>v_L/v_H$, the optimal mechanism of the intermediary is efficient if and only if $\delta^2\geq \Delta v$.

\begin{lem}[Efficient mechanism]
	Suppose $\mu_0>v_L/v_H$.  The intermediary's optimal mechanisms are efficient if and only if $\delta^2\geq \Delta v$.
\end{lem}
When $\delta^2\geq \Delta v$, there is no distortion of service provisions, and the optimal data trade which fully discloses information about consumers' willingness to pay induces perfect price discrimination in the product market. Hence, the intermediary optimal mechanism is efficient under certain market differentiation conditions. 
The efficiency result implies that the data market regulations that intervene in the amount of data may lead to inefficiency. 


\paragraph{Surplus and Welfare.} To investigate the effects of banning data trade, we compare the surpluses and welfare with and without the data market. When the laws ban the data trade to pursue the highest level of consumer privacy protection, the service market and product market work separately because they are not linked by the intermediary’s data trade. The data shared with the intermediary will no longer be used in another market against the consumers’ interests. In this environment, the type-L consumers’ total surplus is zero. The type-H consumers get some surplus when the product price is $v_L$, which corresponds to the case of $\mu_0<\frac{v_L}{v_H}$. In the product market, the producer extracts all the surplus under the prior $\mu_0$. The producer’s surplus is $R_0 = \max\left\{v_L,v_H\mu_0\right\}$.

\begin{lem}[Results without data market]\label{eq_ban}
	Suppose there exists no data trade. Social welfare is $V+\mu_0 v_H$ for $\mu_0> v_L/v_H$, and is $V+v_L+\mu_0\Delta v$ for $\mu_0<v_L/v_H$. 
	\begin{enumerate}[label=\arabic*.]
		
		\item In the service market, the intermediary sets $x_L=L$ and $x_H=H$ and the service fees are $f_L = f_H = V$. Both types of consumers surplus in service are zero and the intermediary surplus is $V$.
		
		\item In the product market, the producer sets price $p_0$. 
		The type-L consumers surplus is $0$. The type-H consumers surplus is $\Delta v$ when $\mu_0<v_L/v_H$ and is $0$ when $\mu_0>v_L/v_H$. The producer gets $\max\left\{v_L,\mu_0v_H\right\}$.
		
	\end{enumerate}

\end{lem}

\Cref{tab_surplus} summarizes the consumers, the producer, and the intermediary surplus with and without data trade ban. The data trade ban does not affect the producer surplus and consumers surplus. The producer surplus is always the producer’s revenue in the product market without buying data, $R_0$. This is because the intermediary has the full bargaining power over the data trade surplus and the producer does not get any share of the data trade surplus even when there is a data market. The type-L consumers surplus is always zero: The intermediary extracts their surplus in the service market and they don’t have any surplus in the product market. The type-H consumers surplus is $\mu_0\Delta v$ in the product market.

Banning data trade affects the intermediary's payoff only when $\mu_0>v_L/v_H$. Note that selling {\it no} information is optimal to the intermediary for $\mu_0<v_L/v_H$ by~\Cref{prop_general_service_smallmu0}, hence, banning data trade won't hurt the intermediary for $\mu_0<v_L/v_H$. However, it will hurt the intermediary for $\mu_0>v_L/v_H$, because selling no information is not optimal from~\Cref{prop_general_service_largemu0}.
The intermediary can only extract the service surplus $V$ without the data market, but gets better off by serving an extra proportion of type-L consumers with the data market.

\begin{table}[h]
	\centering
	\caption{Surplus Distribution}
	\begin{adjustbox}{width=\textwidth}
		\begin{tabular}{ccccccccc}
			\toprule
			\multicolumn{2}{c}{ } & \multicolumn{3}{c}{With Data Market} &  & \multicolumn{3}{c}{Without Data Market} \\
			
			\cmidrule{3-5} \cmidrule{7-9}
			Primitive & Condition & Intermediary & Producer & Consumers &  & Intermediary & Producer & Consumers \\
			
			\cmidrule{1-5}\cmidrule{7-9}    
			\multirow{2}[1]{*}{$\mu_0>\frac{v_L}{v_H}$} 
			&$\frac{\delta^2}{\Delta v} \geq 1$& $V+\left(1-\mu_0\right)v_L$        &$\mu_0v_H$ & 0  & & $V$ & $\mu_0v_H$& 0 \\
			&$\frac{\Delta v}{v_H+v_L}\leq \frac{\delta^2}{\Delta v}<1$& $V+(1-\mu_0)v_L-(1-\mu_0)(\frac{\Delta v-\delta^2}{2\delta})^2$        &$\mu_0v_H$ & 0  & & $V$ & $\mu_0v_H$& 0 \\
			&$\frac{\delta^2}{\Delta v}<\frac{\Delta v}{v_H+v_L}$&$V+(1-\mu_0)v_Lv_H (\frac{\delta}{\Delta v})^2$ & $\mu_0v_H$& 0 & & $V$ & $\mu_0v_H$& 0\\
			
			\cmidrule{1-5}\cmidrule{7-9}
			\multirow{2}[1]{*}{$\mu_0<\frac{v_L}{v_H}$} 
			&$\frac{\delta^2}{\Delta v} \geq 1$ &$V$ &$v_L$ &$\mu_0\Delta v$  & &$V$ & $v_L$& $\mu_0\Delta v$ \\
			&$\frac{\delta^2}{\Delta v} < 1$ &$V$ &$v_L$ &$\mu_0\Delta v$ & &$V$ & $v_L$& $\mu_0\Delta v$   \\
			
			
			\bottomrule
		\end{tabular}%
	\end{adjustbox}
	\label{tab_surplus}
\end{table}

Social welfare is the sum of the intermediary's profit, producer's profit, and consumer surplus in~\Cref{tab_surplus}. Hence, banning the data market decreases social welfare when $\mu_0>v_L/v_H$.\footnote{\citet{acemoglu2022too} find that banning the data market improves welfare as the over-sharing of data in the market induced by data externality disappears. The welfare loss comes from the firm’s over-collection of consumer data in \citet{bergemann2022economics} and \citet{choi2019privacy}.} Intuitively, without data trade, it becomes harder for the producer to price discriminate consumers and extract the most surplus. Specifically, the data trade brings two effects: surplus extraction and market inclusion. The surplus extraction effect does not matter for social welfare as it is a redistribution of surplus between the producer and consumers.
But the positive market inclusion effect will disappear if the data trade is banned.

\begin{prop}[Regulation of data sharing]\label{general_service_welfare}
	Banning data trade will decrease social welfare for $\mu_0>v_L/v_H$, and will not affect social welfare for $\mu_0<v_L/v_H$.
\end{prop}

The result suggests that the privacy laws should be careful about the levels of protection. Depending on the market environment, the highest privacy protection is sometimes optimal (e.g., $\mu_0<v_L/v_H$), while trading full data is sometimes optimal (e.g., $\delta^2\geq \Delta v$).

\section{Intermediary's Bargaining Power}\label{bargaining_power}

One assumption in the baseline analysis of intermediary optimal design is that the intermediary has the full bargaining power in the data market and hence extracts all the data trade surplus. However, transferring data from consumers to the intermediary does not necessarily fully transfer the bargaining power to the latter. 
This section studies the partial bargaining power of the intermediary and its regulatory implications.

In a stochastic ultimatum bargaining between the intermediary and the producer in the data market, the intermediary proposes a take-it-or-leave-it offer with probability $\beta\in(0,1)$, which measures the intermediary's relative bargaining power. When the intermediary extracts a proportion $\beta$ of the data trade surplus, its total payoff becomes
\begin{equation*}
	  f_H \mu_0  + f_L (1-\mu_0) 
	+ \beta [ (1-\mu_0) \pi(l\vert L) v_L - \mu_0 \pi(l\vert H) (v_H - v_L) + (\mu_0 v_H - v_L) \mathbb{I}_{\mu_0<v_L/v_H}  ],
\end{equation*}
while the constraints of the problem are unchanged. Solving the problem yields results about the intermediary's optimality with partial bargaining power.

\begin{prop}[Optimal design\textemdash partial bargaining power] \label{prop:bargaining_general} Suppose a partial bargaining power of the intermediary. In the intermediary optimal mechanism:
	\begin{enumerate}[label=\arabic*.] 
		\item \label{prop:bargaining_general_largemu0} For $\mu_0 > v_L/v_H$, there is no distortion of the service for the high type $x_H = H$ and
		\begin{enumerate}
			
			\item if $\frac{\delta^2}{\Delta v} \geq 1$, then the data trade  $\pi(l\vert L)=1$, $\pi(l\vert H)=C$ and the service is $x_L=L$;

			\item if $\frac{\Delta v}{v_H+v_L} <\frac{\delta^2}{\Delta v}<1$, then the data trade is $\pi(l\vert L) = 1$, $\pi(l\vert H) = C$ and the service is $x_L = L - \frac{\Delta v -\delta^2}{2\delta}$;
			
			\item if $\frac{\delta^2}{\Delta v} \leq \frac{\Delta v}{v_H+v_L}$, then the data trade is $\pi(l\vert L) = \frac{\delta^2}{\Delta v} 
			\frac{v_H+v_L}{\Delta v}$, $\pi(l\vert H) = \frac{C\delta^2}{\Delta v} 
			\frac{v_H+v_L}{\Delta v}$
			and the service is $x_L = L - \frac{\delta v_L }{\Delta v}$.
			
		\end{enumerate}
		
		\item \label{prop:bargaining_general_smallmu0}	For $\mu_0<v_L/v_H$,
		the optimal data trade is $\pi\left(l \vert L\right)=\pi\left(l \vert H\right) = 1$, and services are $x_L=L$, $x_H=H$. 
		
	\end{enumerate}
\end{prop}

Comparing the intermediary's optimality under the partial bargaining power with that under the full bargaining power in the propositions~\ref{prop_general_service_smallmu0} and~\ref{prop_general_service_largemu0}, two main results follow. First, the partial bargaining power can decrease the amount of traded data. The optimal $\pi(l|H)$ under the partial bargaining power becomes unique from multiple solutions. For example, in the case of $\mu_0>v_L/v_H$, fixing the optimal $\pi(l|L)$, the data trade changes from $\{ \pi(l|L), \pi(l|H) \in [ 0,C\pi(l|L)]  \}$ to $\{\pi(l|L),\pi(l|H)=C\pi(l|L)\}$ with $C\in (0,1)$.\footnote{When $\mu_0>\frac{v_{L}}{v_{H}}$, in the optimal mechanism~\eqref{obl} is binding, $\pi(l|H)=C\pi(l|L)$, which together with other constraints pins down the unique solution.}
Hence, the partial bargaining power drives down the amount of information disclosed by the intermediary's optimal mechanism. 

Second, the partial bargaining power does not affect the intermediary's equilibrium payoff, although it shuts off the multiple solutions of data trade.\footnote{In the optimal mechanism, the intermediary's equilibrium payoff is $V$ for $\mu_{0}<\frac{v_{L}}{v_{H}}$, and it is
	\begin{equation*}
		\begin{cases}
			V + (1-\mu_0) v_L  &\mbox{ if } \frac{\delta^2}{\Delta v} \geq 1 \\
			V+ (1-\mu_0) v_L - (1-\mu_0)\left(\frac{\Delta v -\delta^2}{2\delta}\right)^2  &\mbox{ if } \frac{\Delta v}{v_H+v_L} \leq \frac{\delta^2}{\Delta v}  < 1 \\
			V+ (1-\mu_0)\left(\frac{\delta}{\Delta v}\right)^2 v_Lv_H  &\mbox{ if } \frac{\delta^2}{\Delta v} < \frac{\Delta v}{v_H+v_L}
		\end{cases}
	\end{equation*} 
	for $\mu_{0}>\frac{v_{L}}{v_{H}}$, which is exactly the intermediary's equilibrium payoff in the setup of $\beta=1$.}
In the case of $\mu_0>v_L/v_H$, the intermediary's total payoff can be written as
$[ V-(x_H-H)^2+\pi(l|H)\Delta v ]\mu_0 + [ V-(x_L-L)^2](1-\mu_0)
+ \beta [ (1-\mu_0)\pi(l\vert L) v_L -\mu_0 \pi(l\vert H)\Delta v ]
$
where in the optimal mechanism the data revenue is zero because
\begin{eqnarray*}
	(1-\mu_0)\pi(l\vert L) v_L -\mu_0 \pi(l\vert H)\Delta v &=& (1-\mu_0)\pi(l\vert L) v_L -\mu_0 C\pi(l\vert L)\Delta v \\
	&=& (1-\mu_0)\pi(l\vert L) v_L - (1-\mu_0) \pi(l\vert L) v_L \\
	&=& 0
\end{eqnarray*}
where the first equality follows the binding obedience constraint $OB_L$ in~\eqref{obl}: $\pi(l\vert H) = C\pi(l\vert L)$ and the second follows the expression of $C$ in~\Cref{notation}. Hence, the bargaining power parameter doesn't matter for the equilibrium payoff when $\mu_0>v_L/v_H$.

In the case of $\mu_0<v_L/v_H$, the intermediary's total payoff reads as
\begin{eqnarray*}
	&&\underbrace{[V-(x_H-H)^2+\pi(l\vert H)\Delta v -\Delta v]}_{\text{service fee for type-H consumer}}\mu_0 + [V-(x_L-L)^2](1-\mu_0) \\
	&&\hspace{2in} + \beta [(1-\mu_0)\pi(l\vert L) v_L - \mu_0\pi\left(l\vert H\right) \Delta v - (v_L-\mu_0v_H)]
\end{eqnarray*}
There is a tradeoff of increasing the probability of charging a low price for the high type, $\pi(l\vert H)$. As $\pi(l\vert H)$ increases, on one hand, the service revenue increases by $\mu_0\Delta v$ because consumers are more likely to reveal personal information. On the other hand, the data revenue decreases by $\beta\mu_0\Delta v$ because the producer charges the type-H consumers a low price with a higher probability. The first force outweighs the second given the partial bargaining power $\beta <1$. So the intermediary will increase $\pi(l\vert H)$ up to one.\footnote{Recall that when $\beta=1$, the two forces cancel each other out and hence $\pi(l|H)$ is free to have multiple solutions.} Moreover, the intermediary sets $\pi(l\vert L)=1$ to maximize the data revenue. Hence, the optimal data revenue is zero because $(1-\mu_0)\pi(l\vert L) v_L - \mu_0\pi\left(l\vert H\right) \Delta v - (v_L-\mu_0v_H) = 0$. The bargaining power $\beta<1$ does not influence the intermediary's equilibrium total payoff, which is $V$.

\paragraph{Regulation of Power.} Recall that the consumer surplus is $\max\{0,\mu_0\Delta v \mathbb{I}_{\mu_0v_H<v_L} \}$ and the producer surplus is $\max \{v_L,\mu_0v_H \}$. So the partial bargaining power of the intermediary does not affect the consumer surplus, producer profit, nor the total welfare, compared to the full bargaining power.
In sum, the introduction of the partial bargaining power yields the following insights about the regulation of intermediary's power in the data market.

\begin{prop}[Regulation of intermediary's power]
	
	Reducing the intermediary's bargaining power in the data market can decrease the amount of information traded by the intermediary, without hurting social welfare.
	
		
	
\end{prop}

\section{Discussion}\label{discussion}

\paragraph{Other Off-path Beliefs.} The equilibrium selection criterion restricts the off-equilibrium-path belief to be $\mu_0$, i.e., the producer does not update from the consumer's non-participation. This is a restriction of ``no signaling what you don't know" as in~\citet{fudenberg1991perfect}. Alternatively, for example, the producer may treat the consumers as high type if they reject the intermediary's contract, i.e., the off-path belief is $1$. Then the producer sets the product price $p_0=v_H$ and $R_0=v_H$, the type-H consumers' outside option is $0$.
Recall in the analysis of the off-path belief being $\mu_0$, when $\mu_0>v_L/v_H$, the producer sets product price at $v_H$ and the revenue is $\mu_0v_H$, the type-H consumers' outside option is $0$. So this particular off-path belief changes the right-hand side of the producer's participation constraint~\eqref{irp} from $\mu_0 v_H$ to $v_H$, and the other constraints do not change. Hence, the results on data trade in the case of $\mu_0>v_L/v_H$ won't change, while the data fee will change.  Off-path belief restrictions are also seen in related models like \citet{lizzeri1999information}.

\paragraph{Competition between Data Buyers.} In the stylized model, there is one data buyer, i.e., the producer. When there are two producers and both of them buy data from the intermediary, the pricing competition in the product market between the two producers might compete away the producer's profit and hence prevent the intermediary from a beneficial data trade. In this sense, the intermediary would prefer an exclusive data trade, selling data only to one producer, which goes back to the baseline framework of one data buyer.  \citet{bimpikis2019information} and \citet{bonatti2024selling} study information sales to buyers who compete in a downstream market. \citet{armstrong2022consumer} study the intermediary designed information based on which the information users compete in the product price. 

\paragraph{Commitment of the Intermediary.} Recall that our work focuses on the cross-market incentive instead of the inter-temporal one. If one interprets our results from the inter-temporal viewpoint, there can be a commitment problem of the intermediary under some conditions. Take the result of $\mu_0<v_L/v_H$ in~\Cref{cor_notfull} as an example, when $\delta^2 < \Delta v$, the data trade ex ante committed by the optimal mechanism can be inconsistent with what will happen after consumers report the true types in the service market. The optimal mechanism will never sell full information when $\delta^2 < \Delta v$, but the intermediary has incentives to sell all the information to get the maximum data trade gain.\footnote{In the case of $\mu_0<v_L/v_H$, the intermediary’s commitment is not a problem when $\delta^2\geq \Delta v$: the full information data trade is optimal if $\delta^2\geq \Delta v$ and hence there will be no time-inconsistent behaviors of the intermediary after collecting the data.} \citet{bester2001contracting}, \citet{doval2022mechanism} and \citet{doval2025purchase} study the mechanism with the principal's limited commitment.

\paragraph{Transparency of Contract.} The paper assumes consumers observe everything about the data contract offered to the data-using producer. It is not an unreasonable assumption because transparency is one of the key aspects of data regulations, i.e., businesses must be transparent about their data collection and sale practices.\footnote{For example, the CCPA secures privacy rights for California consumers, including the right to know about the personal information a business collects about them and how it is used and shared. The GDPR in Europe also requires firms to disclose how they use data, apart from granting consumers more control over data. See \citet{lizzeri1999information} for a similar observability assumption.}
One interesting special scenario is that in the contract the intermediary commits to a binary data policy (share full information or not), which, however, is not optimal for the intermediary from the second half of part 2 in~\Cref{prop_general_service_largemu0} and hence underlines the importance of a flexible data trade policy.

\paragraph{Continuum of Values.} We assume binary product values. One could expect the continuous distributions of product values would affect consumer surplus. For the value distributions $F_L(v),F_H(v)\in \Delta [v_L,v_H]$, the producer sets a price between $v_L$ and $v_H$ rather than $v_L$ or $v_H$ as the product pricing problem becomes $\max_{p} p[1-F_H(p)]\mu_s + p[1-F_L(p)](1-\mu_s)$. In the setup of binary values with $\mu_0>v_L/v_H$, only high type consumers may get a positive surplus in the product market. But either type of consumer can get a positive one for continuous product values.

\paragraph{Horizontal-vertical Setup.} We model the data-sourcing market as horizontal, while the data-using market is modeled as vertical. The horizontal part endows consumers with incentives to reveal personal information to match their tastes, while the assumption of the vertical element allows the producer to use consumers' data against their interests in the data-using market. If both markets are modeled as horizontal, then cutoffs for the intermediary's optimal data trade would change. If both are modeled as vertical, then consumers may have no incentives to share information in the data-sourcing market.

\section{Conclusion}\label{conclusion}

This paper develops an integrated framework of the intermediary’s data sourcing from consumers and data trade to a third party who uses the data. Consumers benefit from revealing information about their tastes to the intermediary for better-matched services. A third party buys data from the intermediary so that it can extract more surplus from consumers in the product market. 
We characterize the intermediary optimal mechanisms for data sourcing and sales. Due to the cross-market feature of data, the equilibrium data trade depends on the consumers' {\it horizontal} differentiation level in the data-sourcing service market and the {\it vertical} one in the product market. 

In the intermediary optimal mechanism, the intermediary sells no information or partial information to preserve profits in the service market when consumers' mismatch cost from a wrong service is low. The low mismatch cost itself cannot sufficiently incentivize consumers to reveal information. The intermediary has to compensate consumers more (by decreasing service fees) to sell full information. However, the loss in service profit outweighs the net gain from selling full information rather than partial or no information. Hence, the intermediary would prefer to sell partial or no information. 
In contrast, the intermediary can sell full information to the data buyer when consumers' mismatch cost in service is high. In this case, consumers' privacy protection is at the lowest level. Moreover, the provisions of service, product, and data by the optimal mechanism can be efficient under certain market conditions.

Additionally, we introduce the intermediary's partial bargaining power in the data market. The partial bargaining power does not affect the surpluses of the intermediary, producer, or consumers, although it can decrease the amount of information traded by the intermediary with the third party.

Applying the framework to examine the data market regulations, we find that banning data trade may reduce social welfare because it makes it harder to price discriminate in the product market. 
Instead, regulating the intermediary's bargaining power can limit the amount of information traded by the intermediary without hurting social welfare.
The finding suggests that the regulation of the intermediary's power is more efficient than the regulation of the intermediary's data sharing.



\clearpage

\appendix

\section{Appendix: Omitted Proofs}

\subsection{Proof of~\Cref{lem_binary}}
\begin{proof}
	The basic idea is to merge the signals that induce the same posterior \citep{bergemann2016bayes}. Let $\hat{\pi}(l \vert \theta)=\sum_{s:\mu_{s}<v_{L} / v_{H}} \pi(s \vert \theta)$ and $\hat{\pi}(h \vert \theta)=\sum_{s: \mu_{s} \geq v_{L} / v_{H}} \pi(s \vert \theta)$. Then, the data trade is captured by $(\{l, h\}, \hat{\pi})$. In the content after \Cref{lem_binary}, we use the notation of $(\{l, h\}, \pi)$ for the binary direct data trade $(\{l, h\}, \hat{\pi})$ without confusions.
\end{proof}

\subsection{Proof of \Cref{lem_2service}}\label{pf_2service}
For a general service space $\mathcal X=\mathbb{R}$, we have results in the following \Cref{lem_service}. \Cref{lem_2service} directly follows \Cref{lem_service} with $\mathcal X=\{L,H\}$.

\begin{lem}\label{lem_service}
	Suppose a general service space $\mathcal X$. In the intermediary optimal mechanism, the service provided to type-H consumers $x_H\geq H$ and to type-L consumers $x_L\leq L$.
\end{lem}

\begin{proof}[Proof of \Cref{lem_service}]
To make the proof easy to read, we write out explicitly the intermediary’s problem below. That is, $\max\limits_{x,f,\pi,T} f_H \mu_0  + f_L \left(1-\mu_0\right) + T$ subject to {\small
	\begin{align}
		& V - \left( x_H - H \right)^2 - f_H + \pi\left(l\vert H\right)\Delta v
		\geq V - \left( x_L - H \right)^2 - f_L + \pi\left(l\vert L\right) \Delta v  \tag{$IC_H$ } \\
		& V - \left( x_L - L \right)^2  - f_L 
		\geq V -  \left( x_H - L \right)^2  - f_H  \tag{$IC_L$}  \\                                                              
		& V - \left( x_H - H \right)^2  - f_H + \pi\left(l\vert H\right) \Delta v \geq \max\left\{v_H - p_0, 0 \right\} \tag{$IR_H$} \\
		& V -  \left( x_L - L \right)^2 - f_L  \geq 0 \tag{$IR_L$}  \\
		& \mu_0 \left( 1-\pi\left(l\vert H\right) \right) \geq \frac{v_L}{v_H} \left[ \mu_0 \left(1-\pi\left(l\vert H\right) \right) +\left(1-\mu_0\right) \left( 1-\pi\left(l\vert L\right) \right)  \right] \tag{$OB_H$} \\
		& \frac{v_L}{v_H} \left[ \mu_0 \pi\left(l\vert H\right) +\left(1-\mu_0\right)\pi\left(l\vert L\right)  \right] \geq \mu_0 \pi\left(l\vert H\right) \tag{$OB_L$}  \\
		& R\left(\pi\right) - R_0 \geq T \tag{$IR_S$}  \\
		& \pi\left(l\vert L\right),\pi\left(l\vert H\right)\in\left[0,1\right]  \tag{$FE$}
	\end{align} }
	
	First, $x_H\geq H$ in the optimal mechanism. To see this by contradiction. Suppose $x_H< H$, then $\left(x_H - H\right)^2>0$. Increasing $f_H$ while decreasing $\left(x_H - H\right)^2$ by equal size will not affect $IC_H$~\eqref{ich} $IR_H$~\eqref{irh} and $IR_L$~\eqref{irl}. This change does not affect $IC_L$~\eqref{icl} either. To see this, there are two cases to consider: $x_H < L$ and $L < x_H < H$. Case 1: when $x_H < L$, decreasing $\left(x_H - H\right)^2$ means increasing $x_H$, which will decrease
	\begin{equation*}
		V - \left(x_H-L\right)^2 - f_H = V - \left( x_H - H\right)^2 - f_H - \delta^2 - 2\delta \left(x_H-H\right),
	\end{equation*} 
    because changes in $-\left( x_H - H\right)^2$ and $-f_H$ are canceled out, and $- 2\delta \left(x_H-H\right)$ decreases. So $IC_L$~\eqref{icl} is not affected. Case 2: when $L < x_H < H$, decreasing $\left(x_H - H\right)^2$ means increasing $x_H$, which will increase $\left(x_H - L \right)^2$. Then $IC_L$~\eqref{icl} obviously holds in this case. 
    
    Thus, increasing $f_H$ while decreasing $\left(x_H - H\right)^2$ by equal size will not affect $IC_L$~\eqref{icl} and other constraints while increasing the intermediary's payoff. Hence, $x_H< H$ is not optimal. In the optimality, it must be $x_H\geq H$.

	Second, in the optimal mechanism $x_L \leq L$. Suppose $x_L> L$, then $\left(x_L - L\right)^2>0$. Increasing $f_L$ while decreasing $\left(x_L - L\right)^2$ by equal size will not affect $IC_L$~\eqref{icl} $IR_H$~\eqref{irh} and $IR_L$~\eqref{irl}. It does not affect $IC_H$~\eqref{ich} either. To see this, decreasing $\left(x_L - L\right)^2$ means decreasing $x_L$ which will decrease  
	\begin{equation*}
		V - \left(x_L-H\right)^2 - f_L = V - \left(x_L-L\right)^2  - f_L- \delta^2 - 2\delta \left(x_L-L\right)
	\end{equation*}
    because changes in $- \left(x_L-L\right)^2$ and $- f_L$ are canceled out, and $- 2\delta\left(x_L-L\right) $ decreases as $x_L$ decreases. So increasing $f_L$ while decreasing $\left(x_L - L\right)^2$ by equal size will not affect $IC_H$~\eqref{ich} and thus not affect all the constraints but increase the intermediary's payoff. Hence, $x_L>L$ cannot be optimal.
    
    In particular, with $\mathcal X = \{H,L\}$, $x_L=L$ and $x_H=H$ which proves~\Cref{lem_2service}. 
    
\end{proof}

\subsection{Proof of~\Cref{prop_binary}}\label{pf_prop_binary}

Recall the obedience constraints in~\eqref{obh_0} and~\eqref{obl_0} are 
\begin{eqnarray}
	\frac{v_L}{v_H}\left( 1- \mu_0 \right) \pi\left(l\vert L\right) - \left( 1- \frac{v_L}{v_H} \right)\mu_0 \pi\left(l\vert H\right)  &\geq& \frac{v_L}{v_H} - \mu_0  \label{obh} \\
	\frac{v_L}{v_H}\left( 1- \mu_0 \right) \pi\left(l\vert L\right) - \left( 1- \frac{v_L}{v_H } \right) \mu_0 \pi\left(l\vert H\right) &\geq& 0 \notag
\end{eqnarray}
where the second inequality is shortened by~\eqref{obl}.	

\subsubsection{Proof of Part 1 of~\Cref{prop_binary}: $\mu_0 > v_L/v_H$}
		
When $\mu_0 > \frac{v_L}{v_H}$, the obedience constraint~\eqref{obh} is redundant, we can ignore~\eqref{obh} and consider~\eqref{obl} only. The type-H consumers' outside option $\max\left\{v_H - p_0, 0\right\}=0$ as $p_0=v_H$. So the intermediary's problem is simplified into
	\begin{align*}
		&\max_{f,\pi } \text{ }
		f_H \mu_0  + f_L \left(1-\mu_0\right) 
		- \mu_0\pi\left(l\vert H\right)\left(v_H - v_L \right) + \left(1-\mu_0\right)\pi\left(l\vert L\right) v_L  \\
	\text{s.t.} &~\eqref{ich}~\eqref{icl}~\eqref{irl}~\eqref{fesi}~\eqref{obl}~\eqref{irh_largemu}
    \end{align*}
    The constraints are rewritten as follows
    \begin{align*}
		&f_H - f_L \leq \left( H - L \right)^2  - \left[ \pi\left(l\vert L\right) - \pi\left(l\vert H\right)\right] \left( v_H - v_L \right)\tag{$IC_H$}  \\
		&- \left( H - L \right)^2 \leq f_H - f_L \tag{$IC_L$} \\                                                              
		&f_H \leq V + \pi\left(l\vert H\right) \left( v_H - v_L \right)  \tag{$IR_H$}\\
		&f_L  \leq V \tag{$IR_L$}\\
		&(1- \frac{v_L}{v_H }) \mu_0 \pi\left(l\vert H\right) \leq \frac{v_L}{v_H}\left( 1- \mu_0 \right) \pi\left(l\vert L\right) \tag{$OB_L$}\\
		& \pi\left(l\vert L\right),\pi\left(l\vert H\right)\in\left[0,1\right] \tag{$FE$}
	\end{align*}
	
	Note that at least one of $IR_L$~\eqref{irl} and $IR_H$~\eqref{irh_largemu} is binding. Otherwise, increasing $f_L$ and $f_H$ simultaneously by equal size will not affect $IC_H$~\eqref{ich} and $IC_L$~\eqref{icl} and will increase the intermediary's profit. The following Lemma excludes the case of only $IR_H$~\eqref{irh_largemu} being binding.
	
	\begin{lem}
		If $IR_H$~\eqref{irh_largemu} is binding, then $IR_L$~\eqref{irl} is also binding.
	\end{lem}
	\begin{proof}
		Given the binding $IR_H$: $f_H = V + \pi\left(l\vert H\right)\Delta v $, then at least one of $IC_L$ and $IR_L$ is binding. Otherwise, increasing $f_L$ increases the objective without affecting constraints $IC_H, IC_L, IR_L$. Suppose $IC_L$ is binding: $f_H - f_L = -\delta^2$. From binding $IR_H$ and $IC_L$, $f_L = V + \pi\left(l\vert H\right)\Delta v + \delta^2> V$ which violates $IR_L$. Hence, $IC_L$ is not binding and $IR_L$ must be binding.
	\end{proof}
	
	Hence, there are two cases following the above Lemma: $(1)$ only $IR_L$ is binding; $(2)$ both $IR_L$ and $IR_H$ are binding. Now we show~\Cref{prop_binary} by three steps: Step 1 shows that it's not optimal that only $IR_L$ is binding; Step 2 simplifies the intermediary's problem with binding $IR_L$ and $IR_H$; Step 3 solves the problem.
	
	\paragraph{Step 1.} Suppose $IR_L$ is binding and $IR_H$ is not binding. As $IR_H$ is not binding, $IC_H$ must be binding. Otherwise, increasing $f_H$ will not affect $IR_H$, $IC_H$, $IC_L$, $IR_L$ while increases the intermediary's profit. So $IC_H$ must be binding: $f_H - V = \delta^2 - \left[ \pi\left(l\vert L\right) - \pi\left(l\vert H\right) \right]\Delta v$. From $IC_L$ 
	\begin{equation*}
		-\delta^2 
		\leq f_H - V = \delta^2 - \left[ \pi\left(l\vert L\right) - \pi\left(l\vert H\right) \right]\Delta v 
	\end{equation*} 
	and from (non-binding) $IR_H$
	\begin{equation*}
		f_H - V = \delta^2 - \left[ \pi\left(l\vert L\right) - \pi\left(l\vert H\right) \right]\Delta v \leq \pi\left(l\vert H\right)\Delta v
	\end{equation*} 
	so $\pi\left(l\vert L\right) - \pi\left(l\vert H\right) \leq \frac{2\delta^2}{\Delta v}$ and $\pi\left(l\vert L\right) \geq \frac{\delta^2}{\Delta v}$. 
	
	Plugging the binding $IR_L$ and $IC_H$: $f_L = V$ and $ f_H = V + \delta^2 - \left[ \pi\left(l\vert L\right) - \pi\left(l\vert H\right) \right]\Delta v$ into the objective, ignoring the nonbinding $IR_H$, the problem is to maximize
	\begin{eqnarray*}
		&& V + \mu_0 \delta^2 
		-\mu_0 \left[ \pi\left(l\vert L\right) - \pi\left(l\vert H\right) \right]\Delta v - \mu_0\pi\left(l\vert H\right)\left(v_H - v_L \right) + \left(1-\mu_0\right)\pi\left(l\vert L\right) v_L \\
		&=& V + \mu_0 \delta^2 + \pi\left(l\vert L\right) \left[ v_L - \mu_0 v_H \right] 
	\end{eqnarray*}
	subject to $\pi\left(l\vert L\right) - \pi\left(l\vert H\right) \leq \frac{2\delta^2}{\Delta v}$. The solution is $\pi\left(l\vert L\right)=0$ because $v_L - \mu_0 v_H < 0$. And $\pi\left(l\vert H\right)=0$ by $OB_L$: $\left( 1-\frac{v_L}{v_H} \right)\mu_0 \pi\left(l\vert H\right) \leq \frac{v_L}{v_H}\left(1-\mu_0\right) \pi\left(l\vert L\right)$. But the nonbinding $IR_H$ requires $\pi\left(l\vert L\right) \geq \frac{\delta^2}{\Delta v}$, which does not hold at $\pi\left(l\vert L\right)=0$. Hence, it's not feasible that $IR_L$ is binding while $IR_H$ is not binding. 
	
	\paragraph{Step 2.} Following the first step, it must be that both $IR_L$ and $IR_H$ are binding: $f_H = V + \pi\left(l\vert H\right) \Delta v$, $f_L =V$. Then from $IC_L$ 
	\begin{equation*}
		-\delta^2 \leq \pi\left(l\vert H\right) \Delta v 
	\end{equation*}
	and from $IC_H$
	\begin{equation*}
		\pi\left(l\vert L\right) \Delta v - \delta^2 \leq 0
	\end{equation*}
	so $\pi\left(l\vert L\right)\leq \frac{\delta^2}{\Delta v}$. By plugging the binding $IR_L$ and $IR_H$ into the objective, the intermediary chooses $\pi\left(l\vert L\right)$ and $\pi\left(l\vert H\right)$ to maximize 
	\begin{equation*}
		  V + \mu_0 \pi\left(l\vert H\right) \Delta v - \mu_0\pi\left(l\vert H\right)\Delta v
		+ \left(1-\mu_0\right)\pi\left(l\vert L\right) v_L
	\end{equation*}
	subject to $OB_L$, $\pi\left(l\vert L\right)\leq \frac{\delta^2}{\Delta v}$ and $\pi\left(l\vert L\right),\pi\left(l\vert H\right)\in\left[0,1\right]$. 
	
	\paragraph{Step 3.} Hence, the solution depends on $\frac{\delta^2}{\Delta v}$. When $\frac{\delta^2}{\Delta v}\geq 1$, $\pi\left(l\vert L\right) = 1 $ and $\pi\left(l\vert H\right)$ is arbitrary if it satisfies $OB_L$: $\frac{v_L}{v_H}\left( 1- \mu_0 \right) \geq ( 1- \frac{v_L}{v_H } ) \mu_0 \pi\left(l\vert H\right)$. The intermediary's payoff is $V + \left(1-\mu_0\right)v_L $. 
	
    When $\frac{\delta^2}{\Delta v} < 1$, $\pi\left(l\vert L\right) = \frac{\delta^2}{\Delta v}$ and $\pi\left(l\vert H\right)$ is arbitrary if it satisfies $OB_L$: $\pi(l\vert H)\leq \frac{C \delta^2}{\Delta v}$. The corresponding intermediary payoff is $V + \left(1-\mu_0\right)v_L \Delta v$. 
		
\subsubsection{Proof of Part 2 of~\Cref{prop_binary}: $\mu_0 <  v_L/v_H$}
The detailed results of the part 2 of \Cref{prop_binary} are described below:
Suppose $\mu_0 <v_L/v_H$. In the optimal mechanism, the intermediary sets service fees to extract all the consumers surplus in the service and product markets, $f_L =V$, $f_H = V - \left( 1-\pi\left(l\vert H\right) \right) \Delta v$
	and the optimal data trade depends on the cross-market comparison between the mismatch cost $\delta ^2$ and the vertical differentiation $\Delta v$: $\pi\left(l\vert L\right) = 1$, $\pi\left(l\vert H\right) \in [\max\{1-\frac{\delta^2}{\Delta v}, 0\},1]$
	The data trade fee is $T =  \mu_0\Delta v \left( 1 -\pi\left(l\vert H\right) \right)$.
	And social welfare is $V + v_L + \mu_0\Delta v$.

\begin{proof}

When $\mu_0<\frac{v_L}{v_H}$, the obedience constraint for the producer to charge a low price~\eqref{obl} is redundant and we consider~\eqref{obh} only.
Moreover, the participation constraint of the type-H consumers~\eqref{irh} is
\begin{eqnarray}
	V - \left( x_H - H \right)^2 - f_H
	+ \left( v_H - v_L \right) \pi\left( l \vert H \right)
	\geq \Delta v  \label{irh_smallmu}
\end{eqnarray}

Without data trade the producer will charge $v_L$ and all consumers purchase in the product market. The producer's revenue is $R_0= v_L$ without data trade. The net gain in the product revenue from data trade in~\eqref{netgain} becomes	
\begin{eqnarray*}
	R\left(\pi \right) - R_0 
	= - \mu_0\left[ \pi\left(l\vert H\right) - 1 \right]\left(v_H - v_L \right) + \left(1-\mu_0\right)\left[ \pi\left(l\vert L\right) -1\right] v_L 
\end{eqnarray*}

The intermediary's problem is simplified into
	\begin{align*}
		&\max_{f,\pi } \text{ } f_H \mu_0  + f_L \left(1-\mu_0\right) 
		- \mu_0\pi\left(l\vert H\right)\left(v_H - v_L \right) + \left(1-\mu_0\right)\pi\left(l\vert L\right) v_L + \mu_0 v_H - v_L \\
		&\text{s.t. \eqref{ich} \eqref{icl} \eqref{irl} \eqref{fesi} \eqref{obh} \eqref{irh_smallmu} }
	\end{align*}
Note that at least one of $IR_L$~\eqref{irl} and $IR_H$~\eqref{irh_smallmu} is binding. Otherwise, increasing $f_L$ and $f_H$ simultaneously by equal size will not affect $IC_H$~\eqref{ich}, $IC_L$~\eqref{icl}, $IR_H$~\eqref{irh_smallmu}, $IR_L$~\eqref{irl}, and will increase the intermediary's payoff. In addition, the following Lemma excludes the case of only $IR_L$ being binding.
	
	\begin{lem}
		If $IR_L$~\eqref{irl} is binding, then $IR_H$~\eqref{irh_smallmu} is also binding.
	\end{lem}
	\begin{proof}
		Given $IR_L$ is binding $f_L= V$, we need to show \eqref{irh_smallmu} is binding. Suppose \eqref{irh_smallmu} is not binding. Then $IC_H$ must be binding. Otherwise, increasing $f_H$ will not affect $IC_H$, $IR_H$~\eqref{irh_smallmu}, $IC_L$ and increase the intermediary's payoff. So $IC_H$ is binding and it gives $f_H - V = \delta^2 - \left[ \pi\left(l\vert L\right) - \pi\left(l\vert H\right) \right]\Delta v$. From \eqref{irh_smallmu}
		\begin{equation*}
			f_H - V = \delta^2 - \left[ \pi\left(l\vert L\right) - \pi\left(l\vert H\right) \right]\Delta v 
			\leq \left[ \pi\left(l\vert H\right) - 1 \right] \Delta v
		\end{equation*} 
		so $\pi\left(l\vert L\right)\geq 1 + \frac{\delta^2}{\Delta v}$, which contradicts with $\pi\left(l\vert L\right)\in [0,1]$. Hence, \eqref{irh_smallmu} is binding.
	\end{proof}
	
	Hence, there are two cases: $(1)$ both $IR_L$ and $IR_H$~\eqref{irh_smallmu} are binding; $(2)$ only $IR_H$~\eqref{irh_smallmu} is binding.
	
	\paragraph{Step 1.} In this step, we show that both $IR_L$ and $IR_H$~\eqref{irh_smallmu} are binding. Suppose not, then only $IR_H$~\eqref{irh_smallmu} is binding: $f_H = V + \left[ \pi\left(l\vert H\right) -1\right] \Delta v$. As $IR_L$ is not binding, $IC_L$ must be binding. Otherwise, increasing $f_L$ will not affect $IR_L$~\eqref{irl}, $IC_L$~\eqref{icl}, $IC_H$~\eqref{ich}, and will increase the intermediary's payoff. So $IC_L$ is binding. The binding $IC_L$ and $IR_H$ \eqref{irh_smallmu} give $f_L = f_H + \delta^2 = V + \left[ \pi\left(l\vert H\right) -1\right] \Delta v + \delta^2 $. From $IC_H$ 
	\begin{equation*}
		-\delta^2 \leq  \delta^2  - \left[ \pi\left(l\vert L\right) - \pi\left(l\vert H\right)\right] \Delta v 
	\end{equation*} 
	and from (non-binding) $IR_L$
	\begin{equation*}
		V + \left[ \pi\left(l\vert H\right) -1\right] \Delta v + \delta^2 < V
	\end{equation*}
	or we rewrite $IC_H$ and $IR_L$ as
	\begin{eqnarray*}
		\pi\left(l\vert L\right) - \pi\left(l\vert H\right) &\leq& \frac{2 \delta^2}{\Delta v} \\
		\pi\left(l\vert H\right)  &\leq& 1 - \frac{\delta^2}{\Delta v} 
	\end{eqnarray*}  
    Plug $f_L$ and $f_H$ from the binding $IR_H$ and $IC_L$ into the objective and ignore the nonbinding constraint, the intermediary maximizes
    \begin{eqnarray*}
    	&& \mu_0\left[  V + \left( \pi\left(l\vert H\right) -1\right) \Delta v \right] + (1-\mu_0) \left[ V + \left( \pi\left(l\vert H\right)-1 \right) \Delta v + \delta^2  \right ] \\
    	&& + \left(1-\mu_0\right)\pi\left(l\vert L\right) v_L - \mu_0\pi\left(l\vert H\right)\Delta v +\mu_0 v_H-v_L \\
    	&=& V + \left[ \pi\left(l\vert H\right) -1\right] \Delta v + (1-\mu_0) \delta^2 + \left(1-\mu_0\right)\pi\left(l\vert L\right) v_L - \mu_0\pi\left(l\vert H\right)\Delta v +\mu_0v_H-v_L \\
    	&=&	 V + (1-\mu_0) \delta^2 -\left(1-\mu_0 \right) v_H + \left(1-\mu_0\right)\Delta v \pi\left(l\vert H\right) + \left(1 -\mu_0 \right)v_L \pi\left(l\vert L\right) 
    \end{eqnarray*}
    subject to $OB_H$ and $\pi\left(l\vert L\right) - \pi\left(l\vert H\right) \leq \frac{2\delta^2}{\Delta v}$. 
  
    The solution to this reduced problem is $\pi\left(l\vert L\right) = 1$ and $\pi\left(l\vert H\right) = 1 $. Check the nonbinding $IR_L$: $\pi\left(l\vert H\right) = 1 < 1 - \frac{\delta^2}{\Delta v} $, which does not hold. So it is not optimal that only $IR_H$~\eqref{irh_smallmu} is binding. Hence, $IR_L$ and $IR_H$~\eqref{irh_smallmu} are binding.
	
    \paragraph{Step 2.} Following the first step, it must be both $IR_L$ and $IR_H$~\eqref{irh_smallmu} are binding: $f_H = V + \left[ \pi\left(l\vert H\right) - 1 \right] \Delta v$, $f_L =V$. Then plugging $f_L$, $f_H$ into $IC_H$
	\begin{equation*}
		\left[ \pi\left(l\vert H\right) - 1\right] \Delta v \leq \delta^2 - \left[ \pi\left(l\vert L\right) - \pi\left(l\vert H\right) \right]\Delta v
	\end{equation*}
	and into $IC_L$
	\begin{equation*}
		\left[ 1 - \pi\left(l\vert H\right) \right] \Delta v - \delta^2 \leq 0
	\end{equation*}
	or
	\begin{eqnarray*}
		\pi\left(l\vert L\right) &\leq& 1 + \frac{\delta^2}{\Delta v} \\
		\pi\left(l\vert H\right) &\geq& 1 - \frac{\delta^2}{\Delta v}
	\end{eqnarray*}
	
	By plugging $f_H = V + \left[ \pi\left(l\vert H\right) - 1 \right] \Delta v$, $f_L =V$ into the objective, the intermediary chooses $\pi\left(l\vert L\right)$ and $\pi\left(l\vert H\right)$ to maximize
	\begin{eqnarray*}
		&& V + \mu_0 \left[ \pi\left(l\vert H\right) - 1 \right] \Delta v
		- \mu_0\pi\left(l\vert H\right)\Delta v
		+ \left(1-\mu_0\right)\pi\left(l\vert L\right) v_L + \mu_0 v_H - v_L \\
		&=& V + \left(1-\mu_0\right) v_L \pi\left(l\vert L\right) - \left(1-\mu_0\right)v_L
	\end{eqnarray*}
	subject to $OB_H$ and 
	\begin{eqnarray*}
		\pi\left(l\vert H\right)&\geq& 1-\frac{\delta^2}{\Delta v }\\
		\pi\left(l\vert L\right) &\leq& 1 + \frac{\delta^2}{\Delta v }
	\end{eqnarray*}
    
    \paragraph{Step 3.} The solution to the reduced problem is $\pi\left(l\vert L\right)= 1$, and $\pi\left(l\vert H\right)$ depends on $\Delta v$. When $\frac{\delta^2}{\Delta v}\geq 1$, $\pi\left(l\vert H\right)\in \left[0,1\right]$. The profit from the service is $V - \mu_0\Delta v\left(1-\pi\left(l\vert H\right)\right)$ and the profit from the data trade is $\mu_0\Delta v \left( 1 -\pi\left(l\vert H\right) \right) $. So its total payoff is $V$.
	
    When $\frac{\delta^2}{\Delta v} <1$, $\pi\left(l\vert H\right)\in \big[1-\frac{\delta^2}{\Delta v},1\big]$ and the intermediary's payoff is $V$.
	
\end{proof}

\subsection{Proof of~\Cref{cor_notfull}}\label{pf_cor_notfull}

\begin{proof}
	Consider $\mu_0>\frac{v_L}{v_H}$ and $\frac{\delta^2}{\Delta v} < \frac{\Delta v}{v_L+v_H}$ (the other case is similar).
	Suppose the data trade is of full information, $\pi(l|L)=1$ and $\pi(l|H)=0$. Then $OB_L$ holds. By plugging $f_L$ and $f_H$ into the objective, the intermediary chooses $x_H$, $x_L$, $\pi\left(l\vert L\right)$, $\pi\left(l\vert H\right)$ to maximize 
	$V - \mu_0 \left( x_H - H \right)^2 - \left(1-\mu_0\right)\left( x_L - L \right)^2
	+ \left(1-\mu_0\right) v_L $
	subject to the following $IC_L$, $IC_H$ 
	
	\begin{align*}
		\left( x_H - H \right)^2 - \left( x_H - L \right)^2 &\leq 0  \tag{$IC_L$} \\
		\left( x_L - L \right)^2 - \left( x_L - H \right)^2 &\leq - \Delta v  \tag{$IC_H$ }
	\end{align*}
    So $x_H = H$, and if $\delta^2\geq \Delta v$, then $x_L= L$; if $\delta^2 < \Delta v$, then $x_L=\frac{L+H}{2}-\frac{\Delta v}{2\delta}$. So given $\frac{\delta^2}{\Delta v} < \frac{\Delta v}{v_L+v_H}$, we have $x_L=\frac{L+H}{2}-\frac{\Delta v}{2\delta}$ and the intermediary's payoff of selling full information is 
	\begin{equation*}
		V  - \left(1-\mu_0\right)\left( \frac{\delta}{2}-\frac{\Delta v}{2\delta}  \right)^2
		+ \left(1-\mu_0\right) v_L 
	\end{equation*}
	where the service profit is $V  - \left(1-\mu_0\right)\left( \frac{\delta}{2}-\frac{\Delta v}{2\delta}  \right)^2$ and the data profit is $\left(1-\mu_0\right) v_L $.
	
	Recall that the intermediary's payoff of the optimal mechanism is 
	\begin{equation*}
		V+(1-\mu_0)\left(\frac{\delta}{\Delta v} \right)^2 v_Lv_H
	\end{equation*}
	So the difference between the two payoffs is positive
	\begin{eqnarray*}
		&&V+(1-\mu_0)\left(\frac{\delta}{\Delta v} \right)^2 v_Lv_H 
		- \left[
		V  - \left(1-\mu_0\right)\left( \frac{\delta}{2}-\frac{\Delta v}{2\delta}  \right)^2
		+ \left(1-\mu_0\right) v_L
		\right] \\
		&=&(1-\mu_0)\left(\frac{\delta}{\Delta v} \right)^2 v_Lv_H 
		+ \left(1-\mu_0\right)\left( \frac{\delta}{2}-\frac{\Delta v}{2\delta}  \right)^2
		- \left(1-\mu_0\right) v_L \\
		&\propto&\left(\frac{\delta}{\Delta v} \right)^2 v_Lv_H 
		+ \left( \frac{\delta}{2}-\frac{\Delta v}{2\delta}  \right)^2
		- v_L \\
		&\propto& [ (v_L+v_H)\delta^2 - (\Delta v)^2 ]^2 > 0
	\end{eqnarray*}
	which implies selling full information is not optimal.
	
\end{proof}

\subsection{Proof of~\Cref{eq_ban}}\label{pf_eq_ban}
\begin{proof}
	Without a data market, the intermediary maximizes its service revenue by choosing the service provisions and fees to solve
	\begin{align*}
		\max_{f,x} \text{ }& f_H\mu_0+f_L\left(1-\mu_0\right)\\
		\text{subject to } &V - \left(x_H - H\right)^2 - f_H \geq V - \left(x_L - H\right)^2 - f_L \\
		& V - \left(x_L -L\right)^2 - f_L \geq V -\left(x_H-L\right)^2 - f_H \\
		& V - \left(x_H -H\right)^2 - f_H \geq 0 \\
		& V - \left(x_L -L\right)^2 - f_L \geq 0   
	\end{align*}
	
	Under the assumption $\mathcal X = \left\{L,H\right\}$, $x_L = L$ and $x_H = H$ by~\Cref{lem_service}. Then the above constraints turn into $- f_H \geq -\delta^2 - f_L$, $- f_L \geq -\delta^2 - f_H$, $f_H \leq V$, $f_L \leq V$.
	Hence, $f_L=f_H =V$.
	
	 
\end{proof}

\subsection{Proof of~\Cref{prop_general_service_smallmu0}}\label{general_service_smallmu0}

\begin{proof}
	When $\mu_{0}<\frac{v_{L}}{v_{H}}$, the obedience constraint $OB_{L}$ is redundant and the type-H consumers' outside option max $\left\{v_{H}-p_{0}, 0\right\}=v_{H}-v_{L}$. So the intermediary's problem is simplified into
	\begin{equation*}
		\max_{f,x,\pi,T}  f_H \mu_0  + f_L \left(1-\mu_0\right) + T
	\end{equation*}
	subject to the following constraints
	\begin{align}
		& V - \left( x_H - H \right)^2 - f_H + \pi\left(l\vert H\right) \left( v_H - v_L \right) \notag \\
		&\hspace{2cm} \geq V - \left( x_L - H \right)^2 - f_L + \pi\left(l\vert L\right) \left( v_H - v_L \right)  \tag{$IC_H$ } \\
		& V - \left( x_L - L \right)^2  - f_L 
		\geq V -  \left( x_H - L \right)^2  - f_H  \tag{$IC_L$}  \\                                                              
		& V - \left( x_H - H \right)^2  - f_H + \pi\left(l\vert H\right) \left( v_H - v_L \right)  \geq v_H - v_L \tag{$IR_H$} \\
		& V -  \left( x_L - L \right)^2 - f_L  \geq 0 \tag{$IR_L$}  \\
		& \frac{v_{L}}{v_{H}}\left(1-\mu_{0}\right) \pi\left(l \vert L\right)-\left(1-\frac{v_{L}}{v_{H}}\right) \mu_{0} \pi\left(l \vert H\right) \geq \frac{v_{L}}{v_{H}}-\mu_{0} \tag{$OB_H$}  \\
		&R\left(\pi\right) - R_0 \geq T \tag{$IR_S$} \\
		& \pi\left(l\vert L\right),\pi\left(l\vert H\right)\in\left[0,1\right]  \tag{$FE$} 
	\end{align}
	
	\begin{lem}
		If $IR_{L}$~\eqref{irl} is binding, then $IR_{H}$~\eqref{irh_smallmu} is also binding.
	\end{lem}
	\begin{proof}
		Given $IR_{L}$ is binding $f_{L}=V-\left(x_{L}-L\right)^{2}$. First, at least one of $IC_{H}$ and $IR_{H}$ is binding. Suppose not, both are not binding. Increasing $f_{H}$ will not affect $IC_{H}$, $IR_H$, $IC_{L}$ and increase the intermediary's payoff.
		
		Second, $I C_{H}$ is not binding. To see this, suppose $I C_{H}$ is binding, then it gives
		$$
		V-\left(x_{H}-H\right)^{2}-f_{H}+\pi\left(l \vert H\right)\left(v_{H}-v_{L}\right)=V-\left(x_{L}-H\right)^{2}-f_{L}+\pi\left(l \vert L\right)\left(v_{H}-v_{L}\right)
		$$
		Together with $f_{L}=V-\left(x_{L}-L\right)^{2}$, we have
		$$
		\begin{aligned}
			f_{H} &=V+\left(x_{L}-H\right)^{2}-\left(x_{H}-H\right)^{2}-\left(x_{L}-L\right)^{2}-(\pi\left(l \vert L\right)-\pi\left(l \vert H\right)) \Delta v \\
			& \leq V-\left(x_{H}-H\right)^{2}+\pi\left(l \vert H\right) \Delta v-\Delta v
		\end{aligned}
		$$
		where the inequality follows from $I R_{H}$. So
		$2 x_{L} \geq L+H,$
		leading a contradiction with $x_{L} \leq L .$ Hence, $I C_{H}$ is not binding.
		
		Thus, by the above two points, $I R_{H}$ is binding.
	\end{proof}
	
	In addition, we note that at least one of $IR_{L}$ and $IR_{H}$ is binding. Otherwise, increasing $f_{L}$ and $f_{H}$ simultaneously by equal size will not affect $IC_{H}$, $IC_{L}$, $IR_{H}$, $IR_{L}$ and will increase the intermediary's payoff.
	
	Hence, by the above lemma, there are two cases: $(1)$ both $IR_{L}$ and $IR_H$ are binding; $(2)$ only $IR_H$ is binding. The following Step 1 shows that Case (2) is not feasible by contradiction, and Step 2 solves the optimal mechanism under Case (1).
	
	\paragraph{Step 1.} We show that only $IR_H$~\eqref{irh_smallmu} is binding is not feasible. Suppose not, only~\eqref{irh_smallmu} is binding: $f_{H}=V-\left(x_{H}-H\right)^{2}+\left[\pi\left(l \vert H\right)-1\right] \Delta v$. As $I R_{L}$ is not binding, $IC_{L}$ must be binding. Otherwise, increasing $f_{L}$ will not affect $IR_{L}$, $IC_{L}$, $IC_{H}$ and will increase intermediary's payoff. So the binding~\eqref{irh_smallmu} and $IC_{L}$ give
	\begin{eqnarray*}
		f_{H}&=&V-\left(x_{H}-H\right)^{2}+(\pi\left(l \vert H\right)-1) \Delta v \\
		f_{L}&=&V-\left(x_L-L\right)^2 + \left(x_H-L\right)^2 - \left(x_H-H\right)^2 + \left(\pi\left(l\vert H\right)-1\right)\Delta v 
	\end{eqnarray*}	
	Then $IC_{H}$ is
	\begin{equation*}
		2\left(x_L-x_H\right) \leq [\pi\left(l\vert H\right) - \pi\left(l\vert L\right) ] \frac{\Delta v}{\delta}
	\end{equation*}
	and ignore (non-binding) $IR_{L}$ is
	\begin{equation*}
		2 x_H - L - H < [1-\pi\left(l \vert H\right) ]  \frac{\Delta v}{\delta}
	\end{equation*}
	The intermediary's problem is
	$$
	\begin{aligned}
		\max _{x, \pi} V-\left(x_{H}-{H}\right)^{2}&+\pi(l \vert H)\left(1-\mu_{0}\right) \Delta v+\left[\left(x_{H}-{L}\right)^{2}-\left(x_{L}-{L}\right)^{2}\right]\left(1-\mu_{0}\right) \\
		&+\left(1-\mu_{0}\right) \pi\left(l \vert L\right) v_{L}-\left(1-\mu_{0}\right) v_{H}
	\end{aligned}
	$$
	subject to $O B_{H}$ and $I C_{H}$: $2\left(x_L-x_H\right) \leq [\pi\left(l\vert H\right) - \pi\left(l\vert L\right) ] \frac{\Delta v}{\delta}$,
	and $\pi\left(l \vert L\right)$, $\pi\left(l \vert H\right) \in\left[0,1\right]$. If $I C_{H}$ is not binding and ignore the constraint, the solution is
	$$
	x_{L}=L, x_{H}=\frac{H-\left(1-\mu_{0}\right) L }{\mu_{0}}, \pi\left(l \vert L\right)=\pi\left(l \vert H\right)=1
	$$
	which satisfies $IC_{H}$ but doesn't satisfy $IR_{L}$. Hence, $IC_{H}$ is binding: 
	\begin{equation}
		2\left(x_L-x_H\right) = [\pi\left(l\vert H\right) - \pi\left(l\vert L\right) ] \frac{\Delta v}{\delta}.\label{ich_binding}
	\end{equation}
	Plugging~\eqref{ich_binding} into the intermediary's objective function, it maximizes
	\begin{align*}
		&V-\left(x_{H}-H\right)^{2}+\pi\left(l \vert H\right)\left(1-\mu_{0}\right) \Delta v 
		  +\left(1-\mu_{0}\right) \pi\left(l \vert L\right) v_{L}  \\
		& +\Big[\left(x_{H}-L\right)^{2}-\Big(\frac{\pi\left(l \vert H\right)-\pi\left(l \vert L\right)}{2} \frac{\Delta v}{\delta} + x_{H} -L \Big)^{2}\Big]\left(1-\mu_{0}\right) 
		-\left(1-\mu_{0}\right) v_{H} 
	\end{align*}
	subject to $OB_H$, $\pi\left(l \vert L\right)\in\left[0,1\right]$, $\pi\left(l \vert H\right) \in\left[0,1\right]$. The first order derivative with respect to $\pi(l\vert H)$ is positive, so $OB_H$ is binding: $\pi(l\vert H)
	=\frac{\frac{v_L}{v_H}(1-\mu_0)\pi(l\vert L) - (\frac{v_L}{v_H}-\mu_0)}{(1-\frac{v_L}{v_H})\mu_0}$.

	Suppose $\pi(l\vert L)\in [0,1]$ is not binding. The first order condition with respect to $\pi(l\vert L)$ and $x_{H}$ are respectively
	\begin{eqnarray*}
		x_H - L &=& \frac{\delta v_L}{v_L-\mu_0v_H} - \frac{\pi(l\vert H)-\pi(l\vert L)}{2} \frac{\Delta v}{\delta} \\
		x_{H}-H &=&-\left(1-\mu_{0}\right)\frac{ \pi\left(l \vert H\right)-\pi\left(l \vert L\right) }{2} \frac{\Delta v}{\delta}
	\end{eqnarray*}
	from which $H - (1-\mu_0) L - \frac{(1-\mu_0)\delta v_L}{v_L-\mu_0v_H} = \mu_0x_H$ and then
	\begin{eqnarray*}
		\mu_0(x_H - H) &=& H - (1-\mu_0) L - \frac{(1-\mu_0)\delta v_L}{v_L-\mu_0v_H} - \mu_0 H \\
		&=&(1-\mu_0)\delta \frac{-\mu_0v_H}{v_L-\mu_0v_H} < 0
	\end{eqnarray*}
	which contradicts with $x_H\geq H$. Hence, $\pi(l\vert L)=1$ is binding: $\pi(l\vert L)=1$. And $\pi(l\vert H)=1$ by the binding $OB_H$, then $x_H=H$. By~\eqref{ich_binding}, $x_L=x_H= H$ which is a contradiction with $x_L\leq L$. Hence, it's not feasible that only $I R_{H}$ is binding.
	
	\paragraph{Step 2.} Following Step 1, it must be both $I R_{L}$ and $I R_{H}$ are binding:
	\begin{eqnarray*}
		f_{H}&=&V-\left(x_{H}-H\right)^{2}+\pi\left(l \vert H\right) \Delta v-\Delta v \\
		f_{L}&=&V-\left(x_{L}-L\right)^{2}
	\end{eqnarray*}
	Then from $I C_{H}$
	\begin{equation*}
		\Delta v \geq-\left(x_{L}-H\right)^{2}+\left(x_{L}-L\right)^{2}+\pi\left(l \vert L\right) \Delta v
	\end{equation*}
	and from $I C_{L}$
	\begin{equation*}
		0 \geq-\left(x_{H}-L\right)^{2}+\left(x_{H}-H\right)^{2}-\pi\left(l \vert H\right) \Delta v+\Delta v
	\end{equation*}
	equivalently, they are
	\begin{eqnarray*}
		1-\pi\left(l \vert L\right) & \geq& \left( 2 x_{L}- L - H\right)\frac{\delta}{\Delta v}\\
		1-\pi\left(l \vert H\right) & \leq& \left( 2 x_{H}- L - H\right)\frac{\delta}{\Delta v}
	\end{eqnarray*}
	By plugging $f_{H}$ and $f_{L}$ into the objective, the intermediary chooses $\pi\left(l \vert L\right)$ and $\pi\left(l \vert H\right)$ to maximize
	$$
	V-\mu_{0}\left(x_{H}-H\right)^{2}-\left(1-\mu_{0}\right)\left(x_{L}-L\right)^{2}+\left(1-\mu_{0}\right) \pi\left(l \vert L\right) v_{L}-v_{L}\left(1-\mu_{0}\right)
	$$
	subject to $O B_{H}$ and $\pi\left(l \vert L\right)$, $\pi\left(l \vert H\right)\in\left[0,1\right]$ and
	\begin{eqnarray*}
		1-\pi\left(l \vert L\right) & \geq& \left( 2 x_{L}- L - H\right)\frac{\delta}{\Delta v}\\
		1-\pi\left(l \vert H\right) & \leq& \left( 2 x_{H}- L - H\right)\frac{\delta}{\Delta v}
	\end{eqnarray*}
	Hence, $x_{L}=L$, $x_{H}=H$ and $\pi\left(l \vert L\right)=1$ and $\pi\left(l \vert H\right) \in\left[\max\left\{0,1-\frac{\delta^2}{\Delta v } \right\},1\right]$.
	
	In sum, when $\mu_{0}<\frac{v_{L}}{v_{H}}$, $x_L=L$, $x_H=H$, $f_L=V$, $f_H=V-\left[1-\pi(l\vert H)\right]\Delta v$, $\pi\left(l \vert L\right)=1$, $\pi\left(l \vert H\right) \in\left[\max\left\{0,1-\frac{\delta^2}{\Delta v } \right\},1\right]$.
	
\end{proof}

\subsection{Proof of~\Cref{prop_general_service_largemu0} }\label{pf_general_service_largemu0}

\begin{proof}
	When $\mu_0 > \frac{v_L}{v_H}$, 
	the obedience constraint $OB_H$~\eqref{obh} is redundant, we ignore~\eqref{obh} and consider only~\eqref{obl} for obedience constraints. The type-H consumers' outside option $\max\{v_H - p_0, 0\}=0$. So the intermediary's problem is simplified into
	\begin{align*}
		&\max_{f,\pi } \text{ }
		f_H \mu_0  + f_L \left(1-\mu_0\right) 
		- \mu_0\pi\left(l\vert H\right)\left(v_H - v_L \right) + \left(1-\mu_0\right)\pi\left(l\vert L\right) v_L  \\
		\text{s.t. }& ~\eqref{ich}~\eqref{icl}~\eqref{irl}~\eqref{fesi} ~\eqref{obl} ~\eqref{irh_largemu}
	\end{align*}
	which are
	\begin{align*}
		& V - \left( x_H - H \right)^2 - f_H + \pi\left(l\vert H\right) \Delta v
		\geq V - \left( x_L - H \right)^2 - f_L + \pi\left(l\vert L\right) \Delta v  \tag{$IC_H$ } \\
		& V - \left( x_L - L \right)^2  - f_L 
		\geq V -  \left( x_H - L \right)^2  - f_H  \tag{$IC_L$}  \\                                                              
		& V - \left( x_H - H \right)^2  - f_H + \pi\left(l\vert H\right) \Delta v  \geq 0 \tag{$IR_H$} \\
		& V -  \left( x_L - L \right)^2 - f_L  \geq 0 \tag{$IR_L$}  \\
		& \pi\left(l\vert H\right) \leq C\pi\left(l\vert L\right)  \tag{$OB_L$}  \\
		& \pi\left(l\vert L\right),\pi\left(l\vert H\right)\in\left[0,1\right] \tag{$FE$} 
	\end{align*}
	Note that at least one of $IR_L$~\eqref{irl} and $IR_H$~\eqref{irh_largemu} is binding. Otherwise, both of $IR_L$~\eqref{irl} and $IR_H$~\eqref{irh_largemu} are not binding. So increasing $f_L$ and $f_H$ simultaneously by equal size will not affect $IC_H$~\eqref{ich} and $IC_L$~\eqref{icl}, and will increase the intermediary's profit. The following Lemma excludes the case of only $IR_H$~\eqref{irh_largemu} being binding.
	
	\begin{lem}
		If $IR_H$~\eqref{irh_largemu} is binding, then $IR_L$~\eqref{irl} is also binding.
	\end{lem}
	\begin{proof}
		Given the binding $IR_H$: $V - (x_H-H)^2 - f_H + \pi\left(l\vert H\right)\Delta v = 0$, then at least one of $IC_L$ and $IR_L$ is binding. Otherwise, increasing $f_L$ increases the objective without affecting constraints $IC_H$, $IC_L$, $IR_L$. 
		
		Suppose $IC_L$ is binding: $V-(x_L-L)^2-f_L=V-(x_H-L)^2-f_H$. From the binding $IR_H$ and $IC_L$,
		\begin{eqnarray*}
			&&V - (x_L - L)^2 - f_L \\
			&=&V - (x_H - L)^2 - V + (x_H-H)^2 - \pi\left(l\vert H\right)\Delta v \\
			&=&(L-H)(2x_H-L-H) - \pi\left(l\vert H\right)\Delta v  \\
			&<&0 
		\end{eqnarray*}	
		which violates $IR_L$. Hence, $IC_L$ cannot be binding, and $IR_L$ must be binding.
	\end{proof}
	
	Hence, there are two cases following the above Lemma: $(1)$ $IR_L$ is binding, $IR_H$ is not binding; $(2)$ both $IR_L$ and $IR_H$ are binding. Now we show~\Cref{prop_binary} by three steps: Step 1 shows that it's not optimal that only $IR_L$ is binding by contradiction; Step 2 simplifies the intermediary's problem with binding $IR_L$ and $IR_H$; Step 3 solves the problem.
	
	\paragraph{Step 1.} Suppose $IR_L$ is binding, $IR_H$ is not binding. As $IR_H$ is not binding, $IC_H$ must be binding. Otherwise, increasing $f_H$ will not affect $IR_H$, $IC_H$, $IC_L$, $IR_L$ while increases the intermediary's profit. So $IC_H$ is binding:
	\begin{equation}
		V - \left( x_H - H \right)^2 - f_H + \pi\left(l\vert H\right) \Delta v
		=  V - \left( x_L - H \right)^2 - f_L + \pi\left(l\vert L\right) \Delta v  \label{bindingich_largemu0}
	\end{equation}
	also recall the binding $IR_L$: $V - (x_L - L)^2 - f_L = 0$, we have
	\begin{equation}
		f_L = V - (x_L - L)^2 \label{fl_largemu0}
	\end{equation}
	Plugging~\eqref{fl_largemu0} into~\eqref{bindingich_largemu0} to get
	\begin{equation}
		f_H = V + \left( x_L - H \right)^2 - (x_L - L)^2 - \left( x_H - H \right)^2 
		- \pi\left(l\vert L\right) \Delta v + \pi\left(l\vert H\right) \Delta v \label{fh_largemu0}
	\end{equation}
	Using $f_L$ in~\eqref{fl_largemu0} and $f_H$ in~\eqref{fh_largemu0}, $IC_L$ becomes
	\begin{eqnarray*}
		0 
		&\geq & V - (x_H - L)^2 - (x_L - H)^2 + (x_H - H)^2 + (x_L -L)^2 - \left[ \pi\left(l\vert L\right) + \pi\left(l\vert H\right) \right]\Delta v \\
		&=& (x_L -L)^2 - (x_L - H)^2 + (x_H - H)^2  - (x_H - L)^2 - \left[ \pi\left(l\vert L\right) + \pi\left(l\vert H\right) \right]\Delta v 
	\end{eqnarray*} 
	and (non-binding) $IR_H$ is
	\begin{equation}
		(x_L - L)^2 - (x_L - H)^2 + \pi(l\vert L)\Delta v \geq 0 \label{irh_only_irl}
	\end{equation} 
	
	Plugging~\eqref{fl_largemu0} and~\eqref{fh_largemu0} into the objective function, the problem is to maximize
	\begin{eqnarray}
		&&\quad \mu_0 \left[(x_L-H)^2-(x_H-H)^2+( \pi(l\vert H)-\pi(l\vert L) )\Delta v + V - (x_L-L)^2 \right] \notag \\
		&& +(1-\mu_0)\left[V-(x_L-L)^2\right]
		- \mu_0\pi(l\vert H)\Delta v + (1-\mu_0)\pi(l\vert L) v_L \notag \\
		&=& V + \mu_0 \left[ (x_L - H)^2 - (x_H - H)^2 \right] - (x_L - L)^2 + \pi(l\vert L)(v_L - \mu_0 v_H) \label{value_only_irl}
	\end{eqnarray}
	subject to $IC_L$ and $OB_L$.
	
	We note that $IC_L$ is binding. To see this by contradiction, suppose the constraints ($IC_L$) are not binding. We have $\pi(l\vert L) = 0$ as $v_L-\mu_0 v_H<0$. Then $\pi(l\vert H)=0$ and $x_L = \frac{L-\mu_0 H }{1-\mu_0}$ and $x_H = H$. But this contradicts with $IR_H$ in~\eqref{irh_only_irl} because
	\begin{eqnarray*}
		\left(\frac{L-\mu_0 H}{1-\mu_0}-L\right)^2 - \left(\frac{L-\mu_0 H}{1-\mu_0}-H \right)^2 <0.
	\end{eqnarray*}

	Using the binding $IC_L$, i.e.,
	\begin{eqnarray*}
		(x_L -L)^2 - (x_L - H)^2 + (x_H - H)^2  - (x_H - L)^2 - \left[ \pi\left(l\vert L\right) + \pi\left(l\vert H\right) \right]\Delta v = 0
	\end{eqnarray*}
	we have
	\begin{equation*}
		2\delta (x_L - x_H) = \left[\pi(l\vert H) -\pi(l\vert L) \right] \Delta v
	\end{equation*}
    from which by substituting out $x_H$, the objective becomes
	{\small
	\begin{eqnarray*}
		&&V + \mu_0 \left[ (x_L - H)^2 - (x_H - H)^2 \right] - (x_L - L)^2 + \pi(l\vert L)(v_L - \mu_0 v_H) \\
		&=&V + \mu_0 \left[ (x_L - H)^2 - \left( x_L + \frac{(\pi(l\vert L)-\pi(l\vert H)) \Delta v}{2\delta} - H\right)^2 \right] - (x_L - L)^2 + \pi(l\vert L)(v_L - \mu_0 v_H)
	\end{eqnarray*} }
	We have two cases depending on whether $OB_L$ is binding or not and the following shows that neither case can be optimal.
	
	(1) $OB_L$ is not binding. Then the first order conditions with $x_L$, $\pi(l\vert H)$, $\pi(l\vert L)$ are
	\begin{eqnarray*}
		&&x_L = L - \frac{\mu_0 \Delta v \left( \pi(l\vert L) - \pi(l\vert H) \right)}{2\delta} \\
		&&-2\mu_0\left(x_H - H \right) \frac{-\Delta v}{2\delta } \geq 0  \\
		&&-2\mu_0\left( x_H - H \right) \frac{\Delta v}{2\delta} + (v_L -\mu_0 v_H) \leq 0
	\end{eqnarray*}
	Suppose $x_H > H$. Then $\pi(l\vert H)=1$ and $\pi(l\vert L)=0$, which induces a contradiction with $OB_L$, $\pi(l\vert H)\leq \pi(l\vert L)$. Hence, $x_H = H$. Then $\pi(l\vert L) = 0$. From the binding $IC_L$, $2\delta (x_L-x_H)=(\pi(l\vert H)-\pi(l\vert L))\Delta v < 0$, which implies $\pi(l\vert H)<0$, a contradiction.
	
	(2) $OB_L$ is binding: $\pi(l\vert H) = C\pi(l\vert L)$. Then the objective function becomes
	{\small
	\begin{eqnarray*}
		&&V + \mu_0 \left[ (x_L - H)^2 - ( x_L + \frac{(\pi(l\vert L)-\pi(l\vert H)) \Delta v}{2\delta} - H)^2 \right] - (x_L - L)^2 + \pi(l\vert L)(v_L - \mu_0 v_H) \\
		&=&V + \mu_0 \left[ (x_L - H)^2 - ( x_L + \frac{(\pi(l\vert L)-C\pi(l\vert L)) \Delta v}{2\delta} - H)^2 \right] - (x_L - L)^2 + \pi(l\vert L) (v_L - \mu_0 v_H)
	\end{eqnarray*}}
	Its first order derivative with $\pi(l\vert L)$ is $-2\mu_0\left( x_H - H\right)\frac{(1-C)\Delta v}{2\delta}+(v_L - \mu_0 v_H) <0$, so $\pi(l\vert L) = 0$ and then $\pi(l\vert H)=0$ by the binding $OB_L$. Hence, $x_L = x_H = L$ following the first order condition for $x_L$, i.e., $\mu_0(x_L-x_H)-x_L+L=0$ and the binding $IC_L$: $2\delta(x_L-x_H)=0$. It is a contradiction.
	
	In sum, Step 1 shows by contradiction that it is not feasible that $IR_L$ is binding while $IR_H$ is not binding. 
	
	\paragraph{Step 2.} Following the first step, it must be that both $IR_L$ and $IR_H$ are binding: 
	\begin{eqnarray*}                                                     
		f_H &=& V - \left( x_H - H \right)^2 + \pi\left(l\vert H\right) \Delta v \\
		f_L &=& V -  \left( x_L - L \right)^2
	\end{eqnarray*}
	
	By plugging $f_L$ and $f_H$ into the objective, the intermediary chooses $x_H$, $x_L$ and $\pi\left(l\vert L\right)$, $\pi\left(l\vert H\right)$ to maximize
	\begin{equation*}
		V - \mu_0 \left( x_H - H \right)^2 - \left(1-\mu_0\right)\left( x_L - L \right)^2
		+ \left(1-\mu_0\right)\pi\left(l\vert L\right) v_L 
	\end{equation*}
	subject to $OB_L$, $IC_L$, $IC_H$ and $\pi\left(l\vert L\right),\pi\left(l\vert H\right)\in\left[0,1\right]$. 
	
	By plugging $f_L$ and $f_H$ into $IC_L$ and $IC_H$, we have the following $OB_L$, $IC_L$, $IC_H$ 
	\begin{align*}
		\left( x_H - H \right)^2 - \left( x_H - L \right)^2 &\leq \pi\left(l\vert H\right) \Delta v  \tag{$IC_L$} \\
		\left( x_L - L \right)^2 - \left( x_L - H \right)^2 &\leq - \pi\left(l\vert L\right) \Delta v  \tag{$IC_H$ } \\
		\pi\left(l\vert H\right) &\leq C \pi\left(l\vert L\right) \tag{$OB_L$}
	\end{align*}
	
	\paragraph{Step 3.} Solving the simplified problem. Note that $x_H=H$, and $IC_L$ and $OB_L$ are redundant. If $IC_H$ is not binding, then $x_L=L$, $\pi(l\vert L)=1$. $IC_H$ requires $\Delta v\leq \delta^2$.
	
	So if $\Delta v > \delta^2$, $IC_H$ is binding, i.e., $\left( x_L - L \right)^2 - \left( x_L - H \right)^2 = - \pi\left(l\vert L\right) \Delta v$ from which
		\begin{equation*}
			( 2x_L - L - H ) \delta = -\pi(l\vert L) \Delta v
		\end{equation*}
		or 
		\begin{equation*}
			x_L = \frac{L+H}{2} - \frac{\pi(l\vert L) \Delta v}{2\delta}
		\end{equation*}
		Substituting out $x_L$ in the objective function, the intermediary chooses $\pi\left(l\vert L\right)\in \left[0,1\right]$ to maximize
		\begin{equation*}
			- \left[ \frac{\delta }{2} - \frac{\pi(l\vert L) \Delta v}{2\delta} \right]^2 + \pi\left(l\vert L\right) v_L 
		\end{equation*}
		The solution to this problem is
		\begin{itemize}
			\item if $\frac{\delta^2}{\Delta v} \leq \frac{\Delta v}{v_L+v_H}$, then $\pi(l\vert L) = \frac{\delta^2}{\Delta v} \frac{v_H+v_L}{\Delta v} $
			and $x_L= L - \frac{\delta}{\Delta v} v_L$
			
			\item if $\frac{\Delta v}{v_L+v_H}<\frac{\delta^2}{\Delta v}<1$, then $\pi(l\vert L) = 1$ and $x_L = L - \frac{\Delta v -\delta^2}{2\delta}$
		\end{itemize}
		
	Summarizing the above results,
	\begin{itemize}
		\item if $\frac{\delta^2}{\Delta v} \geq 1$, then $\pi(l\vert L)=1$ and $x_L=L$
		
		\item if $\frac{\delta^2}{\Delta v} \leq \frac{\Delta v}{v_L+v_H}$, then $\pi(l\vert L) = \frac{\delta^2}{\Delta v} \frac{v_H+v_L}{\Delta v} $
		and $x_L= L - \frac{\delta}{\Delta v} v_L$
		
		\item if $\frac{\Delta v}{v_L+v_H}<\frac{\delta^2}{\Delta v}<1$, then $\pi(l\vert L) = 1$ and $x_L = L - \frac{\Delta v -\delta^2}{2\delta}$
	\end{itemize}
	The intermediary's total revenue is
	\begin{equation*}
		\begin{cases}
		V+(1-\mu_0)v_L	&\mbox{ if } \frac{\delta^2}{\Delta v} \geq 1 \\
		V+ (1-\mu_0)v_L -(1-\mu_0)\left( \frac{\Delta v -\delta^2}{2\delta } \right)^2 &\mbox{ if } \frac{\Delta v}{v_L+v_H} \leq \frac{\delta^2}{\Delta v}  < 1 \\
		V+(1-\mu_0) \left(\frac{\delta}{\Delta v}\right)^2 v_L v_H &\mbox{ if } \frac{\delta^2}{\Delta v} < \frac{\Delta v}{v_L+v_H}
		\end{cases}
	\end{equation*}
	which is continuous in $\delta^2/\Delta v$ and hence is consistent with Berge's maximum theorem.
	
\end{proof}

\subsection{Proof of~\Cref{general_service_welfare}}

\begin{proof}
	For $\mu_0<\frac{v_L}{v_H}$, social welfare without regulation is $V+(1-\mu_0)v_L+\mu_0v_H$, and social welfare with data banning is $V+v_L+\mu_0\Delta v$. So banning data trade has no effect on social welfare for $\mu_0<\frac{v_L}{v_H}$.
	
	For $\mu_0>\frac{v_L}{v_H}$, social welfare with data banning is $V+\mu_0v_H$, and social welfare without regulation is
	\begin{enumerate}
		\item if $\frac{\delta^2}{\Delta v}\geq 1$: $V+(1-\mu_0)v_L+\mu_0v_H$
		\item if $\frac{\Delta v}{v_L+v_H}<\frac{\delta^2}{\Delta v} < 1$: $V - (1-\mu_0)\left(\frac{\Delta v -\delta^2}{2\delta}\right)^2 
		+ (1-\mu_0)v_L 	+\mu_0 v_H $
		\item if $\frac{\delta^2}{\Delta v}<\frac{\Delta v}{v_L+v_H}$: $V + (1-\mu_0) \left(\frac{\delta}{\Delta v}\right)^2 v_L v_H 
		+ \mu_0 v_H $
	\end{enumerate}
	Welfare in each case is strictly greater than $V+\mu_0v_H$. This is obvious for the first and third cases. In the second case, $v_L - \left(\frac{\Delta v -\delta^2}{2\delta}\right)^2 = \frac{-(\Delta v)^2 - \delta^4 + 2\delta^2(v_L+v_H)}{4\delta^2}>0 $ when $\frac{\Delta v}{v_L+v_H}<\frac{\delta^2}{\Delta v} < 1$. Hence, banning data trade decreases social welfare for $\mu_0>\frac{v_L}{v_H}$.
	
\end{proof}

\subsection{Proof of~\Cref{prop:bargaining_general} }

\begin{proof}
	The following completes the proof by two parts: 
	
	\paragraph{Proof of Part~\ref{prop:bargaining_general_largemu0} 
		of~\Cref{prop:bargaining_general} }
	
	When $\mu_0 > \frac{v_L}{v_H}$, 
	the obedience constraint for the producer to charge a high price $OB_H$~\eqref{obh} is redundant, we can ignore~\eqref{obh}. The type-H consumers' outside option $\max\left\{v_H - p_0, 0\right\}=0$. So the intermediary's problem is simplified into
	\begin{align*}
		&\max_{f,\pi } \text{ }
		f_H \mu_0  + f_L \left(1-\mu_0\right) 
		+ \beta [\left(1-\mu_0\right)\pi\left(l\vert L\right) v_L - \mu_0\pi\left(l\vert H\right)\left(v_H - v_L \right) ]  \\
		\text{s.t. }& ~\eqref{ich}~\eqref{icl}~\eqref{irl}~\eqref{fesi} \eqref{obl} ~\eqref{irh_largemu}
	\end{align*}
	which are
	\begin{align*}
		& V - \left( x_H - H \right)^2 - f_H + \pi\left(l\vert H\right) \Delta v
		\geq V - \left( x_L - H \right)^2 - f_L + \pi\left(l\vert L\right) \Delta v  \tag{$IC_H$ } \\
		& V - \left( x_L - L \right)^2  - f_L 
		\geq V -  \left( x_H - L \right)^2  - f_H  \tag{$IC_L$}  \\                                                              
		& V - \left( x_H - H \right)^2  - f_H + \pi\left(l\vert H\right) \Delta v  \geq 0 \tag{$IR_H$} \\
		& V -  \left( x_L - L \right)^2 - f_L  \geq 0 \tag{$IR_L$}  \\
		& \pi\left(l\vert H\right) \leq C\pi\left(l\vert L\right)  \tag{$OB_L$}  \\
		& \pi\left(l\vert L\right),\pi\left(l\vert H\right)\in\left[0,1\right] \tag{$FE$} 
	\end{align*}
	Note that at least one of $IR_L$~\eqref{irl} and $IR_H$~\eqref{irh_largemu} is binding. Otherwise, both of $IR_L$~\eqref{irl} and $IR_H$~\eqref{irh_largemu} are not binding. So increasing $f_L$ and $f_H$ simultaneously by equal size will not affect $IC_H$~\eqref{ich} and $IC_L$~\eqref{icl}, and will increase the intermediary's profit. The following Lemma excludes the case of only $IR_H$~\eqref{irh_largemu} being binding.
	
	\begin{lem}
		If $IR_H$~\eqref{irh_largemu} is binding, then $IR_L$~\eqref{irl} is also binding.
	\end{lem}
	\begin{proof}
		Given the binding $IR_H$: $V - (x_H-H)^2 - f_H + \pi\left(l\vert H\right)\Delta v = 0$, then at least one of $IC_L$ and $IR_L$ is binding. Otherwise, increasing $f_L$ increases the objective without affecting constraints $IC_H$, $IC_L$, $IR_L$. 
		
		Suppose $IC_L$ is binding: $V-(x_L-L)^2-f_L=V-(x_H-L)^2-f_H$. From the binding $IR_H$ and $IC_L$, we have
		\begin{eqnarray*}
			&&V - (x_L - L)^2 - f_L \\
			&=&V - (x_H - L)^2 - V + (x_H-H)^2 - \pi\left(l\vert H\right)\Delta v \\
			&=&(L-H)(2x_H-L-H) - \pi\left(l\vert H\right)\Delta v  \\
			&<&0 
		\end{eqnarray*}	
		which violates $IR_L$. Hence, $IC_L$ cannot be binding, and $IR_L$ must be binding.
	\end{proof}
	
	Hence, there are two cases following the above Lemma: $(1)$ $IR_L$ is binding, $IR_H$ is not binding; $(2)$ both $IR_L$ and $IR_H$ are binding. Now we show in three steps: Step 1 shows that it's not optimal that only $IR_L$ is binding by contradiction; Step 2 simplifies the intermediary's problem with binding $IR_L$ and $IR_H$; Step 3 solves the problem.
	
	\paragraph{Step 1.} Suppose $IR_L$ is binding, $IR_H$ is not binding. As $IR_H$ is not binding, $IC_H$ must be binding. Otherwise, increasing $f_H$ will not affect $IR_H$, $IC_H$, $IC_L$, $IR_L$ while increases the intermediary's profit. So $IC_H$ is binding:
	\begin{equation}
		V - \left( x_H - H \right)^2 - f_H + \pi\left(l\vert H\right) \Delta v
		=  V - \left( x_L - H \right)^2 - f_L + \pi\left(l\vert L\right) \Delta v  \label{bindingich_largemu0_bargaining}
	\end{equation}
	also recall the binding $IR_L$: $V - (x_L - L)^2 - f_L = 0$, we have
	\begin{equation}
		f_L = V - (x_L - L)^2 \label{fl_largemu0_bargaining}
	\end{equation}
	Plugging~\eqref{fl_largemu0_bargaining} into~\eqref{bindingich_largemu0_bargaining} to get
	\begin{equation}
		f_H = V + \left( x_L - H \right)^2 - (x_L - L)^2 - \left( x_H - H \right)^2 
		- \pi\left(l\vert L\right) \Delta v + \pi\left(l\vert H\right) \Delta v \label{fh_largemu0_bargaining}
	\end{equation}
	Using $f_L$ in~\eqref{fl_largemu0_bargaining} and $f_H$ in~\eqref{fh_largemu0_bargaining}, $IC_L$ becomes
	\begin{eqnarray*}
		0 
		&\geq & V - (x_H - L)^2 - (x_L - H)^2 + (x_H - H)^2 + (x_L -L)^2 - \left[ \pi\left(l\vert L\right) + \pi\left(l\vert H\right) \right]\Delta v \\
		&=& (x_L -L)^2 - (x_L - H)^2 + (x_H - H)^2  - (x_H - L)^2 - \left[ \pi\left(l\vert L\right) + \pi\left(l\vert H\right) \right]\Delta v 
	\end{eqnarray*} 
	and ignore the (non-binding) $IR_H$
	\begin{equation}
		(x_L - L)^2 - (x_L - H)^2 + \pi(l\vert L)\Delta v \geq 0 \label{irh_only_irl_bargaining}
	\end{equation} 
	
	Plugging~\eqref{fl_largemu0_bargaining} and~\eqref{fh_largemu0_bargaining} into the objective function, the problem is to maximize
	\begin{eqnarray}
		&&\quad \mu_0 \left[ V - (x_L-L)^2 + (x_L-H)^2-(x_H-H)^2+( \pi(l\vert H)-\pi(l\vert L) )\Delta v \right] \notag \\
		&& +(1-\mu_0)\left[V-(x_L-L)^2\right] \notag \\
		&& + \beta[(1-\mu_0)\pi(l\vert L) v_L - \mu_0\pi(l\vert H)\Delta v] \notag \\
		&=& V - (x_L - L)^2 + \mu_0 \left[ (x_L - H)^2 - (x_H - H)^2 \right] \notag \\
		&& \quad   + \pi(l\vert L)[-\mu_0\Delta v + \beta(1-\mu_0)v_L]
		+ \pi(l\vert H)(1-\beta)\mu_0\Delta v
		\label{value_only_irl_bargaining}
	\end{eqnarray}
	subject to $IC_L$ and $OB_L$.
	It is useful to note $(1-\beta)\mu_0\Delta v>0$ and $-\mu_0\Delta v + \beta(1-\mu_0)v_L=\beta v_L -\mu_0(\Delta v+\beta v_L)<0$ given $\mu_0>\frac{v_L}{v_H}$.
	
	We also note that $IC_L$ is binding. To see this by contradiction, suppose the constraints ($IC_L$) are not binding and ignore it. There are two cases: (1) $OB_L$ is not binding. We have $\pi(l\vert L) = 0$ as $-\mu_0\Delta v + \beta(1-\mu_0)v_L<0$, and $\pi(l\vert H)=1$ which however violates $OB_L$. (2) $OB_L$ is binding: $\pi(l\vert H)=C\pi(l\vert L)$. The objective function becomes 
	\begin{eqnarray*}
		&&V - (x_L - L)^2 + \mu_0 \left[ (x_L - H)^2 - (x_H - H)^2 \right] \\
		&&\quad  + \pi(l\vert L)[-\mu_0\Delta v + \beta(1-\mu_0)v_L] + C\pi(l\vert L) (1-\beta)\mu_0\Delta v \\
		&=&V - (x_L - L)^2 + \mu_0 \left[ (x_L - H)^2 - (x_H - H)^2 \right] \\
		&&\quad  + \pi(l\vert L)[-\mu_0\Delta v + \beta(1-\mu_0)v_L + C (1-\beta)\mu_0\Delta v ]
	\end{eqnarray*}
	Then $\pi(l\vert L)=0$ as $-\mu_0\Delta v + \beta(1-\mu_0)v_L + C (1-\beta)\mu_0\Delta v <0$,
	and $x_L = \frac{L-\mu_0 H }{1-\mu_0}$ and $x_H = H$. But this contradicts with $IR_H$ in~\eqref{irh_only_irl_bargaining} as
	\begin{eqnarray*}
		\left(\frac{L-\mu_0 H}{1-\mu_0}-L\right)^2 - \left(\frac{L-\mu_0 H}{1-\mu_0}-H \right)^2 <0,
	\end{eqnarray*}
	so $IC_L$ is binding.

	Using the binding $IC_L$, i.e.,
	\begin{eqnarray*}
		(x_L -L)^2 - (x_L - H)^2 + (x_H - H)^2  - (x_H - L)^2 - \left[ \pi\left(l\vert L\right) + \pi\left(l\vert H\right) \right]\Delta v = 0
	\end{eqnarray*}
	we have
	\begin{equation*}
		2\delta (x_L - x_H) = \left[\pi(l\vert H) -\pi(l\vert L) \right] \Delta v
	\end{equation*}
	from which by substituting out $x_H$, the objective in~\eqref{value_only_irl_bargaining} becomes
	\begin{eqnarray*}
		&& V - (x_L - L)^2 + \mu_0 \left[ (x_L - H)^2 - \left( x_L + \frac{(\pi(l\vert L)-\pi(l\vert H)) \Delta v}{2\delta} - H\right)^2 \right]  \notag \\
		&& \quad   + \pi(l\vert L)[-\mu_0\Delta v + \beta(1-\mu_0)v_L]
		+ \pi(l\vert H)(1-\beta)\mu_0\Delta v .
	\end{eqnarray*}
	We have two cases depending on whether $OB_L$ is binding or not and the following shows that neither case can be optimal.
	
	Case (1): $OB_L$ is not binding. 
	The first order derivatives with respect to $\pi(l\vert H)$ and $\pi(l\vert L)$ are
	\begin{eqnarray*}
		-2\mu_0 \left( x_L + \frac{(\pi(l\vert L)-\pi(l\vert H)) \Delta v}{2\delta} - H\right) \frac{-\Delta v}{2\delta} + (1-\beta)\mu_0\Delta v &>& 0 \\
		-2\mu_0  \left( x_L + \frac{(\pi(l\vert L)-\pi(l\vert H)) \Delta v}{2\delta} - H\right) \frac{\Delta v}{2\delta}  -\mu_0\Delta v + \beta(1-\mu_0)v_L &<& 0
	\end{eqnarray*}	
	where the second inequality follows $x_H\geq H$ and $-\mu_0\Delta v + \beta(1-\mu_0)v_L <0$ under $\mu_0>v_L/v_H$.
	So $\pi(l\vert H)=1$ and $\pi(l\vert L)=0$, which contradicts with $OB_L$.   
	
	Case (2): $OB_L$ is binding, i.e., $\pi(l\vert H) = C\pi(l\vert L)$. Then the objective function becomes
	\begin{eqnarray*}
		&& V - (x_L - L)^2 + \mu_0 \left[ (x_L - H)^2 - \left( x_L + \frac{(\pi(l\vert L)-C\pi(l\vert L)) \Delta v}{2\delta} - H\right)^2 \right]  \notag \\
		&& \quad   + \pi(l\vert L)[-\mu_0\Delta v + \beta(1-\mu_0)v_L]
		+ C\pi(l\vert L)(1-\beta)\mu_0\Delta v
	\end{eqnarray*}
	Its first order derivative with $\pi(l\vert L)$ is 
	{\small
	\begin{equation*}
		-2\mu_0  \left( x_L + \frac{(\pi(l\vert L)-\pi(l\vert H)) \Delta v}{2\delta} - H\right) \frac{(1-C)\Delta v}{2\delta}  -\mu_0\Delta v + \beta(1-\mu_0)v_L +  C(1-\beta)\mu_0\Delta v < 0
	\end{equation*}}
	where the inequality follows that $-2\mu_0  \left( x_H - H\right) \frac{(1-C)\Delta v}{2\delta}\leq 0$ and $ -\mu_0\Delta v + \beta(1-\mu_0)v_L +  C(1-\beta)\mu_0\Delta v < 0$.
	So $\pi(l\vert L) = 0$ and then $\pi(l\vert H)=0$ by the binding $OB_L$. Hence, $x_L = x_H $ following the binding $IC_L$. It is a contradiction.
	
	In sum, Step 1 shows by contradiction that it is not feasible that $IR_L$ is binding while $IR_H$ is not binding. 
	
	\paragraph{Step 2.} Following the first step, it must be that both $IR_L$ and $IR_H$ are binding: 
	\begin{eqnarray*}                                                     
		f_H &=& V - \left( x_H - H \right)^2 + \pi\left(l\vert H\right) \Delta v \\
		f_L &=& V -  \left( x_L - L \right)^2
	\end{eqnarray*}
	Plugging $f_L$ and $f_H$ into the objective, the intermediary chooses $x_H$, $x_L$, $\pi\left(l\vert L\right)$, $\pi\left(l\vert H\right)$ to maximize
	\begin{eqnarray*}
		&&	f_H \mu_0  + f_L \left(1-\mu_0\right) 
		+ \beta [\left(1-\mu_0\right)\pi\left(l\vert L\right) v_L - \mu_0\pi\left(l\vert H\right)\left(v_H - v_L \right) ]  \\
		&=& V - \mu_0\left( x_H - H \right)^2  - (1-\mu_0)\left( x_L - L \right)^2 + (1-\beta) \mu_0\pi\left(l\vert H\right) \Delta v \\
		&&+ \beta\left(1-\mu_0\right)\pi\left(l\vert L\right) v_L
	\end{eqnarray*}
	subject to $OB_L$, $IC_L$, $IC_H$ and $\pi\left(l\vert L\right),\pi\left(l\vert H\right)\in\left[0,1\right]$. 
	
	By plugging $f_L$ and $f_H$ into $IC_L$ and $IC_H$, we have the following $OB_L$, $IC_L$, $IC_H$ 
	\begin{align*}
		\left( x_H - H \right)^2 - \left( x_H - L \right)^2 &\leq \pi\left(l\vert H\right) \Delta v  \tag{$IC_L$} \\
		\left( x_L - L \right)^2 - \left( x_L - H \right)^2 &\leq - \pi\left(l\vert L\right) \Delta v  \tag{$IC_H$ } \\
		\pi\left(l\vert H\right) &\leq C \pi\left(l\vert L\right) \tag{$OB_L$}
	\end{align*}
	
	\paragraph{Step 3.} Solving the simplified problem:
	\begin{eqnarray*}
		&& V - \mu_0\left( x_H - H \right)^2  - (1-\mu_0)\left( x_L - L \right)^2 + (1-\beta) \mu_0\pi\left(l\vert H\right) \Delta v \\
		&&+ \beta\left(1-\mu_0\right)\pi\left(l\vert L\right) v_L
	\end{eqnarray*}
	subject to  $\pi\left(l\vert L\right),\pi\left(l\vert H\right)\in\left[0,1\right]$, and
	\begin{align*}
		\left( x_H - H \right)^2 - \left( x_H - L \right)^2 &\leq \pi\left(l\vert H\right) \Delta v  \tag{$IC_L$} \\
		\left( x_L - L \right)^2 - \left( x_L - H \right)^2 &\leq - \pi\left(l\vert L\right) \Delta v  \tag{$IC_H$ } \\
		\pi\left(l\vert H\right) &\leq C \pi\left(l\vert L\right) \tag{$OB_L$}
	\end{align*} 
	Note that $x_H=H$ and $IC_L$ is redundant. 
	
	If $IC_H$ is not binding, then $x_L=L$, $\pi(l\vert L)=1$ and $\pi(l\vert H)=C$. $IC_H$ requires $\Delta v\leq \delta^2$.
	
	If $\Delta v > \delta^2$, then $IC_H$ is binding, i.e., $\left( x_L - L \right)^2 - \left( x_L - H \right)^2 = - \pi\left(l\vert L\right) \Delta v$ from which
	\begin{equation*}
		( 2x_L - L - H ) \delta = -\pi(l\vert L) \Delta v
	\end{equation*}
	or 
	\begin{equation*}
		x_L = \frac{L+H}{2} - \frac{\pi(l\vert L) \Delta v}{2\delta}
	\end{equation*}
	Substituting out $x_L$ in the objective function, the intermediary chooses $\pi\left(l\vert L\right)\in \left[0,1\right]$ to maximize
	\begin{eqnarray*}
		V - (1-\mu_0)\left( \frac{\delta}{2} - \frac{\pi(l\vert L) \Delta v}{2\delta} \right)^2 + (1-\beta) \mu_0\pi\left(l\vert H\right) \Delta v
		+ \beta\left(1-\mu_0\right)\pi\left(l\vert L\right) v_L
	\end{eqnarray*}
	Suppose $OB_L$ is not binding, then $\pi(l\vert H)=1$, contradicting with $OB_L$. So $OB_L$ is binding, and the objective function is 
	\begin{eqnarray*}
		V - (1-\mu_0)\left( \frac{\delta}{2} - \frac{\pi(l\vert L) \Delta v}{2\delta} \right)^2 
		+ (1-\beta) \mu_0 C \pi\left(l\vert L\right) \Delta v
		+ \beta\left(1-\mu_0\right)\pi\left(l\vert L\right) v_L
	\end{eqnarray*}
	If $\pi(l\vert L)$ is not binding, the first order derivative with respect to $\pi(l\vert L)$ is 
	\begin{eqnarray*}
		2(1-\mu_0)\left( \frac{\delta}{2} - \frac{\pi(l\vert L) \Delta v}{2\delta} \right)\frac{\Delta v}{2\delta}
		+ (1-\beta) \mu_0 C \Delta v
		+ \beta\left(1-\mu_0\right) v_L = 0
	\end{eqnarray*}
	from which
	\begin{eqnarray*}
		\pi(l\vert L) 
		&=&  \frac{\delta^2}{\Delta v} 
		+\frac{2(1-\beta) \mu_0 C }{1-\mu_0} \frac{\delta^2}{\Delta v}
		+\frac{2\beta v_L }{\Delta v} \frac{\delta^2}{\Delta v} \\
		&=& \frac{\delta^2}{\Delta v} \left[ 1 + \frac{2(1-\beta) \mu_0 C }{1-\mu_0}
		+\frac{2\beta v_L }{\Delta v} \right] \\
		&=& \frac{\delta^2}{\Delta v} \frac{(1-\mu_0)\Delta v + 2\left[ (1-\beta) \mu_0 C \Delta v + \beta (1-\mu_0)v_L\right] }{(1-\mu_0)\Delta v} \\
		&=& \frac{\delta^2}{\Delta v} \frac{v_H + v_L}{\Delta v}
	\end{eqnarray*}
	where the last equality follows $\mu_0 C \Delta v = (1-\mu_0)v_L$ and the nonbinding $\pi(l\vert L)$ requires $\frac{\delta^2}{\Delta v} \frac{v_H + v_L}{\Delta v} \leq 1$ or $\frac{\delta^2}{\Delta v} \leq \frac{\Delta v}{v_H+v_L}$. 
	Hence, the solution to this problem is
	\begin{itemize}
		\item if $\frac{\delta^2}{\Delta v} \leq \frac{\Delta v}{v_H+v_L}$, 
		then $\pi(l\vert L) =\frac{\delta^2}{\Delta v} \frac{v_H + v_L}{\Delta v}$
		and $x_L = \frac{L+H}{2} - \frac{\Delta v}{2\delta}\pi(l\vert L)
		= L - \frac{\delta v_L}{\Delta v}
		$
		
		\item if $\frac{\Delta v}{v_H+v_L} < \frac{\delta^2}{\Delta v} < 1$, then $\pi(l\vert L) = 1$ and $x_L = \frac{L+H}{2} - \frac{\Delta v }{2\delta} = L - \frac{\Delta v - \delta^2}{2\delta}$
	\end{itemize}
	
	In sum, we have $x_H = H$ and
	\begin{itemize}
		\item if $\frac{\delta^2}{\Delta v} \geq 1$, then $\pi(l\vert L)=1$, $\pi(l\vert H)=C$ and $x_L=L$. 
		
		\item if $\frac{\Delta v}{v_H+v_L} <\frac{\delta^2}{\Delta v}<1$, then $\pi(l\vert L) = 1$, $\pi(l\vert H) = C$ and $x_L = L - \frac{\Delta v -\delta^2}{2\delta}$. 
		
		\item if $\frac{\delta^2}{\Delta v} \leq \frac{\Delta v}{v_H+v_L}$, then $\pi(l\vert L) = \frac{\delta^2}{\Delta v} 
		\frac{v_H+v_L}{\Delta v}$, $\pi(l\vert H) = \frac{C\delta^2}{\Delta v} 
		\frac{v_H+v_L}{\Delta v}$
		and $x_L = L - \frac{\delta v_L }{\Delta v}
		$.
		
	\end{itemize}
	
	In the optimal mechanism, the consumers surplus is $0$ and the producer surplus is $\mu_0v_H$.
	The intermediary's total revenue is
	\begin{equation*}
		\begin{cases}
			V+(1-\beta)\mu_0C\Delta v +\beta(1-\mu_0)v_L  &\mbox{ if } \frac{\delta^2}{\Delta v} \geq 1 \\
			V+ (1-\beta)\mu_0C\Delta v + \beta(1-\mu_0)v_L -(1-\mu_0)\left(\frac{\Delta v -\delta^2}{2\delta}\right)^2  &\mbox{ if } \frac{\Delta v}{v_H+v_L} \leq \frac{\delta^2}{\Delta v}  < 1 \\
			V+ (1-\beta)\mu_0 \frac{C\delta^2(v_H+v_L)}{\Delta v}
			+\beta(1-\mu_0)\frac{(v_H+v_L) v_L \delta^2}{(\Delta v)^2} -(1-\mu_0)\left( \frac{\delta v_L}{\Delta v}  \right)^2  &\mbox{ if } \frac{\delta^2}{\Delta v} < \frac{\Delta v}{v_H+v_L}
		\end{cases}
	\end{equation*}
	By substituting out $C$, the above can be rewritten as
	\begin{equation*}
		\begin{cases}
			V + (1-\mu_0) v_L  &\mbox{ if } \frac{\delta^2}{\Delta v} \geq 1 \\
			V+ (1-\mu_0) v_L - (1-\mu_0)\left(\frac{\Delta v -\delta^2}{2\delta}\right)^2  &\mbox{ if } \frac{\Delta v}{v_H+v_L} \leq \frac{\delta^2}{\Delta v}  < 1 \\
			V+ (1-\mu_0)\left(\frac{\delta}{\Delta v}\right)^2 v_Lv_H  &\mbox{ if } \frac{\delta^2}{\Delta v} < \frac{\Delta v}{v_H+v_L}
		\end{cases}
	\end{equation*}
	which is independent of the intermediary's bargaining power parameter $\beta$ and is the same as the base model.

	\vspace{.18in}	
	Recall in the base setup with $\beta=1$, we have the following:
	\begin{itemize}
		\item if $\frac{\delta^2}{\Delta v} \geq 1$, then $\pi(l\vert L)=1$ and $x_L=L$
		
		\item if $\frac{\delta^2}{\Delta v} \leq \frac{\Delta v}{v_L+v_H}$, then $\pi(l\vert L) = \frac{\delta^2}{\Delta v} \frac{v_H+v_L}{\Delta v} $
		and $x_L= L - \frac{\delta}{\Delta v} v_L$
		
		\item if $\frac{\Delta v}{v_L+v_H}<\frac{\delta^2}{\Delta v}<1$, then $\pi(l\vert L) = 1$ and $x_L = L - \frac{\Delta v -\delta^2}{2\delta}$
	\end{itemize}
	The intermediary's total revenue is also
	\begin{equation*}
		\begin{cases}
			V+(1-\mu_0)v_L	&\mbox{ if } \frac{\delta^2}{\Delta v} \geq 1 \\
			V-(1-\mu_0)\left( \frac{\Delta v -\delta^2}{2\delta } \right)^2 + (1-\mu_0)v_L &\mbox{ if } \frac{\Delta v}{v_L+v_H} \leq \frac{\delta^2}{\Delta v}  < 1 \\
			V+(1-\mu_0) \left(\frac{\delta}{\Delta v}\right)^2 v_L v_H &\mbox{ if } \frac{\delta^2}{\Delta v} < \frac{\Delta v}{v_L+v_H}
		\end{cases}.
	\end{equation*}
	

	\paragraph{Proof of Part~\ref{prop:bargaining_general_smallmu0} of~\Cref{prop:bargaining_general} }
	
		When $\mu_{0}<v_L/v_H$, the obedience constraint $OB_{L}$ is redundant and the type-H consumers' outside option max $\left\{v_{H}-p_{0}, 0\right\}=v_{H}-v_{L}$. Hence, the intermediary solves
		\begin{align}
			&\max_{f,x,\pi,T}  f_H \mu_0  + f_L \left(1-\mu_0\right) + \beta [R\left(\pi\right) - R_0 ] \notag \\
			& \quad =  f_H \mu_0  + f_L \left(1-\mu_0\right) 
			+ \beta \left[ \left(1-\mu_0\right)\pi\left(l\vert L\right) v_L - \mu_0\pi\left(l\vert H\right)\left(v_H - v_L \right) + \mu_0v_H - v_L  \right] \notag \\
			\text{s.t.}& \notag \\
			& V - \left( x_H - H \right)^2 - f_H + \pi\left(l\vert H\right) \left( v_H - v_L \right) \notag \\
			&\hspace{2cm} \geq V - \left( x_L - H \right)^2 - f_L + \pi\left(l\vert L\right) \left( v_H - v_L \right)  \tag{$IC_H$ } \\
			& V - \left( x_L - L \right)^2  - f_L 
			\geq V -  \left( x_H - L \right)^2  - f_H  \tag{$IC_L$}  \\                                                              
			& V - \left( x_H - H \right)^2  - f_H + \pi\left(l\vert H\right) \left( v_H - v_L \right)  \geq v_H - v_L \tag{$IR_H$} \\
			& V -  \left( x_L - L \right)^2 - f_L  \geq 0 \tag{$IR_L$}  \\
			& \frac{v_{L}}{v_{H}}\left(1-\mu_{0}\right) \pi\left(l \vert L\right)-(1-\frac{v_{L}}{v_{H}}) \mu_{0} \pi\left(l \vert H\right) \geq \frac{v_{L}}{v_{H}}-\mu_{0} \tag{$OB_H$}  \\
			& \pi\left(l\vert L\right),\pi\left(l\vert H\right)\in\left[0,1\right]  \tag{$FE$} 
		\end{align}
		
		\begin{lem}
			If $IR_{L}$~\eqref{irl} is binding, then $IR_{H}$~\eqref{irh_smallmu} is also binding.
		\end{lem}
		\begin{proof}
			Given $IR_{L}$ is binding $f_{L}=V-\left(x_{L}-L\right)^{2}$. First, at least one of $IC_{H}$ and $IR_{H}$ is binding. Suppose not, both are not binding. Increasing $f_{H}$ will not affect $IC_{H}$, $IR_H$~\eqref{irh_smallmu}, $IC_{L}$ and increase the intermediary's payoff.
			
			Second, $I C_{H}$ is not binding. To see this, suppose $I C_{H}$ is binding, then it gives
			$$
			V-\left(x_{H}-H\right)^{2}-f_{H}+\pi\left(l \vert H\right)\left(v_{H}-v_{L}\right)=V-\left(x_{L}-H\right)^{2}-f_{L}+\pi\left(l \vert L\right)\left(v_{H}-v_{L}\right)
			$$
			Together with the binding $IR_L$: $f_{L}=V-\left(x_{L}-L\right)^{2}$, we have
			$$
			\begin{aligned}
				f_{H} &=V+\left(x_{L}-H\right)^{2}-\left(x_{H}-H\right)^{2}-\left(x_{L}-L\right)^{2}-(\pi\left(l \vert L\right)-\pi\left(l \vert H\right)) \Delta v \\
				& \leq V-\left(x_{H}-H\right)^{2}+\pi\left(l \vert H\right) \Delta v-\Delta v
			\end{aligned}
			$$
			where the inequality follows from $I R_{H}$. So the above inequality yields
			$$
			2 x_{L} \geq L+H + [1-\pi(l\vert L)]\frac{\Delta v}{\delta} \geq L + H > 2L
			$$
			leading a contradiction with $x_{L} \leq L .$ Hence, $I C_{H}$ is not binding.
			
			Thus, by the above two points, $I R_{H}$ is binding.
		\end{proof}
		
		In addition, we note that at least one of $IR_{L}$ and $IR_{H}$ is binding. Otherwise, increasing $f_{L}$ and $f_{H}$ simultaneously by equal size will not affect $IC_{H}$, $IC_{L}$, $IR_{H}$, $IR_{L}$ and will increase the intermediary's payoff.
		
		Hence, by the above lemma, there are two cases: $(1)$ both $IR_{L}$ and $IR_H$ are binding; $(2)$ only $IR_H$ is binding. 
		The following Step 1 shows that Case (2) is not feasible by contradiction, and Step 2 solves the optimal mechanism under Case (1).
		
		\paragraph{Step 1.} We show that the case (1) of only $IR_H$~\eqref{irh_smallmu} being binding is not feasible by contradiction. Suppose only $IR_H$~\eqref{irh_smallmu} is binding: $f_{H}=V-\left(x_{H}-H\right)^{2}+\left[\pi\left(l \vert H\right)-1\right] \Delta v$. As $I R_{L}$ is not binding, $IC_{L}$ must be binding. Otherwise, increasing $f_{L}$ will not affect $IR_{L}$, $IC_{L}$, $IC_{H}$ and will increase intermediary's payoff. So the binding $IR_H$~\eqref{irh_smallmu} and $IC_{L}$ give
		\begin{eqnarray*}
			f_{H}&=&V-\left(x_{H}-H\right)^{2}+(\pi\left(l \vert H\right)-1) \Delta v \\
			f_{L}&=&V-\left(x_L-L\right)^2 + \left(x_H-L\right)^2 - \left(x_H-H\right)^2 + \left(\pi\left(l\vert H\right)-1\right)\Delta v 
		\end{eqnarray*}	
		Then substituting out $f_H$ and $f_L$ in $IC_{H}$, it becomes
		\begin{equation*}
			2\left(x_L-x_H\right) \leq \left[ \pi\left(l\vert H\right) - \pi\left(l\vert L\right) \right] \frac{\Delta v}{\delta}
		\end{equation*}
		and ignore the (non-binding) $IR_{L}$
		\begin{equation*}
			2 x_H - L - H \leq \left[ 1-\pi\left(l \vert H\right) \right] \frac{\Delta v}{\delta}
		\end{equation*}
		
		The intermediary's problem is
		$$
		\begin{aligned}
			\max _{x, \pi} 
			V-\left(x_{H}-{H}\right)^{2}&+(\pi(l \vert H)-1)\Delta v
			+\left[\left(x_{H}-{L}\right)^{2}-\left(x_{L}-{L}\right)^{2}\right]\left(1-\mu_{0}\right) \\
			&+ \beta \left[ \left(1-\mu_0\right)\pi\left(l\vert L\right) v_L - \mu_0\pi\left(l\vert H\right)\left(v_H - v_L \right) + \mu_0v_H - v_L  \right] \\
			=	V-\left(x_{H}-{H}\right)^{2} &+ \pi(l\vert H)(1-\beta\mu_0)\Delta v + \pi(l\vert L)\beta (1-\mu_0) v_L \\
			&+\left[\left(x_{H}-{L}\right)^{2}-\left(x_{L}-{L}\right)^{2}\right]\left(1-\mu_{0}\right) 
			-(1-\beta\mu_0)v_H-(1-\beta)v_L
		\end{aligned}
		$$
		subject to $O B_{H}$ and $I C_{H}$: $2\left(x_{L}-x_{H}\right) \leq [\pi\left(l \vert H\right)-\pi\left(l \vert L\right) ]\frac{\Delta v}{\delta}$, 
		and $\pi\left(l \vert L\right)$, $\pi\left(l \vert H\right) \in\left[0,1\right]$. There are two cases depending on whether $IC_H$ is binding or not. We show  $I C_{H}$ is binding by contradiction: Suppose $I C_{H}$ is not binding and ignore the constraint, the solution is
		$$
		x_{L}=L, x_{H}=\frac{H-\left(1-\mu_{0}\right) L }{\mu_{0}}, \pi\left(l \vert L\right)=\pi\left(l \vert H\right)=1
		$$
		which indeed satisfies $O B_{H}$ and $IC_{H}$. But the solution violates $IR_{L}$  because $2 \frac{H-\left(1-\mu_{0}\right) L }{\mu_{0}} - L - H = \frac{(2-\mu_0)(H-L)}{\mu_0} > 0 = [1-\pi(l\vert H)]\Delta v/\delta$.
		Hence, $IC_{H}$ is binding: 
		\begin{equation}
			2\left(x_{L}-x_{H}\right)=\left[\pi\left(l \vert H\right)-\pi\left(l \vert L\right) \right] \frac{\Delta v}{\delta}.\label{ich_binding_bargaining}
		\end{equation}
		
		Plugging~\eqref{ich_binding_bargaining} into the intermediary's objective function to substitute out $x_L$, it becomes
		\begin{align*}
			\max _{x, \pi} 
			V-\left(x_{H}-{H}\right)^{2} &+ \pi(l\vert H)(1-\beta\mu_0)\Delta v + \pi(l\vert L)\beta (1-\mu_0) v_L \\
			&+\left[\left(x_{H}-{L}\right)^{2}-\left(x_{L}-{L}\right)^{2}\right]\left(1-\mu_{0}\right) 
			-(1-\beta\mu_0)v_H-(1-\beta)v_L \\
			=V-\left(x_{H}-{H}\right)^{2} &+ \pi(l\vert H)(1-\beta\mu_0)\Delta v + \pi(l\vert L)\beta (1-\mu_0) v_L \\
			&+\left[\left(x_{H}-{L}\right)^{2}-\left( \frac{\pi\left(l \vert H\right)-\pi\left(l \vert L\right)}{2}\frac{\Delta v}{\delta}+x_{H}-L \right)^{2}\right]\left(1-\mu_{0}\right) \\
			&-(1-\beta\mu_0)v_H-(1-\beta)v_L
		\end{align*}
		subject to $OB_H$, $\pi\left(l \vert L\right)\in\left[0,1\right]$, $\pi\left(l \vert H\right) \in\left[0,1\right]$. The first order derivative with respect to $\pi(l\vert H)$ is positive, so $OB_H$ is binding: $\frac{v_L}{v_H}(1-\mu_0)\pi(l\vert L) - (1-\frac{v_L}{v_H})\mu_0 \pi(l\vert H) = \frac{v_L}{v_H}-\mu_0$. Let $\lambda$ be the Lagrangian multiplier of the binding $OB_H$ constraint.

		We show that either the binding or nonbinding $\pi(l\vert L)\in [0,1]$ is impossible. First, suppose $\pi(l\vert L)\in [0,1]$ is not binding and ignore it. The first order conditions with respect to $\pi(l\vert L)$, $\pi(l\vert H)$ and $x_{H}$ are respectively
		\begin{eqnarray*}
			\beta(1-\mu_0)v_L+2(1-\mu_0)\left[\frac{\pi(l\vert H)-\pi(l\vert L)}{2}\frac{\Delta v}{\delta} + x_H -L \right]\frac{\Delta v}{2\delta} 
			+\lambda\frac{v_L}{v_H}(1-\mu_0)  &=& 0\\
			(1-\beta\mu_0)\Delta v - 2(1-\mu_0)\left[\frac{\pi(l\vert H)-\pi(l\vert L)}{2}\frac{\Delta v}{\delta}+x_H-L\right]\frac{\Delta v}{2\delta}
			-\lambda(1-\frac{v_L}{v_H})\mu_0 &=& 0 \\
			-2(x_H-H)-(1-\mu_0)[\pi(l\vert H)-\pi(l\vert L)]\frac{\Delta v}{\delta}&=& 0
		\end{eqnarray*}
		the sum of the first two conditions yields
		\begin{equation*}
			\beta(1-\mu_0)v_L +(1-\beta\mu_0)\Delta v +\lambda (\frac{v_L}{v_H}-\mu_0)= 0
		\end{equation*}
		which contradicts with $\mu_0<\frac{v_L}{v_H}$. Second, suppose $\pi(l\vert L)\in [0,1]$ is binding, then $\pi(l\vert L)=1=\pi(l\vert H)$ and $x_H=H$ by the third condition. By~\eqref{ich_binding_bargaining} $x_L=x_H= H$ which contradicts with $x_L\leq L$. 
		
		In sum, it's not feasible that only $I R_{H}$ is binding.
		
		\paragraph{Step 2.} Following Step 1, it must be both $I R_{L}$ and $I R_{H}$ are binding:
		\begin{eqnarray*}
			f_{H}&=&V-\left(x_{H}-H\right)^{2}+\pi\left(l \vert H\right) \Delta v-\Delta v \\
			f_{L}&=&V-\left(x_{L}-L\right)^{2}
		\end{eqnarray*}
		Then from $I C_{H}$
		\begin{equation*}
			\Delta v \geq-\left(x_{L}-H\right)^{2}+\left(x_{L}-L\right)^{2}+\pi\left(l \vert L\right) \Delta v
		\end{equation*}
		and from $I C_{L}$
		\begin{equation*}
			0 \geq-\left(x_{H}-L\right)^{2}+\left(x_{H}-H\right)^{2}-\pi\left(l \vert H\right) \Delta v+\Delta v
		\end{equation*}
		equivalently, they are
		\begin{eqnarray*}
			1-\pi\left(l \vert L\right) & \geq& \left( 2 x_{L}- L - H\right)\frac{\delta}{\Delta v}\\
			1-\pi\left(l \vert H\right) & \leq& \left( 2 x_{H}- L - H\right)\frac{\delta}{\Delta v}
		\end{eqnarray*}
		By plugging $f_{H}$ and $f_{L}$ into the intermediary's objective, we can rewrite it as
		\begin{align*}
			&f_H \mu_0  + f_L \left(1-\mu_0\right) 
			+ \beta \left[ \left(1-\mu_0\right)\pi\left(l\vert L\right) v_L - \mu_0\pi\left(l\vert H\right)\left(v_H - v_L \right) + \mu_0v_H - v_L  \right] \notag \\
			=&[V-\left(x_{H}-H\right)^{2}+\pi\left(l \vert H\right) \Delta v-\Delta v ] \mu_0  
			+[V-\left(x_{L}-L\right)^{2}] \left(1-\mu_0\right) \\
			&+ \beta \left[ \left(1-\mu_0\right)\pi\left(l\vert L\right) v_L - \mu_0\pi\left(l\vert H\right)\left(v_H - v_L \right) + \mu_0v_H - v_L  \right] \notag \\
			=&V-\mu_0 \left(x_{H}-H\right)^{2}+\mu_0 \pi\left(l \vert H\right) \Delta v-\mu_0 \Delta v 
			-\left(1-\mu_0\right)\left(x_{L}-L\right)^{2}  \\
			&+ \beta \left[ \left(1-\mu_0\right)\pi\left(l\vert L\right) v_L - \mu_0\pi\left(l\vert H\right)\left(v_H - v_L \right) + \mu_0v_H - v_L  \right] \notag \\
			=&V-\mu_0 \left(x_{H}-H\right)^{2}
			-\left(1-\mu_0\right)\left(x_{L}-L\right)^{2} +\mu_0(1-\beta) \pi\left(l \vert H\right) \Delta v   
			+ \beta  \left(1-\mu_0\right)\pi\left(l\vert L\right) v_L \\
			& + \beta  \mu_0v_H - \beta v_L -\mu_0 \Delta v
		\end{align*}
		The intermediary maximizes the objective subject to $O B_{H}$ and $\pi\left(l \vert L\right)$, $\pi\left(l \vert H\right)\in\left[0,1\right]$ and
		\begin{eqnarray*}
			1-\pi\left(l \vert L\right) & \geq& \left( 2 x_{L}- L - H\right)\frac{\delta}{\Delta v}\\
			1-\pi\left(l \vert H\right) & \leq& \left( 2 x_{H}- L - H\right)\frac{\delta}{\Delta v}
		\end{eqnarray*}
		Hence, given $\beta\in(0,1)$, $x_{L}=L$, $x_{H}=H$ and $\pi\left(l \vert L\right)=1$ and $\pi\left(l \vert H\right)=1$. 
		
		In sum, when $\mu_{0}<\frac{v_{L}}{v_{H}}$, $x_L=L$, $x_H=H$, $f_L=f_H=V$, $\pi\left(l \vert L\right)=\pi\left(l \vert H\right) = 1$. The intermediary's payoff is $V$; the producer's payoff is $v_L$; consumers' surplus is $\mu_0\Delta v$.
		
	\end{proof}

\bibliography{references}

\begin{thebibliography}{56}
\providecommand{\natexlab}[1]{#1}
\providecommand{\url}[1]{\texttt{#1}}
\expandafter\ifx\csname urlstyle\endcsname\relax
  \providecommand{\doi}[1]{doi: #1}\else
  \providecommand{\doi}{doi: \begingroup \urlstyle{rm}\Url}\fi

\bibitem[Acemoglu et~al.(2022)Acemoglu, Makhdoumi, Malekian, and
  Ozdaglar]{acemoglu2022too}
Daron Acemoglu, Ali Makhdoumi, Azarakhsh Malekian, and Asu Ozdaglar.
\newblock Too much data: Prices and inefficiencies in data markets.
\newblock \emph{American Economic Journal: Microeconomics}, 14\penalty0
  (4):\penalty0 218--256, 2022.

\bibitem[Acquisti and Varian(2005)]{acquisti2005conditioning}
Alessandro Acquisti and Hal~R Varian.
\newblock Conditioning prices on purchase history.
\newblock \emph{Marketing Science}, 24\penalty0 (3):\penalty0 367--381, 2005.

\bibitem[Acquisti et~al.(2016)Acquisti, Taylor, and
  Wagman]{acquisti2016economics}
Alessandro Acquisti, Curtis Taylor, and Liad Wagman.
\newblock The economics of privacy.
\newblock \emph{Journal of Economic Literature}, 54\penalty0 (2):\penalty0
  442--92, 2016.

\bibitem[Admati and Pfleiderer(1990)]{admati1990direct}
Anat~R Admati and Paul Pfleiderer.
\newblock Direct and indirect sale of information.
\newblock \emph{Econometrica: Journal of the Econometric Society}, pages
  901--928, 1990.

\bibitem[Ali et~al.(2023)Ali, Lewis, and Vasserman]{ali2023voluntary}
S~Nageeb Ali, Greg Lewis, and Shoshana Vasserman.
\newblock Voluntary disclosure and personalized pricing.
\newblock \emph{Review of Economic Studies}, 90\penalty0 (2):\penalty0
  538--571, 2023.

\bibitem[Argenziano and Bonatti(2023)]{argenziano2023datam}
Rossella Argenziano and Alessandro Bonatti.
\newblock Data markets with privacy-conscious consumers.
\newblock In \emph{AEA Papers and Proceedings}, volume 113, pages 191--196.
  American Economic Association, 2023.

\bibitem[Argenziano and Bonatti(Forthcoming)]{argenziano2023datal}
Rossella Argenziano and Alessandro Bonatti.
\newblock Data linkages and privacy regulation.
\newblock Technical report, Forthcoming.

\bibitem[Armstrong and Zhou(2022)]{armstrong2022consumer}
Mark Armstrong and Jidong Zhou.
\newblock Consumer information and the limits to competition.
\newblock \emph{American Economic Review}, 112\penalty0 (2):\penalty0 534--577,
  2022.

\bibitem[Babaioff et~al.(2012)Babaioff, Kleinberg, and
  Paes~Leme]{babaioff2012optimal}
Moshe Babaioff, Robert Kleinberg, and Renato Paes~Leme.
\newblock Optimal mechanisms for selling information.
\newblock In \emph{Proceedings of the 13th ACM Conference on Electronic
  Commerce}, pages 92--109, 2012.

\bibitem[Bergemann and Bonatti(2015)]{bergemann2015selling}
Dirk Bergemann and Alessandro Bonatti.
\newblock Selling cookies.
\newblock \emph{American Economic Journal: Microeconomics}, 7\penalty0
  (3):\penalty0 259--94, 2015.

\bibitem[Bergemann and Bonatti(2019)]{bergemann2019markets}
Dirk Bergemann and Alessandro Bonatti.
\newblock Markets for information: An introduction.
\newblock \emph{Annual Review of Economics}, 11:\penalty0 85--107, 2019.

\bibitem[Bergemann and Morris(2016)]{bergemann2016bayes}
Dirk Bergemann and Stephen Morris.
\newblock Bayes correlated equilibrium and the comparison of information
  structures in games.
\newblock \emph{Theoretical Economics}, 11\penalty0 (2):\penalty0 487--522,
  2016.

\bibitem[Bergemann and Morris(2019)]{bergemann2019information}
Dirk Bergemann and Stephen Morris.
\newblock Information design: A unified perspective.
\newblock \emph{Journal of Economic Literature}, 57\penalty0 (1):\penalty0
  44--95, 2019.

\bibitem[Bergemann and Ottaviani(2021)]{bergemann2021information}
Dirk Bergemann and Marco Ottaviani.
\newblock Information markets and nonmarkets.
\newblock In \emph{Handbook of Industrial Organization}, volume~4, pages
  593--672. Elsevier, 2021.

\bibitem[Bergemann and V{\"a}lim{\"a}ki(2006)]{bergemann2006information}
Dirk Bergemann and Juuso V{\"a}lim{\"a}ki.
\newblock Information in mechanism design.
\newblock In \emph{World Congress of the Econometric-Society}, pages 186--221.
  Cambridge University Press, 2006.

\bibitem[Bergemann et~al.(2015)Bergemann, Brooks, and
  Morris]{bergemann2015limits}
Dirk Bergemann, Benjamin Brooks, and Stephen Morris.
\newblock The limits of price discrimination.
\newblock \emph{American Economic Review}, 105\penalty0 (3):\penalty0 921--57,
  2015.

\bibitem[Bergemann et~al.(2018)Bergemann, Bonatti, and
  Smolin]{bergemann2018design}
Dirk Bergemann, Alessandro Bonatti, and Alex Smolin.
\newblock The design and price of information.
\newblock \emph{American Economic Review}, 108\penalty0 (1):\penalty0 1--48,
  2018.

\bibitem[Bergemann et~al.(2022)Bergemann, Bonatti, and
  Gan]{bergemann2022economics}
Dirk Bergemann, Alessandro Bonatti, and Tan Gan.
\newblock The economics of social data.
\newblock \emph{The RAND Journal of Economics}, 53\penalty0 (2):\penalty0
  263--296, 2022.

\bibitem[Bester and Strausz(2001)]{bester2001contracting}
Helmut Bester and Roland Strausz.
\newblock Contracting with imperfect commitment and the revelation principle:
  the single agent case.
\newblock \emph{Econometrica}, 69\penalty0 (4):\penalty0 1077--1098, 2001.

\bibitem[Bimpikis et~al.(2019)Bimpikis, Crapis, and
  Tahbaz-Salehi]{bimpikis2019information}
Kostas Bimpikis, Davide Crapis, and Alireza Tahbaz-Salehi.
\newblock Information sale and competition.
\newblock \emph{Management Science}, 65\penalty0 (6):\penalty0 2646--2664,
  2019.

\bibitem[Bonatti and Cisternas(2020)]{bonatti2020consumer}
Alessandro Bonatti and Gonzalo Cisternas.
\newblock Consumer scores and price discrimination.
\newblock \emph{The Review of Economic Studies}, 87\penalty0 (2):\penalty0
  750--791, 2020.

\bibitem[Bonatti et~al.(2024)Bonatti, Dahleh, Horel, and
  Nouripour]{bonatti2024selling}
Alessandro Bonatti, Munther Dahleh, Thibaut Horel, and Amir Nouripour.
\newblock Selling information in competitive environments.
\newblock \emph{Journal of Economic Theory}, 216:\penalty0 105779, 2024.

\bibitem[Calzolari and Pavan(2006{\natexlab{a}})]{calzolari2006monopoly}
Giacomo Calzolari and Alessandro Pavan.
\newblock Monopoly with resale.
\newblock \emph{The RAND Journal of Economics}, 37\penalty0 (2):\penalty0
  362--375, 2006{\natexlab{a}}.

\bibitem[Calzolari and Pavan(2006{\natexlab{b}})]{calzolari2006optimality}
Giacomo Calzolari and Alessandro Pavan.
\newblock On the optimality of privacy in sequential contracting.
\newblock \emph{Journal of Economic Theory}, 130\penalty0 (1):\penalty0
  168--204, 2006{\natexlab{b}}.

\bibitem[Choi et~al.(2019)Choi, Jeon, and Kim]{choi2019privacy}
Jay~Pil Choi, Doh-Shin Jeon, and Byung-Cheol Kim.
\newblock Privacy and personal data collection with information externalities.
\newblock \emph{Journal of Public Economics}, 173:\penalty0 113--124, 2019.

\bibitem[Conitzer et~al.(2012)Conitzer, Taylor, and Wagman]{conitzer2012hide}
Vincent Conitzer, Curtis~R Taylor, and Liad Wagman.
\newblock Hide and seek: Costly consumer privacy in a market with repeat
  purchases.
\newblock \emph{Marketing Science}, 31\penalty0 (2):\penalty0 277--292, 2012.

\bibitem[Dosis and Sand-Zantman(2023)]{dosis2023ownership}
Anastasios Dosis and Wilfried Sand-Zantman.
\newblock The ownership of data.
\newblock \emph{The Journal of Law, Economics, and Organization}, 39\penalty0
  (3):\penalty0 615--641, 2023.

\bibitem[Doval and Skreta(2022)]{doval2022mechanism}
Laura Doval and Vasiliki Skreta.
\newblock Mechanism design with limited commitment.
\newblock \emph{Econometrica}, 90\penalty0 (4):\penalty0 1463--1500, 2022.

\bibitem[Doval and Skreta(2025)]{doval2025purchase}
Laura Doval and Vasiliki Skreta.
\newblock Purchase history and product personalization.
\newblock \emph{The RAND Journal of Economics}, 2025.

\bibitem[Dworczak(2020)]{dworczak2020mechanism}
Piotr Dworczak.
\newblock Mechanism design with aftermarkets: Cutoff mechanisms.
\newblock \emph{Econometrica}, 88\penalty0 (6):\penalty0 2629--2661, 2020.

\bibitem[Fainmesser et~al.(2023)Fainmesser, Galeotti, and
  Momot]{fainmesser2023digital}
Itay~P Fainmesser, Andrea Galeotti, and Ruslan Momot.
\newblock Digital privacy.
\newblock \emph{Management Science}, 69\penalty0 (6):\penalty0 3157--3173,
  2023.

\bibitem[Frankel and Kartik(2022)]{frankel2022improving}
Alex Frankel and Navin Kartik.
\newblock Improving information from manipulable data.
\newblock \emph{Journal of the European Economic Association}, 20\penalty0
  (1):\penalty0 79--115, 2022.

\bibitem[Fudenberg and Tirole(1991)]{fudenberg1991perfect}
Drew Fudenberg and Jean Tirole.
\newblock Perfect bayesian equilibrium and sequential equilibrium.
\newblock \emph{journal of Economic Theory}, 53\penalty0 (2):\penalty0
  236--260, 1991.

\bibitem[Fudenberg and Villas-Boas(2006)]{fudenberg2006behavior}
Drew Fudenberg and J~Miguel Villas-Boas.
\newblock Behavior-based price discrimination and customer recognition.
\newblock \emph{Handbook on economics and information systems}, 1:\penalty0
  377--436, 2006.

\bibitem[Goldfarb and Que(2023)]{goldfarb2023economics}
Avi Goldfarb and Verina~F Que.
\newblock The economics of digital privacy.
\newblock \emph{Annual Review of Economics}, 15:\penalty0 267--286, 2023.

\bibitem[Haghpanah and Siegel(2022)]{haghpanah2022limits}
Nima Haghpanah and Ron Siegel.
\newblock The limits of multiproduct price discrimination.
\newblock \emph{American Economic Review: Insights}, 4\penalty0 (4):\penalty0
  443--458, 2022.

\bibitem[Haghpanah and Siegel(2023)]{haghpanah2023pareto}
Nima Haghpanah and Ron Siegel.
\newblock Pareto-improving segmentation of multiproduct markets.
\newblock \emph{Journal of Political Economy}, 131\penalty0 (6):\penalty0
  1546--1575, 2023.

\bibitem[Hidir and Vellodi(2021)]{hidir2021privacy}
Sinem Hidir and Nikhil Vellodi.
\newblock Privacy, personalization, and price discrimination.
\newblock \emph{Journal of the European Economic Association}, 19\penalty0
  (2):\penalty0 1342--1363, 2021.

\bibitem[Ichihashi(2020)]{ichihashi2020online}
Shota Ichihashi.
\newblock Online privacy and information disclosure by consumers.
\newblock \emph{American Economic Review}, 110\penalty0 (2):\penalty0 569--95,
  2020.

\bibitem[Ichihashi(2021)]{ichihashi2021competing}
Shota Ichihashi.
\newblock Competing data intermediaries.
\newblock \emph{The RAND Journal of Economics}, 52\penalty0 (3):\penalty0
  515--537, 2021.

\bibitem[Jullien et~al.(2020)Jullien, Lefouili, and
  Riordan]{jullien2020privacy}
Bruno Jullien, Yassine Lefouili, and Michael~H Riordan.
\newblock Privacy protection, security, and consumer retention.
\newblock Technical report, Working paper, 2020.

\bibitem[Kamenica and Gentzkow(2011)]{kamenica2011bayesian}
Emir Kamenica and Matthew Gentzkow.
\newblock Bayesian persuasion.
\newblock \emph{American Economic Review}, 101\penalty0 (6):\penalty0
  2590--2615, 2011.

\bibitem[Kartik and Zhong(2023)]{kartik2023lemonade}
Navin Kartik and Weijie Zhong.
\newblock Lemonade from lemons: Information design and adverse selection.
\newblock \emph{arXiv preprint arXiv:2305.02994}, 2023.

\bibitem[Laffont and Tirole(1988)]{laffont1988dynamics}
Jean-Jacques Laffont and Jean Tirole.
\newblock The dynamics of incentive contracts.
\newblock \emph{Econometrica: Journal of the Econometric Society}, pages
  1153--1175, 1988.

\bibitem[Lizzeri(1999)]{lizzeri1999information}
Alessandro Lizzeri.
\newblock Information revelation and certification intermediaries.
\newblock \emph{The RAND Journal of Economics}, pages 214--231, 1999.

\bibitem[Marschak(1968)]{marschak1968economics}
Jacob Marschak.
\newblock Economics of inquiring, communicating, deciding.
\newblock \emph{The American Economic Review}, 58\penalty0 (2):\penalty0 1--18,
  1968.

\bibitem[Maskin and Riley(1984)]{maskin1984monopoly}
Eric Maskin and John Riley.
\newblock Monopoly with incomplete information.
\newblock \emph{The RAND Journal of Economics}, 15\penalty0 (2):\penalty0
  171--196, 1984.

\bibitem[Mussa and Rosen(1978)]{mussa1978monopoly}
Michael Mussa and Sherwin Rosen.
\newblock Monopoly and product quality.
\newblock \emph{Journal of Economic Theory}, 18\penalty0 (2):\penalty0
  301--317, 1978.

\bibitem[Myerson(1979)]{myerson1979incentive}
Roger~B Myerson.
\newblock Incentive compatibility and the bargaining problem.
\newblock \emph{Econometrica: Journal of the Econometric Society}, pages
  61--73, 1979.

\bibitem[Myerson(1981)]{myerson1981optimal}
Roger~B Myerson.
\newblock Optimal auction design.
\newblock \emph{Mathematics of Operations Research}, 6\penalty0 (1):\penalty0
  58--73, 1981.

\bibitem[Park and Evans(2025)]{park2025selling}
In-Uck Park and Robert Evans.
\newblock Selling information for bilateral trade.
\newblock \emph{Available at SSRN 5362376}, 2025.

\bibitem[Rhodes and Zhou(2021)]{rhodes2021personalized}
Andrew Rhodes and Jidong Zhou.
\newblock Personalized pricing and privacy choice.
\newblock Technical report, Working paper, 2021.

\bibitem[Roesler and Szentes(2017)]{roesler2017buyer}
Anne-Katrin Roesler and Bal{\'a}zs Szentes.
\newblock Buyer-optimal learning and monopoly pricing.
\newblock \emph{American Economic Review}, 107\penalty0 (7):\penalty0
  2072--2080, 2017.

\bibitem[Taylor(2004)]{taylor2004consumer}
Curtis~R Taylor.
\newblock Consumer privacy and the market for customer information.
\newblock \emph{RAND Journal of Economics}, pages 631--650, 2004.

\bibitem[Villas-Boas(2004)]{villas2004price}
J~Miguel Villas-Boas.
\newblock Price cycles in markets with customer recognition.
\newblock \emph{RAND Journal of Economics}, pages 486--501, 2004.

\bibitem[Yang(2022)]{yang2022selling}
Kai~Hao Yang.
\newblock Selling consumer data for profit: Optimal market-segmentation design
  and its consequences.
\newblock \emph{American Economic Review}, 112\penalty0 (4):\penalty0
  1364--1393, 2022.

\end{thebibliography}


\end{document}